\documentclass[a4paper, 12pt]{article}
\usepackage[a4paper]{geometry}
\newgeometry
{textheight=9in,
textwidth=6.5in,
top=1in,
headheight=14pt,
headsep=25pt,
footskip=30pt}

\usepackage{amsmath, amssymb, amsthm}
\usepackage{mathrsfs}
\usepackage[colorlinks=true,citecolor=blue, linkcolor=blue, urlcolor=blue, hypertexnames=false]{hyperref}
\usepackage{cleveref}
\usepackage{graphicx}
\usepackage{xifthen}
\usepackage{enumitem}
\usepackage{xcolor}
\usepackage{multirow}
\usepackage{tikz}
\usepackage{pgfplots}
\pgfplotsset{compat=1.18}

\usepackage{aliascnt}
\let\oldtheorem\newtheorem
\RenewDocumentCommand{\newtheorem}{s m o m O{}}{%
\IfBooleanTF{#1}%
{\oldtheorem{#2}{#4}}%
{\IfNoValueTF{#3}{\oldtheorem{#2}{#4}[#5]}%
{\newaliascnt{#2}{#3}%
\oldtheorem{#2}[#2]{#4}%
\aliascntresetthe{#2}}}}
\newtheorem{theorem}{Theorem}[section]

\newtheorem*{claim*}{Claim}
\newtheorem{example}[theorem]{Example}
\newtheorem{fact}[theorem]{Fact}
\newtheorem{lemma}[theorem]{Lemma}
\newtheorem{proposition}[theorem]{Proposition}
\newtheorem{corollary}[theorem]{Corollary}
\theoremstyle{definition}
\newtheorem{definition}[theorem]{Definition}
\newtheorem{remark}[theorem]{Remark}
\newtheorem*{remark*}{Remark}
\newtheorem{construction}[theorem]{Construction}
\newtheorem{condition}{Condition}

\def\*#1{\boldsymbol{#1}} %
\def\+#1{\mathcal{#1}} %
\def\-#1{\mathrm{#1}} %
\def\=#1{\mathbb{#1}} %

\DeclareMathOperator{\oPr}{\mathbf{Pr}}
\renewcommand{\Pr}[2][]
{\ifthenelse{\isempty{#1}}
  {\oPr\left[#2\right]}
  {\oPr_{#1}\left[#2\right]}
} %

\DeclareMathOperator{\oE}{\mathbb{E}}
\newcommand{\E}[2][]
{\ifthenelse{\isempty{#1}}
  {\oE\left[#2\right]}
  {\oE_{#1}\left[#2\right]}
} %

\DeclareMathOperator{\oVar}{\mathbf{Var}}
\newcommand{\Var}[2][]
{\ifthenelse{\isempty{#1}}
  {\oVar\left[#2\right]}
  {\oVar_{#1}\left[#2\right]}
}
\def\oEnt{\mathbf{Ent}}
\NewDocumentCommand{\Ent}{ O{} O{} m }{
  \ifthenelse{\isempty{#1}} {
    \ifthenelse{\isempty{#2}} {
      \oEnt\left[#3\right]
    } {
      \oEnt^{#2}\left[#3\right]
    }
  } {
    \ifthenelse{\isempty{#2}} {
      \oEnt_{#1}\left[#3\right]
    } {
      \oEnt_{#1}^{#2}\left[#3\right]
    }
  }
}

\newcommand{\tp}[1]{\left(#1\right)}
\newcommand{\set}[1]{\left\{#1\right\}}
\newcommand{\abs}[1]{\left\lvert#1\right\rvert}
\newcommand{\inner}[2]{\left\langle #1,#2\right\rangle}
\newcommand{\ol}[1]{\overline{#1}}

\newcommand{\pcor}[1]{\-{Cor}^{(2)}_{#1}}
\newcommand{\pSI}[1]{\lambda_{\max}\left(\pcor{#1}\right)}
\newcommand{\Tmix}{T_{\mathrm{mix}}}
\newcommand{\DTV}[2]{\mathrm{D}_{\mathrm{TV}}\tp{#1,#2}}
\newcommand{\Cov}{\mathrm{Cov}}
\newcommand{\diag}{\mathrm{diag}}

\newcommand{\DKL}[2]{\-D_{\-{KL}}\left(#1 \parallel #2\right)}
\newcommand{\dist}{\mathrm{dist}}
\newcommand{\e}{\mathrm{e}}
\renewcommand{\epsilon}{\varepsilon}
\renewcommand{\emptyset}{\varnothing}
\newcommand{\norm}[1]{\left\Vert#1\right\Vert}
\newcommand{\tuple}[1]{\left(#1\right)}
\newcommand{\eps}{\varepsilon}

\newcommand{\swcor}[1]{\mathrm{Cor}_{#1}^{\mathrm{SW}}}
\newcommand{\cor}[1]{\mathrm{Cor}_{#1}}

\newboolean{Ent}
\setboolean{Ent}{true} %

\title{Edge-Tilting Field Dynamics:\\[0.5em]{Rapid Mixing at the Uniqueness Threshold\\ and Optimal Mixing for Swendsen--Wang Dynamics}}

\def\AND{\hspace{2ex}}
\author{Xiaoyu Chen\thanks{Massachusetts Institute of Technology, Cambridge, MA, USA. Email: \texttt{xiaoyu@mit.edu}. Supported by the NSF CAREER grant CCF-2443045, and the Reed Fund at MIT.}%
  \AND Zhe Ju\thanks{State Key Laboratory for Novel Software Technology, New Cornerstone Science Laboratory, Nanjing University, China. Emails: \texttt{\{juzhe, miaotianshun\}@smail.nju.edu.cn, yinyt@nju.edu.cn}}
  \AND Tianshun Miao\footnotemark[2]
  \AND Yitong Yin\footnotemark[2]
  \AND Xinyuan Zhang\thanks{Georgia Institute of Technology, Atlanta, GA, USA. Email: \texttt{normalvra@gmail.com}}}

\date{}

\begin{document}

\maketitle

\begin{abstract}
    We prove two results on the mixing times of Markov chains for two-spin systems. 
    First, we show that the Glauber dynamics mixes in polynomial time for the Gibbs distributions of antiferromagnetic two-spin systems at the critical threshold of the uniqueness phase transition of the Gibbs measure on infinite regular trees.
    This result complements the near-linear-time optimal mixing in the uniqueness regime~[Chen--Liu--Vigoda, STOC '21; Chen--Feng--Yin--Zhang, FOCS '22] and the NP-hardness of approximate sampling in the non-uniqueness regime~[Sly--Sun, FOCS '12].
    Along the way, we prove a tight spectral-independence bound for antiferromagnetic two-spin systems by introducing a new control-function approach to total influence.

    Second, we prove an optimal $O(\log n)$ mixing-time bound and an optimal $\Omega(1)$ spectral-gap bound for the Swendsen--Wang dynamics for the ferromagnetic Ising model with external field on bounded-degree graphs. 
    These bounds improve the previous general $O(n\log n)$ mixing-time bound and resolve a conjecture of~[Feng--Guo--Wang, IANDC '23].

    A key ingredient in both proofs is a new family of localization schemes that extends the field dynamics of~[Chen--Feng--Yin--Zhang, FOCS '21] by tilting general edge (or hyperedge) weights rather than vertex fields. 
    This framework, which subsumes the classical Swendsen--Wang dynamics as a special case, extends the localization framework of~[Chen--Eldan, FOCS '22] beyond stochastic and field localizations, and enables controlled tilting of interaction strengths while preserving external fields.
\end{abstract}

\thispagestyle{empty}

\newpage

\tableofcontents
\thispagestyle{empty}
\pagebreak

\setcounter{page}{1}

\section{Introduction}

A central phenomenon in the study of sampling algorithms is the computational phase transition of antiferromagnetic two-spin systems at the uniqueness threshold.
Let $G = (V,E)$ be an undirected graph. 
A two-spin system on $G$ is specified by parameters $(\beta, \gamma, \lambda)$, where $\beta, \gamma \geq 0$ corresponds to the \emph{edge activities} and $\lambda > 0$ corresponds to the \emph{external field}.
The associated Gibbs distribution $\mu$ is defined by
\begin{align*}
\mu(\sigma) \propto \beta^{m_1(\sigma)}\gamma^{m_0(\sigma)}\lambda^{\|\sigma\|}, \quad \forall\sigma\in\set{0,1}^V,
\end{align*}
where $m_{b}(\sigma) \triangleq \abs{\{(u,v) \in E : \sigma_u = \sigma_v = b\}}$ denotes the number of monochromatic edges with spin $b\in\{0,1\}$, and $\|\sigma\|$ denotes the number of vertices with spin $1$ in $\sigma$.

A two-spin system is \emph{ferromagnetic} if $\beta\gamma > 1$ and \emph{antiferromagnetic} if $\beta\gamma < 1$.
Up to normalization, any antiferromagnetic two-spin system can be specified by $(\beta,\gamma,\lambda)$ satisfying
\begin{align} \label{eq:cond-anti-ferro}
    0 \leq \beta \leq \gamma, \quad \gamma>0, \quad \lambda>0, \quad \text{and} \quad \beta\gamma < 1.
\end{align}

The following \emph{uniqueness condition} characterizes the uniqueness of the Gibbs measure on the infinite $(d+1)$-regular tree.

\begin{definition}[Uniqueness condition] \label{def:d-unique}
Let $\delta \in [0,1]$ and let $d\ge 1$ be an integer. 
Let $(\beta,\gamma,\lambda)$ be parameters satisfying \eqref{eq:cond-anti-ferro}.
We say that $(\beta,\gamma,\lambda)$ is \emph{$d$-unique with slack $\delta$} if
\begin{align*}
\abs{f'(\hat{x}_d)} = \frac{d(1-\beta\gamma)\hat{x}_d}{(\beta \hat{x}_d+1)(\hat{x}_d+\gamma)} \leq 1-\delta,
\end{align*}
where
\begin{align*}
f_d(x) :=\lambda\left(\frac{\beta x+1}{x+\gamma}\right)^d,
\end{align*}
and  $\hat{x}_d$ is the unique positive fixed point of $f_d$, that is, $\hat{x}_d = f(\hat{x}_d)$.

We say that $(\beta,\gamma,\lambda)$ is \emph{critically $d$-unique} if
\begin{align*}
     \abs{f'(\hat{x}_d)} = 1.
\end{align*}
\end{definition}

Remarkably, the same uniqueness/non-uniqueness phase transition also captures the computational complexity of sampling.
In particular, 
when $(\beta,\gamma,\lambda)$ lies in the non-uniqueness regime (i.e., $\abs{f'(\hat{x}_d)}>1$), 
it was shown in~\cite{sly2012computational,galanis2016inapproximability} that approximate sampling from the Gibbs distribution of the antiferromagnetic two-spin system with parameters $(\beta,\gamma,\lambda)$ on $(d+1)$-regular graphs, for any $d \ge 2$, is computationally intractable unless $\-{NP} = \-{RP}$.
In contrast, when $(\beta,\gamma,\lambda)$ is $d$-unique with positive slack for all $d < \Delta$ (i.e., \emph{up-to-$\Delta$-unique}), Gibbs sampling is tractable in polynomial time on graphs with constant max-degree $\Delta$~\cite{LLY13}.

The Glauber dynamics is the canonical Markov chain for sampling from the Gibbs distribution $\mu$.
Given a current state $X$ in the support $\Omega(\mu)$ of $\mu$, the chain updates as follows:
\begin{itemize}
\item choose a vertex $v \in V$ uniformly at random;
\item resample $X_v$ according to the conditional distribution $\mu_v(\cdot \mid X(V \setminus {v}))$.
\end{itemize}
It is well known that the Glauber dynamics has stationary distribution $\mu$.
Let $(X_t)_{t \geq 0}$ denote the trajectory of the chain.
The \emph{mixing time} of the chain is defined as:
\begin{align*}
    \Tmix := \max_{x \in \Omega(\mu)}\min\set{t \mid \DTV{\Pr{X_t = \cdot \mid X_0 = x}}{\mu} \leq 1/4},
\end{align*}
where, for two distributions $\nu$ and $\mu$ on a finite set $\Omega$, the \emph{total variation distance} is defined by
    $\DTV{\nu}{\mu} := \frac{1}{2} \sum_{\sigma \in \Omega} \abs{\nu(\sigma) - \mu(\sigma)}$.

In recent years, a new class of techniques based on high-dimensional expanders (HDX) and spectral independence (SI) has led to major advances in the analysis of mixing times \cite{anari2019logconcave, alev2020improved, anari2020spectral, chen2021optimal}. 
For antiferromagnetic two-spin systems, rapid mixing of the Glauber dynamics up to the uniqueness threshold has been established in a series of works \cite{anari2020spectral, chen2020rapid, chen2021optimal, chen2021rapid, anari2021entropic, blanca2022mixing, anari2022entropic, chen2022localization, chen22optimal},
which prove polynomial and in some cases near-linear bounds on the mixing time. 
In contrast, it has long been known \cite{mossel2009hardness} that beyond the uniqueness threshold, the Glauber dynamics requires exponential time to mix in the worst case.

Despite substantial progress on tractability in the uniqueness regime and hardness in the non-uniqueness regime, the critical case has remained unresolved.
Existing results have been largely restricted to specific graph families \cite{levin2010glauber,ding2009mixing,ding2010mixing,lubetzky2012critical,bauerschmidt2024log,prodromidis2025polynomial}, leaving the mixing behavior on general graphs unclear.

Recently, \cite{chen2025uniqueness} established polynomial bounds on the mixing time of the Glauber dynamics for the critical zero-field Ising model and the critical hardcore model by exploiting spectral and entropic stability within the framework of \emph{localization schemes}  introduced in \cite{chen2022localization}, in particular \emph{stochastic} localization \cite{eldan2013thin} and  \emph{negative-field} localization (also known as \emph{field dynamics}) \cite{chen22optimal}.

In this paper, we study the mixing time of the Glauber dynamics at the uniqueness threshold for general antiferromagnetic two-spin systems. 
In particular, we consider critical instances of antiferromagnetic two-spin systems, formalized by the following condition.

\begin{condition}[Critical instance] \label{cond:critical}
We call an instance of an antiferromagnetic two-spin system with parameters $(\beta,\gamma,\lambda)$ \emph{critical} if it is defined on a graph $G$ of maximum degree $\Delta \ge 3$ and satisfies:
\begin{itemize}
    \item if $\gamma< 1$, then $(\beta,\gamma,\lambda)$ is critically $(\Delta-1)$-unique;
    \item if $\gamma\ge 1$, then $(\beta,\gamma,\lambda)$ is critically $(\Delta-1)$-unique, and $G$ is $\Delta$-regular.
\end{itemize}
\end{condition}

\begin{remark}[Criticality and regularity]
On $\Delta$-regular graphs, the notion of criticality is natural: $(\beta,\gamma,\lambda)$ is critically $(\Delta-1)$-unique. 

For general graphs, defining criticality is more subtle. 
In particular, when $\gamma > 1$, the uniqueness property is not monotone in the degree, so the worst-case behavior may not occur at the maximum degree $\Delta$. 
To avoid this complication, \Cref{cond:critical} restricts to $\Delta$-regular graphs in this regime.  
In contrast, for strictly antiferromagnetic $2$-spin systems with $\gamma \le 1$, the uniqueness property is monotone in the degree, so criticality at degree $\Delta$ automatically implies uniqueness at all smaller degrees (i.e., $(\beta,\gamma,\lambda)$ is up-to-$\Delta$-unique), and our analysis naturally extends to general graphs with maximum degree $\Delta$. Moreover, a gap between $\gamma$ and $1$ is needed in the analysis for irregular graphs. Hence, the case $\gamma=1$ is moved to the regular branch.
\end{remark}

The subcritical analogue of \Cref{cond:critical} (i.e., uniqueness holds with slack) was adopted as the uniqueness regime in \cite{chen22optimal}; 
this represents the state of the art for $O(n\log n)$  mixing time for general antiferromagnetic two-spin systems with {unbounded} maximum degree.

For antiferromagnetic two-spin systems at the critical threshold specified by \Cref{cond:critical}, we establish the following upper bound on the mixing time of the Glauber dynamics.

\begin{theorem}[Main theorem]
\label{thm:critical-mixing}
Let $\mu$ be the Gibbs distribution of an antiferromagnetic two-spin system with parameters $(\beta,\gamma,\lambda)$ on an $n$-vertex graph of maximum degree $\Delta \ge 3$ satisfying \Cref{cond:critical}.
Then the Glauber dynamics for $\mu$ has mixing time
\begin{align*}
\Tmix = O_{\beta,\gamma,\lambda}\tuple{\Delta n^{2\sqrt{2}+4+O(1/\Delta)}}.
\end{align*}
\end{theorem}
\begin{remark}
\Cref{thm:critical-mixing} gives, to the best of our knowledge, the first polynomial mixing-time bound for general antiferromagnetic two-spin systems at criticality. It applies beyond the regimes covered by previous methods. In particular, it applies to the critical antiferromagnetic Ising model with external field.

Prior polynomial bounds at criticality were known only for special cases: for the critical Ising model with zero external field, a $\widetilde{O}\!\left(n^{3+O(1/\Delta)}\right)$ bound was established in \cite{chen2025uniqueness}; and for the critical hard-core model, a $\widetilde{O}\!\left(n^{4+O(1/\Delta)}\right)$ bound was established in \cite{chen2025improved,chen2025uniqueness}. 
In contrast, \Cref{thm:critical-mixing} applies to general antiferromagnetic two-spin systems,  substantially broadening the scope of polynomial-time mixing at criticality.
\end{remark}

Complementing the near-linear mixing time bounds for the Glauber dynamics in the subcritical regime corresponding to \Cref{cond:critical}, established in \cite{chen2021optimal} for bounded-degree graphs and in \cite{chen22optimal} for general graphs, 
as well as the hardness of sampling in the non-uniqueness regime shown in~\cite{sly2012computational, galanis2016inapproximability}, \Cref{thm:critical-mixing} fills the critical case for instances satisfying \Cref{cond:critical}, corresponding to the subcritical uniqueness regime considered in~\cite{chen22optimal}:
\[
\genfrac{}{}{0pt}{}{\text{near-linear time}}{\text{(subcritical)}}
      \,\,\longrightarrow \,\,
\genfrac{}{}{0pt}{}{\text{polynomial time}}{\text{(critical)}}
      \,\,\longrightarrow \,\,
\genfrac{}{}{0pt}{}{\text{NP-hardness}}{\text{(supercritical)}}.
\]

\subsection{Edge-tilting field dynamics}\label{sec:intro-edge-field-dynamics}
A key step in proving \Cref{thm:critical-mixing} is the construction of an \emph{edge-tilting} variant of the \emph{field dynamics}. 
This yields a new family of localization schemes that achieves {controlled interaction} localization.
Compared to the original field dynamics introduced in \cite{chen2021rapid}, which reduces critical instances to subcritical ones with smaller external fields by tilting vertex fields, the new scheme instead tilts edge strengths toward the subcritical (high-temperature) regime. 
Previously, such edge tilting was achieved via stochastic localization, which induces uncontrolled drifts in the external fields; in contrast, the new scheme performs edge tilting in a controlled manner, without introducing such drift.

This ability to tilt interaction strengths in a controlled (rather than stochastic) manner fills a gap in the current localization scheme framework and is crucial for resolving the critical case for general antiferromagnetic two-spin systems. 

Given an arbitrary distribution $\mu$ over $\set{0, 1}^V$ and a graph $G=(V,E)$, 
we write $\theta \otimes \mu$ for the distribution obtained by tilting the interactions by a factor of $\theta>0$. 
Specifically,
\begin{align}
\label{eq:def-otimes}
\theta\otimes\mu(\sigma) \propto \mu(\sigma)\cdot\theta^{m(\sigma)}, \qquad \forall\sigma\in\set{0, 1}^V,
\end{align}
where $m(\sigma) := \abs{\set{\set{u,v} \in E \mid \sigma_u = \sigma_v}}$ denotes the number of monochromatic edges in~$\sigma$. In particular, suppose $\mu=\mu_{\beta,\gamma,\lambda}$ is the Gibbs distribution of a two-spin system. 
Then, $\theta\otimes \mu=\mu_{\theta\beta,\theta\gamma,\lambda}$ gives the new Gibbs distribution with tilted edge activities $\theta\beta,\theta\gamma$ and with the external field $\lambda$ unchanged.

We now introduce a Markov chain called the \emph{edge-tilting field dynamics}, or simply the \emph{edge-field dynamics}.
Given a distribution $\mu$ over $\set{0,1}^V$, a graph $G=(V,E)$, and a parameter $\theta\in(0,1)$, the \emph{edge-field dynamics  for $\mu$ with tilt parameter $\theta$} is a Markov chain on the state space $\Omega(\mu)$. 
From a current configuration $\sigma\in\Omega(\mu)$, denote by $N(\sigma) := \set{\set{u,v} \in E \mid \sigma_u \neq \sigma_v}$ the set of non-monochromatic edges in $\sigma$. The chain performs the following update:
\begin{enumerate}[leftmargin=2cm]
\item[(\textsc{Down})] %
Independently remove each edge in $N(\sigma)$  with probability $\theta$. 
Let $T\subseteq V$ be the set of  vertices incident to the remaining edges.
\item[(\textsc{Up})] Resample a new configuration from the conditional distribution $\tp{\frac{1}{\theta}\otimes\mu}^{\sigma_{T}}$.%
\end{enumerate}
The process is an edge-tilting variant of the \emph{vertex-field dynamics} introduced in \cite{chen2021rapid}, where it was simply called the \emph{field dynamics}.
In the vertex-field dynamics with tilt parameter $\theta\in(0,1)$, at each step the subset $T\subseteq V$ is constructed by independently removing each vertex $v$ from the set $\set{v \in V \mid \sigma_v = 1}$ with probability $\theta$.
The new configuration is then sampled from the conditional distribution $\theta * \mu^{\sigma_{T}}$, where
\begin{align}\label{eq:def-field-tilt}
\theta*\mu(\sigma)\propto \mu(\sigma)\cdot\theta^{\|\sigma\|}, \qquad\forall \sigma\in\set{0,1}^V,
\end{align}
i.e., the distribution obtained from $\mu$ by tilting the external fields by a factor of $\theta$.

\begin{remark}[A generalized event-based framework for field dynamics] \label{rem:event-FD}
More generally, 
consider a family of events $\mathcal{A}$, where each $A \in \mathcal{A}$ is a subset $A \subseteq \set{0,1}^V$, 
and let $\boldsymbol{\theta}=(\theta_A)_{A\in\mathcal{A}}\in[0,1]^{\mathcal{A}}$ be the associated tilts.

The \emph{event-field dynamics} for a distribution $\mu$ over $\set{0,1}^V$, with respect to the event family $\mathcal{A}$ and tilts $\boldsymbol{\theta}$, is a Markov chain on $\Omega(\mu)$.
From a current configuration $\sigma \in \Omega(\mu)$, 
let $\mathcal{A}(\sigma) := \{A \in \mathcal{A} \mid \sigma \in A\}$ be the set of events that occur under $\sigma$. 
The chain updates as:
\begin{enumerate}[leftmargin=1.8cm]
\item[(\textsc{Down})] 
Generate a random subset $T\subseteq \mathcal{A}(\sigma)$  by independently removing each event $A \in \mathcal{A}(\sigma)$ with probability $\theta_A$.

\item[(\textsc{Up})] Resample a new configuration   from the conditional distribution $\nu^T$, where
\[
 \nu(\sigma) \propto \mu(\sigma)\cdot
 \prod_{A\in\mathcal{A}} \theta_A^{\mathbf{1}[\sigma\in A]}, \qquad \forall\sigma\in\set{0, 1}^V,
\]
and $\nu^T$ denotes $\nu$ conditioned on the occurrence of all events in $T$. 
\end{enumerate}
This generalized framework subsumes several Markov chains as special cases, including the vertex-field dynamics, the edge-field dynamics, and the classical Swendsen--Wang dynamics:
\begin{itemize}
    \item \emph{Vertex-field dynamics:} each $A\in\mathcal{A}$ corresponds to a vertex $v\in V$, and occurs if $\sigma_v=1$.
    \item \emph{Edge-field dynamics:} each $A\in\mathcal{A}$ corresponds to an oriented edge $(u,v)$ with $\set{u,v} \in E$, and occurs if $\sigma_u = 1$ and $\sigma_v = 0$.
    \item \emph{Swendsen--Wang dynamics:} 
    each $A \in \mathcal{A}$ corresponds to an edge $\{u,v\} \in E$, 
    occurs if $\sigma_u = \sigma_v$, 
    and is associated with tilt $\theta_A = \beta^{-1}$. 
On the ferromagnetic Ising model (i.e., two-spin systems with $\beta = \gamma > 1$), this recovers the classical Swendsen--Wang chain.
\end{itemize}

As we shall see in \Cref{sec:edge-field-dynamics}, this general class of field dynamics falls into the framework of localization schemes of \cite{chen2022localization}, and is guaranteed to have $\mu$ as its stationary distribution.
\end{remark}

The following theorem states a mixing time bound for the edge-field dynamics for critical antiferromagnetic two-spin systems with soft constraints ($\beta>0$).
\begin{theorem}[Mixing of edge-field dynamics at criticality]\label{thm:mixing-edge-field-dynamics}
Let $\mu$ be the Gibbs distribution of an antiferromagnetic two-spin system with parameters $(\beta,\gamma,\lambda)$ on an $n$-vertex graph of maximum degree $\Delta \ge 3$ satisfying \Cref{cond:critical}, 
and suppose that $\beta>0$. 
Then the edge-field dynamics for $\mu$ with tilt parameter $\theta=\sqrt{\beta\gamma}$ has mixing time
\begin{align*}
\Tmix \le \Delta(\beta\gamma)^{-\Delta}\cdot (C_{\beta,\gamma,\lambda}\cdot  n)^{c+1}, 
\end{align*}
where $c=\frac{2\Delta^2(1-\sqrt{\beta\gamma})}{(\Delta-1)\sqrt{\beta\gamma}}$ and $C_{\beta,\gamma,\lambda}$ is a constant that depends only on $\beta,\gamma,\lambda$. %
\end{theorem}

\begin{remark}[Mixing of vertex-field dynamics at criticality]\label{remark:vertex-field-dynamics}
Note that the exponent $c$ in \Cref{thm:mixing-edge-field-dynamics} is small when $\beta\gamma$ is close to $1$.
In particular, when $\sqrt{\beta\gamma}\ge 1-\Theta(\frac{1}{\Delta})$, we have $c=O(1)$,
in which case the edge-field dynamics has mixing time  $O_{\beta,\gamma,\lambda}(\Delta n^{O(1)})$.
On the other hand, when $\sqrt{\beta\gamma}\le 1-\Theta(\frac{1}{\Delta})$, 
the spin system behaves more like the hardcore model, where the analysis for the critical hardcore model in \cite{chen2025uniqueness} applies,
showing that in this regime the \emph{vertex}-field dynamics mixes in polynomial time at criticality.
\end{remark}

The rapidly mixing edge- and vertex-field dynamics in \Cref{thm:mixing-edge-field-dynamics} and \Cref{remark:vertex-field-dynamics} provide polynomial-time samplers for critical antiferromagnetic two-spin systems satisfying \Cref{cond:critical}.
To implement these Markov chains, each transition of the chain involves sampling from a subcritical Gibbs distribution with tilted parameters, which can be simulated via Glauber dynamics that is optimally mixing in the subcritical regime due to \cite{chen22optimal}.

To establish the rapid mixing of the Glauber dynamics in \Cref{thm:critical-mixing}, we interpret the edge-tilting field dynamics within the localization framework of \cite{chen2022localization}; this is developed in \Cref{sec:edge-field-dynamics}.
\Cref{thm:critical-mixing} then follows from the stability and conservation of mixing properties under this framework, as shown in \Cref{sec:rapid-mixing-at-criticality}.

\subsection{Tight SI bound for antiferromagnetic two-spin systems}

As part of our analysis of mixing at criticality for antiferromagnetic two-spin systems, we establish a tight bound on \emph{spectral independence} ({SI}), a notion introduced by Anari, Liu, and Oveis Gharan \cite{anari2020spectral} 
that has led to recent breakthroughs in the study of Markov chain mixing times and also plays a key role in our analysis at criticality.

For a distribution $\mu$ over $\set{0,1}^{V}$, let $X\sim\mu$. 
Define the \emph{influence matrix} $\Psi_{\mu}\in\mathbb{R}^{V\times V}$ by
\begin{align*}
\Psi_{\mu}(u,v):=\begin{cases}
\Pr{X_v=1\mid X_u=1}-\Pr{X_v=1\mid X_u=0} &\quad \text{if }\E{X_u}\in(0,1)\text{ and }u\ne v,\\
0 &\quad \text{otherwise}.
\end{cases}
\end{align*}
A distribution $\mu$ is said to be $\rho$-\emph{spectrally independent} if, for any {feasible pinning} $\tau$ (i.e., $\tau \in \{0,1\}^{\Lambda}$ for some $\Lambda \subseteq V$ such that $\mu_\Lambda(\tau)>0$), 
it holds that $\lambda_{\max}(\Psi_{\mu^\tau})\le \rho$, for the largest eigenvalue $\lambda_{\max}(\Psi_{\mu^\tau})$ of $\Psi_{\mu^\tau}$.

The following theorem establishes a SI bound for antiferromagnetic two-spin systems.

\begin{theorem}\label{thm:optimal-SI}
Let $\delta\in(0,1)$. 
Let $\mu$ be the Gibbs distribution of an antiferromagnetic two-spin system with parameters $(\beta,\gamma,\lambda)$ on a graph $G$ of maximum degree $\Delta \ge 3$, satisfying the $\delta$-slack analogue of \Cref{cond:critical}, that is:
\begin{itemize}
\item if $\gamma< 1$, then $(\beta,\gamma,\lambda)$ is $(\Delta-1)$-unique with slack $\delta$;
\item if $\gamma\ge 1$, then $(\beta,\gamma,\lambda)$ is $(\Delta-1)$-unique with slack $\delta$, and $G$ is $\Delta$-regular.
\end{itemize}
Then, for any feasible pinning $\tau$ for $\mu$, it holds that
\begin{align*}
\lambda_{\max}\bigl(\Psi_{\mu^{\tau}}\bigr) \le \frac{\Delta(1-\delta)}{(\Delta-1)\delta}.
\end{align*}
\end{theorem}

This SI bound is essentially tight, in the sense that the constant factor cannot be further improved. 
To see this, the value $\frac{\Delta(1-\delta)}{(\Delta-1)\delta}$ is attained (in the limit) by $\|\Psi_{\pi}\|_{\infty}$, 
where $\pi$ is a Gibbs distribution with the same slack $\delta$ defined on the infinite $\Delta$-regular tree $\mathbb{T}_{\Delta}$. 
Although $\mathbb{T}_{\Delta}$ is infinite, 
one can construct a sequence of distributions $\{\mu_n\}$ on finite graphs that locally converge to $\pi$, 
and verify that $\lambda_{\max}(\Psi_{\mu_n})$ approaches $\frac{\Delta(1-\delta)}{(\Delta-1)\delta}$ in the limit. 
See \Cref{sec:new-CI} for the proof of \Cref{thm:optimal-SI} and a matching lower bound stated in \Cref{thm:SI-lowerbound}.

We emphasize that the tightness of the SI bound, even at the level of the constant factor, is critical in the analysis of mixing rates at criticality, as it directly affects the exponent in the resulting polynomial mixing time.

\subsection{Optimal mixing of the Swendsen--Wang dynamics}
For the ferromagnetic Ising model (i.e., two-spin systems with edge activities $\beta = \gamma > 1$), 
the \emph{Swendsen--Wang dynamics}, introduced in~\cite{Swendsen1987}, defines a global Markov chain for sampling from the Gibbs distribution.
Despite its strong empirical performance, a tight analysis of its mixing time remains largely open.

The question of proving $O(\log n)$ mixing time for the Swendsen--Wang dynamics has been viewed as an important benchmark to bypass bottlenecks of local Gibbs updates. 
For example, \cite{pmlr-v54-park17b} established $O(\log n)$ mixing time on stochastic partitioned graphs in certain parameter regimes with no external fields; and \cite{blanca2022mixing} established $O(\log n)$ mixing time in the tree-uniqueness regime.
In particular, it is conjectured in~\cite{feng2023swendsen-wang} that the Swendsen--Wang dynamics mixes in $O(\log n)$ steps for the ferromagnetic Ising model when the external field $\lambda < 1$ (and, symmetrically, when $\lambda > 1$).

We resolve this conjecture for bounded-degree graphs.
As observed in \Cref{rem:event-FD}, the Swendsen--Wang dynamics can be interpreted as a variant of the edge-field dynamics. 
This allows us to leverage the analysis developed for edge-field dynamics to study its mixing behavior. 
In particular, we prove the following bounds on its spectral gap and mixing time. 

\begin{theorem}\label{thm:SW-main}
    Fix sufficiently small $\delta \in (0, 1)$. 
    Let $G = (V, E)$ be a graph on $n$ vertices with maximum degree~$\Delta$.
    For the Swendsen--Wang dynamics for the ferromagnetic Ising model on $G$ with edge activity $\beta \in [1,\infty)$ and external field $\lambda \in [0, 1 - \delta]$, 
    the spectral gap is at least $\Omega\bigl((\beta \Delta)^{-2/\delta^2}\bigr) = \Omega_{\beta,\delta,\Delta}(1)$, 
    and the mixing time is at most $O((\beta \Delta)^{4/\delta^2} \log n)=O_{\beta,\delta,\Delta}(\log n)$.
\end{theorem}
For arbitrary graphs of bounded maximum degree $\Delta$, all $\beta\ge1$, and external fields $0\le\lambda\le1-\delta$ with fixed $\delta\in(0,1)$, \Cref{thm:SW-main} gives an optimal $\Omega_{\beta,\delta,\Delta}(1)$ lower bound on the spectral gap and an optimal $O_{\beta,\delta,\Delta}(\log n)$ upper bound on the mixing time of the Swendsen--Wang dynamics for the ferromagnetic Ising model.
For optimality, we defer to~\Cref{sec:C}, which gives an $\Omega(\log n)$ lower bound in the worst case over bounded-degree graphs.
These bounds improve the previous general $O(n\log n)$ mixing-time bound~\cite{feng2023swendsen-wang} and resolve the conjecture of Feng, Guo, and Wang~\cite{feng2023swendsen-wang} for bounded-degree graphs.
See \Cref{sec:rapid-mixing-SW-dynamics} for the proof of \Cref{thm:SW-main}.

\begin{remark}[An ``edge-to-edge'' SI]
Notably, spectral independence (SI) does not hold for the ferromagnetic Ising model beyond the uniqueness threshold, and the Glauber dynamics in the ``spin world'' mixes torpidly in this regime~\cite{gerschenfeld2007reconstruction}.

In contrast, we show that the Swendsen--Wang dynamics mixes at an optimal rate over the entire ferromagnetic regime in the presence of an external field. 
A key step in our proof is to establish a notion of SI over edges, despite the absence of SI in the spin world.
This edge-to-edge SI supports rapid mixing of the corresponding edge-tilting field dynamics, which in our case is the Swendsen--Wang dynamics.
\end{remark}

\subsection{Related work}
\paragraph{Computational phase transition of antiferromagnetic two-spin systems.}
The computational phase transition of antiferromagnetic two-spin systems has attracted considerable attention~\cite{dyer2002counting,goldberg2003computational,weitz2006counting,sly2010computational,li2012approximate,sly2012computational,LLY13,sinclair2014approximation,galanis2015inapproximability,galanis2016inapproximability}.
In a seminal work, \cite{weitz2006counting} developed a correlation-decay-based algorithm via the self-avoiding walk (SAW) tree, proving polynomial-time tractability of the hardcore model up to the uniqueness threshold. 
This was subsequently extended to all antiferromagnetic two-spin systems in the uniqueness regime~\cite{li2012approximate,LLY13,sinclair2014approximation}.
In the non-uniqueness regime, it was shown~\cite{sly2010computational,sly2012computational,galanis2015inapproximability,galanis2016inapproximability} that the problem of Gibbs sampling is intractable in polynomial time unless $\mathrm{NP} = \mathrm{RP}$.

\paragraph{Rapid mixing of Glauber dynamics up to criticality.}
The notion of spectral independence was introduced in \cite{anari2020spectral}, where techniques based on high-dimensional expanders (HDX), developed in \cite{anari2019logconcave,alev2020improved}, were applied to analyze the mixing time of Glauber dynamics. This line of work led to a series of major advances establishing rapid mixing of Glauber dynamics up to the uniqueness threshold, including~\cite{chen2020rapid,chen2021optimal,chen2021rapid,anari2021entropic,blanca2022mixing,anari2022entropic,chen2022localization,chen22optimal}.
In particular, \cite{chen22optimal} established an optimal near-linear bound on the mixing time of Glauber dynamics for general antiferromagnetic two-spin systems on graphs with potentially unbounded maximum degree, throughout the  regime below the critical threshold specified in \Cref{cond:critical}.

\paragraph{Rapid mixing at criticality.}
For systems at the critical threshold, prior work has focused on the Ising model on specific graph families, including complete graphs \cite{levin2010glauber,ding2009mixing}, trees \cite{ding2010mixing}, the grid $\mathbb{Z}^2$ \cite{lubetzky2012critical}, higher-dimensional lattices ($\mathbb{Z}^d$ with $d>4$) \cite{bauerschmidt2024log}, and sparse random graphs \cite{prodromidis2025polynomial}. However, they do not extend directly to general graphs.

Recently, \cite{chen2025uniqueness} proved polynomial mixing time bounds for the Glauber dynamics of the hardcore model and the zero-field Ising model at the critical uniqueness threshold, 
by applying the localization scheme framework introduced in \cite{chen2022localization},
which, in particular, unifies stochastic localization \cite{eldan2013thin} and the negative-field localization (field dynamics) introduced in \cite{chen2021rapid}. 
The result for zero-field Ising model was also established concurrently in \cite{prodromidis2024polynomial}.
Furthermore, in follow-up work, \cite{chen2025improved} improved the mixing time bound for the critical hardcore model.

\paragraph{Mixing times of Swendsen--Wang dynamics.}
Theoretical analysis of the mixing time of the Swendsen--Wang dynamics has been a major challenge since its introduction~\cite{Swendsen1987}.
For the zero-field Ising model on complete graphs (also known as the mean-field model), the chain mixes in $O(n^{1/4})$ steps, and this bound is tight at criticality \cite{long2014SW}. It has been conjectured that this $O(n^{1/4})$ mixing-time bound extends to general graphs.
For the zero-field Ising model, sharp mixing-time bounds have recently been obtained on $\mathbb{Z}^d$~\cite{blanca2022entropy} and within the uniqueness regime~\cite{blanca2022mixing}.
For the ferromagnetic Ising model on general graphs, polynomial mixing time bounds for the Swendsen--Wang dynamics were established only recently~\cite{guo2018random,feng2023swendsen-wang}.
In particular, on bounded-degree graphs with an external field, an $O(n \log n)$ mixing time bound was shown in~\cite{feng2023swendsen-wang} for the bounded-degree graphs via a comparison with Glauber dynamics on random clusters.

\section{Preliminaries}
\subsection{\texorpdfstring{$\phi$}{phi}-divergence and \texorpdfstring{$\phi$}{phi}-entropy}
Let $\phi : D \to \mathbb{R}$ be a convex function, where $D \subseteq \mathbb{R}$ is its domain.
Let $\mu$ be a distribution on a finite set $\Omega$, and let $X \sim \mu$.
For a real-valued random variable $f : \Omega \to D$, we define the \emph{$\phi$-entropy} of $f$ with respect to $\mu$ as
\begin{align*}
    \Ent[\mu][\phi]{f} := \Ent[][\phi]{f(X)} := \E{\phi(f(X))} - \phi(\E{f(X)}).
\end{align*}
We also define the \emph{conditional $\phi$-entropy} with respect to a random variable $Y$ as
\begin{align*}
    \Ent[][\phi]{f(X) \mid Y} := \E{\phi(f(X)) \mid Y} - \phi(\E{f(X) \mid Y}).
\end{align*}
Since $\phi$ is convex, by Jensen's inequality $\Ent[\mu][\phi]{f} \ge 0$.

For two distributions $\nu$ and $\mu$ on $\Omega$, we say that $\nu$ is \emph{absolutely continuous} with respect to $\mu$ (denoted $\nu \ll \mu$) if $\mathrm{supp}(\nu) \subseteq \mathrm{supp}(\mu)$.
In this case, $f = \nu/\mu$ is well-defined, and the \emph{$\phi$-divergence} between $\nu$ and $\mu$ is defined as
\begin{align*}
    \mathrm{D}_{\phi}(\nu \parallel \mu) := \Ent[\mu][\phi]{\frac{\nu}{\mu}}.
\end{align*}

In particular, we consider the following choices of $\phi$:
\begin{itemize}
    \item $\mathrm{TV}(x) := \tfrac{1}{2}\lvert x - 1\rvert$, which induces the total variation distance
    \[\DTV{\nu}{\mu} = \frac{1}{2} \sum_{\sigma \in\Omega} \abs{\nu(\sigma) - \mu(\sigma)}; \]
    \item $\chi^2(x) := x^2$, which induces the $\chi^2$-divergence
    \[\mathrm{D}_{\chi^2}(\nu\parallel \mu) = \tp{\sum_{\sigma\in \Omega} \nu(\sigma) \frac{\nu(\sigma)}{\mu(\sigma)} } - 1;\]
    \item $\mathrm{KL}(x) := x\log x$ (with the convention $0 \log 0 = 0$), which induces the KL-divergence
    \[\DKL{\nu}{\mu} = \sum_{\sigma \in \Omega} \nu(\sigma) \log \frac{\nu(\sigma)}{\mu(\sigma)}.\]
\end{itemize}

The $\mathrm{KL}$-entropy is typically referred to simply as \emph{entropy} (denoted $\Ent[\mu]{f}$), while the $\chi^2$-entropy is commonly called \emph{variance} (denoted $\Var[\mu]{f}$). In particular, for $X \sim \mu$, 
\begin{align*}
    \Ent[\mu]{f} = \Ent{f(X)} &= \E{f(X) \log f(X)} - \E{f(X)}\log \E{f(X)}; \\
    \Var[\mu]{f} = \Var{f(X)} &= \E{f(X)^2} - \E{f(X)}^2.
\end{align*}

Throughout the paper, we will focus  on the cases $\phi \in \{\chi^2, \mathrm{KL}\}$. Accordingly, we assume that the domain of $\phi$ is
\begin{align*}
    D = \mathbb{R}_{\geq 0}.
\end{align*}

\subsection{Markov chain mixing time}

Let $(X_t)_{t \ge 0}$ be a Markov chain defined over a finite state space $\Omega$ with transition matrix $P \in \mathbb{R}_{\ge 0}^{\Omega \times \Omega}$. 
A distribution $\mu$ over $\Omega$ is called a \emph{stationary distribution} of $P$ if $\mu = \mu P$.

For a convergent Markov chain $P$, its \emph{mixing time}  is defined as
\begin{align*}
\Tmix := \max_{x \in \Omega}\min\set{t \mid \DTV{P^t(x,\cdot)}{\mu} \leq 1/4}.
\end{align*}

A Markov chain $P$ is \emph{time reversible} with respect to the distribution $\mu$ if it satisfies the {detailed balance equation}:
\begin{align*}
    \forall x,y \in \Omega, \quad \mu(x)P(x,y) = \mu(y)P(y,x).
\end{align*}
This implies that $\mu$ is a stationary distribution of $P$. 
We assume throughout that all Markov chains considered in this paper are time reversible.

For a reversible Markov chain $P$ on a state space $\Omega$, all eigenvalues are real. 
Denote the eigenvalues of $P$ by
    $1 = \lambda_1 \geq \lambda_2 \geq \cdots \geq \lambda_{\abs{\Omega}} \geq -1$.
The \emph{spectral gap} of $P$ is defined by
    $\gamma(P) := 1 - \lambda_2$.
Moreover, the \emph{absolute spectral gap} is defined by
\[\gamma_\star(P) := 1 - \max\set{\abs{\lambda} \mid \lambda \text{ is an eigenvalue of $P$}, \lambda \neq 1}.\]
We note that all Markov chains considered in this paper have nonnegative eigenvalues, in which case $\gamma(P) = \gamma_\star(P)$.

We now define decay of $\phi$-entropy, which captures convergence under $\phi$-divergence.
\begin{definition}[Decay of $\phi$-entropy]
Let $(X_t)_{t\ge 0}$ be a Markov chain over $\Omega$ with transition matrix $P$, and let $\mu$ be its stationary distribution. 
We say that the Markov chain satisfies \emph{decay of $\phi$-entropy with rate $\kappa$} if for all functions $f:\Omega \to \mathbb{R}$,
\begin{align*}
\Ent[\mu][\phi]{Pf} \le (1-\kappa)\Ent[\mu][\phi]{f}.
\end{align*}
\end{definition}
For time-reversible chains, decay of variance implies a lower bound on the absolute spectral gap. Specifically, 
if a reversible chain $P$ satisfies decay of variance with rate $\kappa$, then its absolute spectral gap satisfies
\begin{align}\label{eq:var-decay-implies-spectral-gap}
    \gamma_{\star}(P)\ge \frac{1}{2}\kappa
\end{align}
Moreover, for time-reversible chains, depending on the choice of $\phi$,
decay of $\phi$-entropy implies upper bounds on the mixing time. In particular,
\begin{itemize}
    \item if a reversible chain satisfies decay of variance with rate $\kappa$, then its mixing time satisfies
\begin{align}\label{eq:var-decay-implies-mixing}
\Tmix \le O\tp{\kappa^{-1}\log\frac{1}{\mu_{\min}}},
\end{align}
    \item if a reversible chain satisfies decay of entropy with rate $\kappa$, then its mixing time satisfies~\cite[Lemma~2.4 and Remark~2.5]{blanca2022entropy}
    \begin{align}\label{eq:ent-decay-implies-mixing}
\Tmix \le O\tp{\kappa^{-1}\log\log\frac{1}{\mu_{\min}}}, %
\end{align}
where $\mu_{\min} := \min_{\sigma \in \Omega(\mu)}{\mu(\sigma)}$.
\end{itemize}

When specializing to Glauber dynamics, the mixing behavior can be captured via approximate tensorization of $\phi$-entropy. We use a unified formulation of the variance and entropy notions in~\cite[Definition~2.1 and Appendix~A.1]{chen2021optimal}.

\begin{definition}[Approximate tensorization of $\phi$-entropy]
For a distribution $\mu$ over $\Omega = \{0,1\}^n$, let $X \sim \mu$. We say that $\mu$ satisfies \emph{$K$-approximate tensorization of $\phi$-entropy} if for any function $f:\Omega \to \mathbb{R}$,
\begin{align*}
\Ent[][\phi]{f(X)} \le K \cdot \sum_{i=1}^{n}\E{\Ent[][\phi]{f(X) \mid X_{-i}}},
\end{align*}
where $X_{-i}$ denotes the configuration $X$ restricted to coordinates $[n]\setminus\{i\}$.
\end{definition}

Approximate tensorization of $\phi$-entropy can be interpreted as a notion of closeness to a product distribution. 
In particular, if $\mu$ is a product distribution, i.e., $\mu = \bigotimes_{i=1}^n \mu_i$, then $\mu$ satisfies $1$-approximate tensorization of $\phi$-entropy for $\phi \in \{\chi^2, \mathrm{KL}\}$; see~\cite[Lemma~4.1 and Appendix~A.1]{chen2021optimal}.

\begin{proposition} \label{lem:AT-prod-dist}
If $\mu$ is a product distribution on $\{0,1\}^n$, then $\mu$ satisfies $1$-approximate tensorization of variance and $1$-approximate tensorization of entropy.
\end{proposition}

In particular, approximate tensorization of variance implies rapid mixing of Glauber dynamics. 
If a distribution $\mu$ on $\{0,1\}^n$ satisfies $K$-approximate tensorization of variance, then its Glauber dynamics has spectral gap at least $1/(nK)$ by~\cite[Fact~A.3]{chen2021optimal}. The standard spectral mixing bound~\cite[Theorem~12.4]{levin2017markov} then gives
\begin{align} \label{eq:AT-implies-mixing}
    \Tmix \leq n K \log\frac{1}{\mu_{\min}}, \quad \text{where} \quad \mu_{\min} := \min_{\sigma \in \Omega(\mu)} \mu(\sigma).
\end{align}

\subsection{Spectral independence}

Let $\mu$ be a distribution over $\{0,1\}^V$. 
A partial configuration $\tau \in \{0,1\}^{\Lambda}$ on a subset $\Lambda \subseteq V$ is called a \emph{feasible pinning} 
if $\tau \in \Omega(\mu_\Lambda)$ (i.e., $\mu_\Lambda(\tau) > 0$).

Let $X \sim \mu$. The \emph{influence matrix} $\Psi_{\mu}\in\mathbb{R}^{V\times V}$ is defined as: 
\begin{align}\label{eq:influence-matrix}
\Psi_{\mu}(u,v):=\begin{cases}
\Pr{X_v=1\mid X_u=1}-\Pr{X_v=1\mid X_u=0} &\quad \text{if }\E{X_u}\in(0,1)\text{ and }u\ne v,\\
0 &\quad \text{otherwise}.
\end{cases}
\end{align}
We say that $\mu$ is \emph{$\rho$-spectrally independent} if for every feasible pinning $\tau$,
\begin{align*}
\lambda_{\max}(\Psi_{\mu^\tau}) \le \rho.
\end{align*}

To upper bound the maximum eigenvalue of the influence matrix, it suffices to consider the notion of total influence defined below.
\begin{definition}[Total influence]\label{def:total-influence}
Given a distribution $\mu$ over $\{0,1\}^{V}$ and a vertex $r \in V$, the \emph{total influence} of $r$ is defined as
\begin{align*}
\-{TI}^{\mu}_{r} := 1+\sum_{v\in V}\abs{\Psi_{\mu}(r,v)}.
\end{align*}
The additional $+1$ accounts for the self-influence of $r$. 
\end{definition}

Total influence provides an upper bound on spectral independence.

\begin{lemma}\label{lem:TI-to-SI}
Let $\mu$ be a distribution over $\{0,1\}^V$. 
Suppose that for any vertex $r \in V$ and any feasible pinning $\tau$, it holds that
\begin{align*}
\-{TI}^{\mu^{\tau}}_{r} \le \rho.
\end{align*}
Then $\mu$ is $(\rho-1)$-spectrally independent.
\end{lemma}

\begin{proof}
For any feasible pinning $\tau$, the maximum eigenvalue of $\Psi_{\mu^{\tau}}$ is upper bounded by
\begin{align*}
\lambda_{\max}(\Psi_{\mu^{\tau}}) \le \norm{\Psi_{\mu^{\tau}}}_{\infty}
= \max_{r\in V}\sum_{v\in V}\abs{\Psi_{\mu^{\tau}}(r,v)} = \max_{r\in V}\-{TI}^{\mu^{\tau}}_r-1.
\end{align*}
Since $\-{TI}^{\mu^{\tau}}_r\le\rho$ for any $r\in V$ and any feasible pinning $\tau$, we have $\lambda_{\max}(\Psi_{\mu^{\tau}})\le\rho-1$ for any feasible pinning $\tau$. Thus $\mu$ is $(\rho-1)$-spectrally independent.
\end{proof}

The following stronger notion of coupling independence is effective for analyzing spectral independence.
\begin{definition}[Coupling independence]\label{def:CI}
Let $\mu$ be a distribution over $\{0,1\}^n$. Let $C > 0$. 
We say that $\mu$ is \emph{$C$-coupling independent} if for all $\Lambda \subseteq [n]$ and all $\sigma, \tau \in \Omega(\mu_\Lambda)$ with Hamming distance $\dist(\sigma,\tau) = 1$,
    \begin{align*}
        \inf_{\xi} \E[(X,Y) \sim \xi]{\dist(X,Y)} \leq C,
    \end{align*}
where the infimum is taken over all couplings $\xi$ between $\mu^\sigma$ and $\mu^\tau$.
\end{definition}

\subsection{Localization schemes}
The notion of \emph{localization schemes}, introduced in \cite{chen2022localization}, provides a continuous interpolation between a trivial measure and the target measure.

Let $\mu$ be a probability distribution over $\Omega$. 
A localization scheme associated with the target distribution $\mu$ consists of a pair of continuous-time processes:
\begin{itemize}
    \item The \emph{noising process} $(X_t)_{t \in [0,1]}$ is a Markov process such that $X_0\sim\mu$ and $X_1$ follows a Dirac measure. 
\item The \emph{denoising process} $(Y_t)_{t \in [0,1]}$ is defined as the time reversal of $(X_t)_{t \in [0,1]}$, that is, $Y_\theta \stackrel{d}{=} X_{1-\theta}$. In particular, $Y_0$ follows a Dirac measure and $Y_1 \sim \mu$.
\end{itemize}
Since the two processes are time reversals of each other, we will use the denoising process $(Y_t)_{t \in [0,1]}$ to represent the localization scheme, and call it the \emph{localization process}.

A localization scheme naturally induces a Markov chain with stationary distribution $\mu$.

\begin{definition}[Down-up chain induced by a localization scheme]
\label{def:LS-induce-down-up-walk}
Let $\mu$ be a distribution over $\Omega$, and let $(Y_t)_{t \in [0,1]}$ be a denoising process with respect to $\mu$. 
Fix any $\theta \in [0,1]$. 
The \emph{down-up chain} $P_{1 \leftrightarrow \theta}$ induced by $(Y_t)_{t \in [0,1]}$ is a Markov chain on $\Omega$ defined as follows. 

Given the current state $x_t \in \Omega$,  the next state $x_{t+1}$ is generated by:
\begin{enumerate}[leftmargin=2cm]
\item[(\textsc{Down})]
Sample $x'\sim\-{Law}(Y_{\theta}\mid Y_1=x_t)$;
\item[(\textsc{Up})]
Sample $x_{t+1}\sim\-{Law}(Y_1\mid Y_{\theta}=x')$.
\end{enumerate}
\end{definition}

It is straightforward to verify that the induced down-up chain $P_{1 \leftrightarrow \theta}$ has stationary distribution $\mu$. 
Moreover, its mixing time can be analyzed via the following notion of approximate conservation of variance, 
which is a key property of localization schemes.

\begin{definition}[Approximate conservation of $\phi$-entropy]\label{def:approx-conservation-var}
Given a denoising process $(Y_t)_{t \in [0,1]}$ and a parameter $\theta \in [0,1]$, we say that $(Y_t)_{t \in [0,1]}$ satisfies \emph{$R$-approximate conservation of $\phi$-entropy up to time $\theta$} if for all functions $f:\Omega \to \mathbb{R}_{\geq 0}$,
\begin{align*}
\Ent[][\phi]{f(Y_1)} \le R\cdot\E{\Ent[][\phi]{f(Y_1)\mid Y_{\theta}}}.
\end{align*}
\end{definition}

Approximate conservation of $\phi$-entropy of the denoising process implies $\phi$-entropy decay for the induced down-up chain.

\begin{lemma}[{\cite[Lemma 3.11]{chen2025uniqueness}}]
\label{lem:AC-var-to-var-decay}
If the denoising process $(Y_t)_{t \in [0,1]}$ satisfies $R$-approximate conservation of $\phi$-entropy up to time $\theta$, 
then the induced down-up chain $P_{1 \leftrightarrow \theta}$ satisfies decay of $\phi$-entropy with rate $1/R$.
\end{lemma}

An important application of localization schemes is the ``boosting'' (or ``lifting'') of mixing properties from a regime where they are known to hold to a new regime.

\begin{lemma}[{\cite[Lemma 3.13]{chen2025uniqueness}}]\label{lem:boosting}
Let $(Y_t)_{t \in [0,1]}$ be a denoising process with respect to a distribution $\mu$. 
Suppose there exists $\theta \in [0,1]$ such that:
\begin{enumerate}[label=(\roman*)]
\item The conditional distribution $\-{Law}(Y_1\mid Y_{\theta})$ always satisfies $K$-approximate tensorization of $\phi$-entropy;
\item $(Y_t)_{t \in [0,1]}$ satisfies $R$-approximate conservation of $\phi$-entropy up to time $\theta$.
\end{enumerate}
Then $\mu$ satisfies $(KR)$-approximate tensorization of $\phi$-entropy.
\end{lemma}

Finally, approximate conservation of $\phi$-entropy can be established via $\phi$-entropic stability, which quantifies the rate of local $\phi$-entropy contraction along the denoising process.

\begin{definition}[$\phi$-entropic stability] \label{def:phi-ent-stable}
Let $(Y_t)_{t \in [0,1]}$ be a denoising process with respect to a distribution $\mu$ over $\Omega$.
We say that $(Y_t)_{t \in [0,1]}$ is \emph{$\phi$-entropically stable with rate $C$ at time $\theta$} if for any function $f:\Omega \to \mathbb{R}$, letting $F = f(Y_1)$, it holds for every $T$ with $\Pr{Y_{\theta}=T}>0$,
\begin{align*}
\lim_{h\to 0^+}\frac{1}{h}\Ent[][\phi]{\E{F\mid Y_{\theta+h}}\mid Y_{\theta}=T} \le \frac{C}{1-\theta}\Ent[][\phi]{F\mid Y_{\theta}=T}. 
\end{align*}
In particular, the following conventional terminology is used:
\begin{itemize}
    \item $\mathrm{KL}$-entropic stability is called \emph{entropic stability};
    \item $\chi^2$-entropic stability is called \emph{spectral stability}.
\end{itemize}
\end{definition}

The following lemma connects $\phi$-entropic stability to approximate conservation of $\phi$-entropy.

\begin{lemma}[{\cite[Theorem 3.16]{chen2025uniqueness}}]
\label{lem:SS-to-AC-variance}
Let $C:[0,1] \to \mathbb{R}_{\ge 0}$. 
If a denoising process $(Y_t)_{t \in [0,1]}$ is $\phi$-entropically stable with rate $C(t)$ at any time $t\in [0,1]$, 
then for any $\theta \in [0,1]$, the process $(Y_t)_{t \in [0,1]}$ satisfies $R$-approximate conservation of $\phi$-entropy up to time~$\theta$, where
\begin{align*}
R = \exp\tp{\int_{0}^{\theta}\frac{C(t)}{1-t}\-{d}t}.
\end{align*}
\end{lemma}

\subsection{Field dynamics}\label{sec:prelim-field-dynamics}

An important example of localization schemes is the \emph{field dynamics} introduced in \cite{chen2021rapid}, 
which, as observed in \cite{chen2022localization}, corresponds to \emph{negative-field localization}.
To distinguish it from the edge-tilting variant introduced in this paper, we will refer to it as the 
\emph{vertex-tilting field dynamics}, or simply \emph{vertex-field dynamics}.

We focus on distributions over $\{0,1\}^n$. 
All results stated here extend naturally to distributions over $\{\pm 1\}^n$.
Moreover, when the context is clear, we identify a configuration $\sigma \in \{0,1\}^n$ with the subset $\{v \in [n] : \sigma_v = 1\}$.

The following defines the denoising process for the vertex-field dynamics. 
In particular, the vertex-field dynamics is the  down–up chain induced by this denoising process.
\begin{definition}[Vertex-field denoising process] \label{def:field-dynamics-denoising}
Let $\mu$ be a distribution over $\set{0,1}^n$, and let $X \sim \mu$. 
For each $v \in [n]$, assign an independent random variable $r_v \sim \-{Uniform}[0,1]$. 

The \emph{vertex-field denoising process for $\mu$}, denoted by $(Y_t)_{t \in [0,1]}$, is defined by
\begin{align*}
Y_t := \set{v\in [n]:X_v=1\land r_v\le t}.
\end{align*}
\end{definition}

For $\theta \in [0,1]$, recall the  distribution $\theta * \mu$ with tilted external fields, defined by
\begin{align*}
(\theta*\mu)(\sigma) \propto \mu(\sigma)\cdot\theta^{\abs{\sigma}}, \quad \forall\sigma\in\set{0, 1}^n.
\end{align*}

A key property of the vertex-field denoising process  $(Y_t)_{t \in [0,1]}$ is that the posterior distribution corresponds to such a reweighting.

\begin{proposition} \label{prop:posterior-FD}
Let $(Y_t)_{t \in [0,1]}$ be the vertex-field denoising process for $\mu$. Then
\begin{align*}
\-{Law}(Y_1\mid Y_t)=(1-t)*\mu^{Y_t},
\end{align*}
where $\mu^{Y_t}$ denotes the conditional distribution obtained by pinning all variables in $Y_t$ to $1$.
\end{proposition}

\begin{proof}
For any $\sigma\in\set{0, 1}^n$ and $T\subseteq V$ such that $\Pr{Y_t=T}>0$, by the Bayes's theorem,%
\begin{align}
\nonumber
\Pr{X=\sigma\mid Y_t=T} =& \frac{\Pr{X=\sigma}}{\Pr{Y_t=T}}\cdot\Pr{Y_t=T\mid X=\sigma}\\
\nonumber
=& \frac{\mu(\sigma)}{\Pr{Y_t=T}}\cdot t^{\abs{T}}(1-t)^{\abs{\sigma}-\abs{T}} \cdot \*1[T \subseteq \sigma]\\
\label{eq:FD-posterior-calc}
\propto& (1-t)^{\abs{\sigma}}\mu(\sigma) \cdot \*1[T \subseteq \sigma],
\end{align}
where in the expression $T \subseteq \sigma$ we view $\sigma$ as a set $\set{v \mid \sigma_v = 1}$.
This finishes the proof.
\end{proof}

It follows from \Cref{prop:posterior-FD} that the induced down-up chain of the vertex-field denoising process is the following Markov chain, previously known as the field dynamics.
\begin{definition}[Vertex-field dynamics] \label{def:field-dynamics-down-up-chain}
Let $\mu$ be a distribution over $\{0,1\}^n$. 
Let $\theta \in [0,1]$. 
The \emph{vertex-field dynamics} for $\mu$ with tilt $\theta$ is the following Markov chain on $\{0,1\}^n$.

Given a configuration $\sigma \in \{0,1\}^n$, the chain updates as follows:
\begin{enumerate}[leftmargin=2cm]
\item[(\textsc{Down})]
Generate $T \subseteq \{v : \sigma_v = 1\}$ by independently removing each $v$ with probability $\theta$;
\item[(\textsc{Up})]
Resample $\sigma\sim(\theta*\mu)^{\sigma_T}$.
\end{enumerate}
\end{definition}

\paragraph{Spectral stability.} 
Spectral stability of the vertex-field denoising process admits an equivalent characterization in terms of an operator-norm bound on the covariance matrix.
\begin{lemma}[\text{\cite[Proposition 3.3]{chen2025rapid}}] \label{lem:spectral-stable-matrix-form}
Let $(Y_t)_{t \in [0,1]}$ be the vertex-field denoising process for $\mu$ over $\set{0,1}^n$. 
Then for any $\theta \in [0,1]$, the following statements are equivalent:
    \begin{itemize}
    \item $(Y_t)_{t \in [0,1]}$ is spectrally stable with rate $C$ at time $\theta$;
    \item for any $S$ with $\Pr{Y_\theta = S} > 0$,
        \begin{align*}
            \Cov((1-\theta) * \mu^S) \preceq C \cdot \diag\set{\*m((1-\theta) * \mu^S)},
        \end{align*}
    where for a distribution $\nu$ over $\{0,1\}^n$, $\*m(\nu) := \mathbb{E}_{X \sim \nu}[X]$ denotes its mean vector.
    \end{itemize}
\end{lemma}

An effective way to establish spectral stability is via the notion of coupling independence (\Cref{def:CI}).
Specifically, it has been established (see, e.g., \cite[Proposition 4.3]{chen2023near} and \cite[Corollary 16]{chen2025spectral}) that 
$C$-coupling independence of the distribution $(1-\theta) * \mu^S$ implies
\begin{align*}
    \Cov((1-\theta) * \mu^S) \preceq C \cdot \diag\set{\*m((1-\theta) * \mu^S)}.
\end{align*}
Thus, we obtain the following corollary of \Cref{lem:spectral-stable-matrix-form}.

\begin{lemma}
\label{lem:cond-stable-FD} %
Let $(Y_t)_{t \in [0,1]}$ be the vertex-field denoising process for $\mu$. 
Fix any $\theta \in (0,1]$. 
If $(1 - \theta) * \mu$ is $C$-coupling independent, then $(Y_t)_{t \in [0,1]}$ is $C$-spectrally stable at time $\theta$.
\end{lemma}

\paragraph{Entropic stability.}
Similarly, entropic stability of the vertex-field denoising process also admits an equivalent characterization.
\begin{lemma}[\text{\cite[Proposition 3.7]{chen2025rapid}}] \label{lem:entropic-stable-max-ent-form}
Let $(Y_t)_{t \in [0,1]}$ be the vertex-field denoising process for $\mu$ over $\set{0,1}^n$. 
Then for any $\theta \in [0,1]$, the following statements are equivalent:
    \begin{itemize}
    \item $(Y_t)_{t \in [0,1]}$ is entropically stable with rate $C$ at time $\theta$;
    \item for any $S$ with $\Pr{Y_\theta = S} > 0$, and any $\nu \ll (1-\theta)\ast \mu^S$,
        \begin{align*}
            \sum_{i\in [n]\setminus S} q_i \log \frac{q_i}{p_i} - (q_i - p_i) \leq C \cdot \DKL{\nu}{(1-\theta) * \mu^S},
        \end{align*}
    where $\*p := \*m((1-\theta) * \mu^S)$ and $\*q := \*m(\nu)$ are the mean vectors.
    \end{itemize}
\end{lemma}

\section{Barriers to Localization Schemes} \label{sec:proof-main-overview}

\paragraph{Localization schemes.}
Localization schemes, introduced in \cite{chen2022localization}, have proven effective for establishing rapid mixing of Glauber dynamics at criticality \cite{chen2025uniqueness,chen2025improved}.

A localization scheme interpolates between a target distribution $\mu$ and a trivial distribution via a denoising process $(Y_t)_{t \in [0,1]}$.
In particular, $(Y_t)_{t \in [0,1]}$ is a Markov process such that $Y_0 = \emptyset = \mathbf{0}$ corresponds to a state with no information, and $\mathrm{Law}(Y_1) = \mu$.
Using this denoising process, the target distribution $\mu$ can be decomposed as a mixture of distributions via a direct application of the law of total expectation:
\begin{align*}
    \forall S, \quad \mu(S) = \Pr{Y_1 = S} = \E{\Pr{Y_1 = S \mid Y_t}} = \E{\mu_t(S)},
\end{align*}
where $\mu_t = \mathrm{Law}(Y_1 \mid Y_t)$ is a random measure depending on $Y_t$.

This interpolation provides a systematic mechanism for lifting mixing properties from subcritical regimes to the critical threshold.

In particular, consider the Gibbs distribution $\mu$ of an antiferromagnetic Ising model with edge activity $\beta = \gamma < 1$ and external field $\lambda > 0$ in the uniqueness regime.
By applying suitable localization schemes at appropriately chosen times $t^*$, the target distribution $\mu$ can be transformed into ``easy'' distributions $\mu_{t^*}$:

\begin{itemize}
\item \emph{Stochastic localization}: $\mu_t$ is the Gibbs distribution of an Ising model with tilted edge activities $\beta_t = \gamma_t = \beta^{\frac{1-2t}{1-t}}$ and randomly tilted external fields $\boldsymbol{\lambda}' \in \mathbb{R}^n_{>0}$, which are determined by the localization trajectory up to time $t$.
By choosing $t^* = \tfrac{1}{2}$, the resulting distribution $\mu_{t^*}$ becomes a product distribution.
\item \emph{Negative-field localization}: $\mu_t$  is the Gibbs distribution of an Ising model with unchanged edge activities $\beta_t = \gamma_t = \beta$ and a tilted external field $\lambda_t = (1 - t)\lambda$.  
By choosing $t^* = \tfrac{9}{10}$, the distribution $\mu_{t^*}$ satisfies the Dobrushin condition.
\end{itemize}

The following boosting argument for the mixing time of Glauber dynamics, from the “easy” distribution $\mu_{t^*}$ to the target  $\mu$, is supported by the localization framework:
\begin{align*}
    \Tmix(\mu) 
    &\lesssim \exp\tp{\int_0^{t^*} \sup_{Y_t} \norm{\Cov(\mu_t)}_2 \-d t} \times \sup_{Y_{t^*}} \Tmix(\mu_{t^*}) \\
    &\lesssim \exp\tp{\underbrace{\int_0^{t^*} \min\set{n, \sup_{Y_t} \frac{1}{\delta(\mu_t)}} \-d t}_{=:\mathrm{I}}} \times \mathrm{poly}(n),
\end{align*}
where $\lesssim$ hides polynomial factors, and $\delta(\mu_t)$ represents the “slackness” between the distribution $\mu_t$ and the target distribution $\mu$, to be discussed later.
Therefore, in order to show that $\Tmix(\mu) = \mathrm{poly}(n)$, it suffices to bound the integral $\mathrm{I}$ by $O(\log n)$.
By a direct calculation,
\begin{align*}
    \mathrm{I} \leq \int_{0}^{1/n} n \,\-d t + \int_{1/n}^{t^*} \sup_{Y_t} \frac{1}{\delta(\mu_t)} \,\-d t.
\end{align*}
This implies that we require the following condition on the slackness of the distributions $\mu_t$:
\begin{align} \label{eq:required-bound-for-slack}
\delta(\mu_t) = \Omega(t) \quad \text{for all} \quad  Y_t.
\end{align}

\paragraph{Slackness in the uniqueness condition.}
For antiferromagnetic two-spin systems,  $\delta(\mu_t)$ corresponds to the slackness of the contraction in the uniqueness condition.

In particular, let $\mu$ be the Gibbs distribution of an antiferromagnetic two-spin system with parameters $(\beta,\gamma,\lambda)$ on a $\Delta$-regular graph, where $\Delta - 1 =: d \geq 2$.
We say that $\mu$ satisfies the uniqueness condition with \emph{exact slack} $\delta$ if
$\delta = \delta(\mu)  := 1 - \abs{f'(\hat{x}_d)}$  
where $\hat{x}_d$ is the unique positive fixed point of the tree recursion $f_d$, i.e., $\hat{x}_d = f(\hat{x}_d)$ (see \Cref{def:d-unique}).

The uniqueness regime itself undergoes a phase transition at the threshold $\sqrt{\beta\gamma} = \frac{d-1}{d+1}$.
Intuitively, the ``flexibility'' of the external field changes sharply at this threshold:
\begin{itemize}
 \item If $\sqrt{\beta\gamma} \geq \frac{d-1}{d+1}$, then  $(\beta,\gamma,\lambda)$ is $d$-unique for all $\lambda > 0$ (i.e., $\delta \geq 0$).
 \item If $\sqrt{\beta\gamma} < \frac{d-1}{d+1}$, then there exists $\lambda > 0$ such that $(\beta,\gamma,\lambda)$ is not $d$-unique (i.e., $\delta < 0$).
\end{itemize}
We refer the reader to \Cref{fig:Ising-unique-illustration} for an illustration of this phenomenon.

This critical behavior suggests that when $\sqrt{\beta\gamma} \in [\frac{d-1 - \epsilon}{d+ 1+ \epsilon}, \frac{d-1}{d+1})$ for sufficiently small $\epsilon > 0$, the slack $\delta$ may depend poorly on the external field $\lambda$.
In particular, when $ \sqrt{\beta\gamma} \in [\frac{d-1 - \epsilon}{d+ 1+ \epsilon}, \frac{d-1}{d+1})$ and $\delta(\mu) = 0$ (i.e., at criticality), the slack $\delta(\mu_t)$ for the  distribution $\mu_t$ with tilted external field $(1-t)\lambda$ (as in the negative-field localization) satisfies
\begin{align} \label{eq:bad-reliance}
\delta(\mu_t) = O(t \cdot \epsilon^{1/2}).
\end{align}
We defer the calculation of this bound to \Cref{rem:bad-dependency-field-dynamics}.

\begin{figure}
    \centering

\begin{tikzpicture}[ 
  scale=1.4,
  xscale=20, yscale=0.9, 
  declare function={
    d = 4;
    theta(\b) = d * (1 - \b^2) - (1 + \b^2);
    x1(\b) = (theta(\b) - sqrt(theta(\b)^2 - 4 * \b^2))/(2 * \b);
    L(\b) = x1(\b) ((x1(\b) + \b)/(\b * x1(\b) + 1))^d;
    R(\b) = 1/L(\b);
  }]
  \begin{scope}
  \fill[red!20] plot[domain=0.5:0.599] (\x, {L(\x)}) -- plot[domain=0.599:0.5] (\x, {R(\x)}) -- cycle;
  \node at (0.545, 1.1) {\color{red!50!black}non-unique};
  \fill[green!20] plot[domain=0.5:0.599] (\x, {L(\x)}) -- plot[domain=0.599:0.5] (\x, {R(\x)}) -- (0.5,4.2) -- (0.695,4.2) -- (0.695,0) -- (0.5,0) -- cycle;
  \node at (0.65, 3.5) {\color{green!50!black}unique};
  \draw[blue, thick, domain=0.5:0.6] plot (\x, {L(\x)});
  \draw[blue, thick, domain=0.5:0.6] plot (\x, {R(\x)});
  \draw[black,->] (0.5,0) -- (0.5,4.5) node[left] {$\lambda$};
  \draw[black,->] (0.5,0) node[below] {$0.5$} -- (0.7,0) node[below] {$\beta$};
  \draw[-, dashed] (0.6,0) node[below] {$0.6$} -- (0.6,4.2);
  \end{scope}
\end{tikzpicture}

    \caption{Illustration of uniqueness regime for the Ising model ($\beta = \gamma$) with $d = 4$}
    \label{fig:Ising-unique-illustration}
\end{figure}
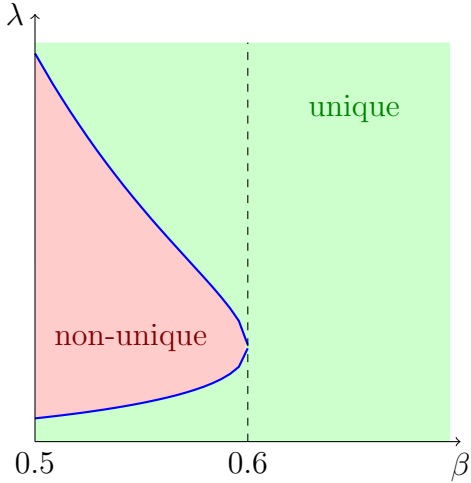

\paragraph{A barrier case for existing localization schemes.}
Following the preceding discussion, we identify a barrier case for existing approaches based on stochastic localization and negative-field localization.
Specifically, when
$\sqrt{\beta\gamma} \in [\frac{d-1 - \epsilon}{d+ 1+ \epsilon}, \frac{d-1}{d+1})$ 
for sufficiently small $\epsilon > 0$, neither scheme is able to establish rapid mixing at criticality:
\begin{itemize}
\item \emph{Stochastic localization} fails in this regime because it introduces uncontrolled external fields. Since the uniqueness condition is highly sensitive to the external field when $\sqrt{\beta\gamma} \leq \frac{d-1}{d+1}$, these perturbations preclude a direct proof of rapid mixing.
\item \emph{Negative-field localization} also fails because the slack $\delta$ in the uniqueness condition exhibits a poor dependence on the external field $\lambda$ in this regime. In particular, combining \eqref{eq:required-bound-for-slack} and \eqref{eq:bad-reliance}, we see that for sufficiently small $\epsilon$, the integral derived earlier is insufficient to yield a polynomial mixing-time bound.
\end{itemize}

\paragraph{Our solution and novelty.}
We address this barrier by introducing a new localization scheme, which we call \emph{edge-tilting field dynamics}, a controlled interaction localization scheme.
Unlike negative-field localization, which tilts the external field, this  approach instead tilts the edge activities $\beta$ and $\gamma$, in a manner reminiscent of stochastic localization.
Crucially, this is achieved through the introduction of pinnings (which preserve the uniqueness condition), rather than introducing uncontrollable perturbations to the external field.

This key feature ensures that the dynamics remains well-behaved near criticality for general antiferromagnetic two-spin systems, and in particular, makes it applicable to instances in the barrier regime
$\sqrt{\beta\gamma} \in [\frac{d-1 - \epsilon}{d+ 1+ \epsilon}, \frac{d-1}{d+1})$.

\section{Edge-Tilting Field Dynamics}\label{sec:edge-field-dynamics}

\subsection{Generalizing the field dynamics}\label{sec:generalized-field-dyanmics}
Let $\mu$ be a distribution over $\set{0,1}^V$, where $V$ is a finite ground set of Boolean variables.
Let $\mathcal{A}$ be a collection of events, i.e., each $A \in \mathcal{A}$ is a subset of $\set{0,1}^V$.
Let $X \sim \mu$, and for each event $A \in \mathcal{A}$ define the indicator variable $Z_A = Z_A(X) := \mathbf{1}[X \in A]$. 
Let $Z = Z(X) := (Z_A)_{A \in \mathcal{A}}$.
Furthermore, we identify configurations $\sigma \in \set{0,1}^V$ (and $T \in \set{0,1}^{\mathcal{A}}$) with the subsets $\sigma \subseteq V$ (and $T \subseteq \mathcal{A}$) that they indicate.

Let $\pi$ denote the joint distribution of $(X, Z)$, i.e.,
\begin{align} \label{eq:dist-joint}
    \pi := \mathrm{Law}(X, Z(X)),
\end{align}
which is a distribution over $\set{0,1}^{V \cup \mathcal{A}}$.
By construction, $\mu = \pi_V$. 
We also write $\pi_{\mathcal{A}} = \mathrm{Law}(Z)$ for the marginal distribution on the event indicators.

We now formalize the event-field dynamics introduced in \Cref{rem:event-FD} using the above notation.
For simplicity, we assume that all tilts $\{\theta_A\}_{A \in \mathcal{A}}$ are identical.

\begin{definition}[Event-field dynamics]\label{def:event-FD-explicit}
Let $\theta \in (0,1)$. 
The \emph{event-field dynamics} for $\mu$ with respect to the event family  $\mathcal{A}$ and tilt parameter $\theta$ is a Markov chain on the state space $\Omega(\mu)$.
From the current state $X_t \in \Omega(\mu)$, the next state $X_{t+1}$ is generated as follows:
\begin{itemize}[leftmargin=2.5cm]
    \item[($V \to \mathcal{A}$)] Set $Z_t := Z(X_t) = \{A \in \mathcal{A} \mid X_t \in A\}$;
    \item[({Down}-$\mathcal{A}$)] Generate a random subset $T \subseteq Z_t$ by independently removing each $A \in Z_t$ with probability $\theta$;
    \item[({Up}-$\mathcal{A}$)] Sample $Z_{t+1}\sim (\theta * \pi_{\mathcal{A}})(\cdot \mid Z_{t+1} \supseteq T)$, where $\theta * \pi_{\mathcal{A}}$ is as defined in \eqref{eq:def-field-tilt};
    \item[($\mathcal{A} \to V$)] Sample $X_{t+1} \sim \pi_V(\cdot \mid Z_{t+1})$.
\end{itemize}
\end{definition}

\begin{remark}[Generalization of vertex-field dynamics]\label{remark:general-field-dynamics}
The process defined in \Cref{def:event-FD-explicit} generalizes the vertex-field dynamics  (\Cref{def:field-dynamics-down-up-chain}).
Specifically, the two intermediate steps (Down-$\mathcal{A}$) and (Up-$\mathcal{A}$) together simulate a single transition of the vertex-field dynamics for the marginal distribution $\pi_{\mathcal{A}}$. 
The outer wrapper steps ($V \to \mathcal{A}$) and ($\mathcal{A} \to V$) then lift this Markov chain on $\pi_{\mathcal{A}}$ to a Markov chain on $\mu = \pi_V$, from which we aim to sample.

We present the event-field dynamics in this form so that analyses of vertex-field dynamics can be lifted to this more general setting.
\end{remark}

The event-field dynamics unifies the vertex-field dynamics, the edge-field dynamics, and the Swendsen--Wang dynamics as special cases.

\begin{example} \label{exm:event-FD}
Let $\mu$ be a distribution on $\set{0,1}^V$, and let $G = (V,E)$ be a graph. 
    \begin{itemize}
    \item \textbf{Vertex-field dynamics.} We set 
    \[
    \mathcal{A} := \{A_v \mid v \in V\}, \quad \text{where } A_v := \{\sigma \in \set{0,1}^V \mid \sigma_v = 1\}, \quad \forall v \in V.
    \]
For $\sigma \in \set{0,1}^V$, let $\mathcal{A}(\sigma) := \set{A \in \mathcal{A}\mid \sigma \in A}$ denote the set of events occurring on~$\sigma$. 

Given $T \subseteq \mathcal{A}$ generated in the (Down-$\mathcal{A}$) step, we have
\begin{align}
    \nonumber
    \Pr{X_{t+1} = \sigma \mid T} 
    =\,& \Pr{X_{t+1} = \sigma, Z_{t+1} = \mathcal{A}(\sigma) \mid T} \\
    \nonumber
    =\,& \pi_V(\sigma \mid Z = \mathcal{A}(\sigma)) \cdot\Pr{Z_{t+1} = \mathcal{A}(\sigma) \mid T} \\
    \label{eq:FD-pi-A}
    \propto\,& \pi_V(\sigma \mid Z = \mathcal{A}(\sigma)) \cdot \theta^{\abs{\mathcal{A}(\sigma)}} \cdot \pi_{\mathcal{A}}(\mathcal{A}(\sigma)) \cdot \*1[\mathcal{A}(\sigma) \supseteq T]  \\
    \label{eq:resample-FD-event}
    =\,& \mu(\sigma) \cdot \theta^{\abs{\mathcal{A}(\sigma)}} \cdot \*1[\mathcal{A}(\sigma) \supseteq T],
\end{align}
where \eqref{eq:FD-pi-A} follows from \Cref{remark:general-field-dynamics}, which observes that the transition $T \to Z_{t+1}$ (via the (Down-$\mathcal{A}$) and (Up-$\mathcal{A}$) steps) simulates the vertex-field dynamics on $\pi_{\mathcal{A}}$.

For our choice of $\mathcal{A}$ above, $|\mathcal{A}(\sigma)|=\|\sigma\|$.
Therefore, by \eqref{eq:resample-FD-event}, 
        \begin{align*}
        \Pr{X_{t+1} = \sigma \mid T} 
        &\propto \mu(\sigma) \cdot \theta^{\|\sigma\|} \cdot \*1[\mathcal{A}(\sigma)\supseteq T] 
        \propto (\theta * \mu)^T(\sigma),
        \end{align*}
where $\mu^T$ denotes $\mu$ conditioned on all the events in $T$ occurring.

\item \textbf{Edge-field dynamics.}  Define $\vec{E} := \{(u,v), (v,u) \mid \{u,v\} \in E\}$ as the set of oriented edges, and  set
        \[\+A := \set{A_{uv} \mid (u,v) \in \vec{E}}, \quad \text{ where } A_{uv} := \set{\sigma \mid \sigma_u = 1 \land \sigma_v = 0},\quad \forall (u,v) \in \vec{E}.\]
Then, by the same argument as above (in particular, by \eqref{eq:resample-FD-event}), we obtain
        \begin{align*}
            \Pr{X_{t+1} = \sigma \mid T} 
            &\propto \mu(\sigma) \cdot \theta^{\abs{\mathcal{A}(\sigma)}} \cdot \*1[\mathcal{A}(\sigma)\supseteq T] \\
            &= \mu(\sigma) \cdot \theta^{\abs{E} - m(\sigma)} \cdot \*1[\mathcal{A}(\sigma)\supseteq T] \\
            &\propto (\theta^{-1} \otimes \mu)^T(\sigma),
        \end{align*}
where $m(\sigma)$ denotes the number of monochromatic (undirected) edges in $\sigma$.
    \item %
   \textbf{Swendsen--Wang dynamics.} For the ferromagnetic Ising model (i.e., $\beta = \gamma > 1$), 
we choose the tilt parameter $\theta = \beta^{-1}$ and set
\begin{align*}
        \mathcal{A} := \set{A_{\set{u,v}} \mid \set{u,v} \in E}, \quad \text{where } A_{\set{u,v}} := \set{\sigma \mid \sigma_u = \sigma_v}, \quad\forall \set{u,v} \in E.
    \end{align*}
Let $\mu$ be the Gibbs distribution of this Ising model. 
By the same argument as in \eqref{eq:resample-FD-event}, 
    \begin{align*}
        \Pr{X_{t+1} = \sigma \mid T} 
        &\propto \mu(\sigma) \cdot \theta^{\abs{\mathcal{A}(\sigma)}} \cdot \*1[\mathcal{A}(\sigma)\supseteq T] \\
        &= \mu(\sigma) \cdot \beta^{-m(\sigma)} \cdot \*1[\mathcal{A}(\sigma)\supseteq T] \\
        &\propto \lambda^{\abs{\sigma}}\cdot\*1[\mathcal{A}(\sigma)\supseteq T] \\
        &\propto \*1[\mathcal{A}(\sigma)\supseteq T] \cdot \prod_{C \in \kappa(V,E_T) \atop C\subseteq \sigma} \frac{\lambda^{\abs{C}}}{1 + \lambda^{\abs{C}}} \prod_{C \in \kappa(V,E_T) \atop C\not\subseteq \sigma} \tp{1 - \frac{\lambda^{\abs{C}}}{1 + \lambda^{\abs{C}}} },
    \end{align*}
    where $\kappa(V,E_T)$ denotes the set of connected components of the graph $(V,E_T)$, and 
$$E_T := \set{\set{u,v} \mid A_{\set{u,v}} \in T}.$$
In other words, conditioned on a fixed $T$, the new state $X_{t+1}$ is generated by selecting each connected component $C \in \kappa(V,E_T)$ independently with probability $\frac{\lambda^{|C|}}{1 + \lambda^{|C|}}$.
    \end{itemize}
\end{example}

\subsection{Localization scheme for edge-field dynamics}

Like the vertex-field dynamics, which corresponds to the negative-field localization scheme, the generalized event-field dynamics  also admits a localization interpretation. In particular, its associated denoising process can be defined as follows.

\begin{definition}[Event-field denoising process] \label{def:FD-event-denoising}
Let $\mu$ be a distribution on $\set{0,1}^V$, and let $\mathcal{A}$ be a collection of events. 
Let $\pi$ be the joint distribution on $\set{0,1}^{V \cup \mathcal{A}}$ defined in~\eqref{eq:dist-joint}.
Let $(Y_t')_{t \in [0,1]}$ be the vertex-field denoising process (see \Cref{def:field-dynamics-denoising}) for the marginal distribution $\pi_{\mathcal{A}}$, so that $\mathrm{Law}(Y_1') = \pi_{\mathcal{A}}$.

The \emph{event-field denoising process for $\mu$} with respect to $\+A$, denoted by $(Y_t)_{t \in [0,1]}$, is defined as follows:
\begin{itemize}
    \item Sample $Y_1 \sim \pi_V(\cdot \mid Y_1')$;
    \item For all $t \in [0,1)$, set $Y_t = Y_t'$.
\end{itemize}
That is, the process evolves according to the vertex-field denoising process on the event indicators, and at time $t=1$ projects onto the variables in $V$ via the conditional distribution.
\end{definition}

Note that the two processes $(Y_t)_{t\in [0,1]}$ and $(Y_t')_{t\in [0,1]}$ coincide for all $t < 1$, and differ only at the terminal projection step. 
In particular, the step that samples $Y_1 \sim \pi_V(\cdot \mid Y_1')$ corresponds to the wrapper step ($\mathcal{A} \to V$) in the event-field dynamics (see \Cref{def:event-FD-explicit}).

To apply the boosting argument in the localization framework (see \Cref{lem:boosting,lem:SS-to-AC-variance}) 
to establish the mixing time bound of the Glauber dynamics via the denoising process $(Y_t)_{t\in[0,1]}$, 
we need to establish  spectral stability of $(Y_t)_{t\in [0,1]}$.

Intuitively, since $(Y_t)_{t\in[0,1]}$ and $(Y_t')_{t\in [0,1]}$ are almost identical, 
the spectral stability of $(Y_t)_{t\in[0,1]}$ should be implied by that of $(Y_t')_{t\in [0,1]}$, %
as formalized below.

\begin{proposition}[Invariant of spectral stability] \label{prop:spectral-stable-compare}
    Let $(Y_t)_{t\in [0,1]}$ and $(Y_t')_{t \in [0,1]}$  be the two denoising processes defined in \Cref{def:FD-event-denoising}.
    For any $\theta \in [0,1)$, if $(Y_t')_{t\in [0,1]}$ is spectrally stable with rate $C$ at time $\theta$, then $(Y_t)_{t\in[0,1]}$ is also spectrally stable with rate $C$ at time $\theta$. 
\end{proposition}

In fact, \Cref{prop:spectral-stable-compare}  holds in greater generality for all denoising processes satisfying the abstract relation in \Cref{def:FD-event-denoising} (referred to as \emph{projected} denoising processes). We prove this more general statement in \Cref{sec:proof-spectral-stable-compare}. %

\paragraph{Edge-field denoising process.} Now, by instantiating the events  $\mathcal{A}$ as in \Cref{exm:event-FD}, we obtain the explicit construction of the denoising process associated with edge-field dynamics. 

\begin{definition}[Edge-field denoising process] \label{def:field-dynamics-on-edge}
Let $\mu$ be a distribution on $\set{0, 1}^V$. 
Let $G = (V, E)$ be a graph, and define $\vec{E} := \{(u,v), (v,u) \mid \{u,v\} \in E\}$.
Define 
\[
\+A := \set{A_{uv} \mid (u,v) \in \vec{E}}, \quad \text{ where } A_{uv} := \set{\sigma \mid \sigma_u = 1 \land \sigma_v = 0},\quad \forall (u,v) \in \vec{E}.
\]
Let $(Y_t)_{t\in[0,1]}$ be the event-field denoising process for $\mu$ with respect to this event family $\mathcal{A}$. %
We call the resulting process $(Y_t)_{t \in [0,1]}$ the \emph{edge-field denoising process} for $\mu$ on $G$.
\end{definition}

Recall the notation $\theta \otimes \mu$ for the distribution $\mu$ with tilted interaction, as defined in~\eqref{eq:def-otimes}. 
In particular, when $\mu$ is the Gibbs distribution of a two-spin system on a graph $G$ with parameters $(\beta, \gamma, \lambda)$, the distribution $\theta \otimes \mu$ corresponds to the Gibbs distribution of the spin system on the same $G$ with parameters $(\theta\beta, \theta\gamma, \lambda)$.

For an edge-field denoising process $(Y_t)_{t\in [0,1]}$, as defined in \Cref{def:field-dynamics-on-edge}, its posterior distribution is the target distribution $\mu$ with tilted interaction.

\begin{proposition}\label{prop:posterior-edge-field-dynamics}
Let $(Y_t)_{t\in [0,1]}$ be the edge-field denoising process for a distribution $\mu$ on a graph $G$.
Its posterior distribution satisfies
\begin{align*}
\-{Law}(Y_1\mid Y_t)=\frac{1}{1-t}\otimes\mu^{Y_t},
\end{align*}
where $\mu^{Y_t}$ denotes the conditional distribution of $\mu$ given that for every $(u,v)\in \vec{E}$ with $Y_t(uv)=1$, 
the event $A_{uv}$ occurs; that is, $\sigma_u=1$ and $\sigma_v=0$.
\end{proposition}

\begin{proof}
For $\sigma\in \set{0,1}^V$, let ${Z}(\sigma)\in\set{0,1}^{\mathcal{A}}$ indicate whether each $A\in\mathcal{A}$ occurs under $\sigma$.
Let $X \sim \mu$, and recall the joint distribution $\pi:= \mathrm{Law}(X, Z(X))$ on $\set{0,1}^{V\cup\mathcal{A}}$ defined in \eqref{eq:dist-joint}.

Let  $(Y'_t)_{t\in [0,1]}$ be the auxiliary vertex-field denoising process for $\pi_{\+A}$ used in \Cref{def:FD-event-denoising}   for constructing the edge-field denoising process $(Y_t)_{t\in [0,1]}$.

Now fix $\sigma\in \Omega(\pi_V)=\Omega(\mu)$ and $T\in \Omega(Y_t)$. 
By \Cref{def:FD-event-denoising,def:field-dynamics-on-edge}, we have
\begin{align*}
& \Pr{Y_1=\sigma\mid Y_t=T} \\
=\,& \Pr{Y_1=\sigma\mid Y_1' = Z(\sigma)} \cdot \Pr{Y_1' = Z(\sigma) \mid Y_t'=T} \\ %
\tag{by \eqref{eq:FD-posterior-calc}} 
\propto\,& \Pr{Y_1=\sigma\mid Y_1' = Z(\sigma)} \cdot (1-t)^{\abs{Z(\sigma)}} \cdot \pi_{\mathcal{A}}(Z(\sigma)) \cdot \*1[T \subseteq Z(\sigma)] \\
\propto\,& (1-t)^{-m(\sigma)}\cdot \pi(\sigma, Z(\sigma)) \cdot \*1[T \subseteq Z(\sigma)]\\
=\,& (1-t)^{-m(\sigma)}\cdot \mu(\sigma) \cdot \*1[T \subseteq Z(\sigma)],
\end{align*}
where %
the last equality follows since ${Z}(\sigma)$ is determined by $\sigma$.
By definition of $\theta\otimes\mu$ in \eqref{eq:def-otimes}, this implies $\-{Law}(Y_1\mid Y_t)=\frac{1}{1-t}\otimes\mu^{Y_t}$.
\end{proof}

\subsection{Spectral stability of edge-field dynamics}
The following notion of a \emph{second-order correlation matrix} provides a natural matrix criterion for the spectral stability of edge-field dynamics.

\begin{definition}[Second-order correlation matrix]\label{def-2nd-correlation-matrix}
Let $\mu$ be a distribution on $\set{0,1}^V$.
Let $G = (V, E)$ be a graph, with the set of oriented edges $\vec{E} := \{(u,v), (v,u) \mid \{u,v\} \in E\}$.

For $u,v \in V$, denote by $u\ol{v}$ the event 
\[
u\ol{v} := \set{\sigma \in \set{0,1}^V \mid \sigma_u = 1 \land \sigma_v = 0}.
\]
The \emph{second-order correlation matrix} $\pcor{\mu} \in \mathbb{R}^{\vec{E} \times \vec{E}}$ is defined by
\begin{align*}
\forall (u,v), (w,z) \in \vec{E}, \quad
\pcor{\mu}(uv,wz) :=
\begin{cases}
\mu(w\ol{z}\mid u\ol{v})-\mu(w\ol{z}) \quad &\text{if }\mu\tp{u\ol{v}} > 0,\\
0 \quad &\text{otherwise}.
\end{cases}
\end{align*}
\end{definition}

The following corollary follows from \Cref{def:field-dynamics-on-edge}, \Cref{prop:spectral-stable-compare}, 
and an application of \Cref{lem:spectral-stable-matrix-form} to the auxiliary vertex-field denoising process $(Y_t')_{t\in [0,1]}$ on the event variables used for constructing the edge-field denoising process $(Y_t)_{t\in [0,1]}$ in \Cref{def:FD-event-denoising}.
\begin{corollary}[Spectral stability of edge-field dynamics via second-order correlations]
\label{lem:pcor-to-stable}
Let $(Y_t)_{t\in[0,1]}$ be the edge-field denoising process for $\mu$ on a graph $G=(V,E)$.

Fix $\theta \in [0,1)$.
If for every $\Lambda\subseteq V$ and every feasible pinning $\tau \in \Omega(\mu_\Lambda)$, it holds that
\begin{align*}
\pSI{\frac{1}{1-\theta}\otimes\mu^{\tau}} \le C,
\end{align*}
then $(Y_t)_{t\in [0,1]}$ is spectrally stable with rate $C$ at time $\theta$.
\end{corollary}

Finally, we note that the eigenvalues of the second-order correlation matrix are closely related to coupling independence (see \Cref{def:CI}), which gives the following lemma.
\begin{lemma}[Spectral stability from coupling independence]
\label{lem:cond-stable-pFD}
Let $(Y_t)_{t\in[0,1]}$ be the edge-field denoising process for $\mu$ on a graph $G = (V,E)$ with maximum degree $\Delta$. 
Fix $\theta \in [0,1)$.
If $\frac{1}{1-\theta}\otimes\mu$ is $C$-coupling independent, then $(Y_t)_{t\in[0,1]}$ is $(2\Delta C)$-spectrally stable at time $\theta$.
\end{lemma}
\begin{proof}%
Fix an arbitrary feasible pinning $\tau$, and let $\pi:=\frac{1}{1-\theta}\otimes\mu^{\tau}$. 
Note that $\pi$ is $C$-coupling independent since $\frac{1}{1-\theta}\otimes\mu$ is $C$-coupling independent.
By \Cref{lem:pcor-to-stable}, it suffices to show,
\[
\pSI{\pi}\le 2\Delta C.
\]

For $v\in V$, we use $v$ (resp.\ $\ol{v}$) to denote the event $\{\sigma \mid \sigma_v = 1\}$ (resp.\ $\{\sigma \mid \sigma_v = 0\}$).
The maximum eigenvalue is upper bounded by the maximum row sum:
\begin{align}
\label{eq:pSI-rowsum}
\pSI{\pi} \le \max_{uv}\sum_{wz}\abs{\pcor{\pi}(uv,wz)}
= \max_{uv}\sum_{wz}\abs{\pi(w\ol{z}\mid u\ol{v})-\pi(w\ol{z})}.
\end{align}
By the triangle inequality, we have
\begin{align*}
&\abs{\pi(w\ol{z}\mid u\ol{v})-\pi(w\ol{z})}\\
\le\,
&\abs{\pi(w\ol{z}\mid u\ol{v})-\pi(w\ol{z}\mid u)} + \abs{\pi(w\ol{z}\mid u)-\pi(w\ol{z})}\\
=\,
&\pi(v \mid u) \cdot \abs{\pi(w\ol{z}\mid uv)-\pi(w\ol{z}\mid u\ol{v})} + \pi(\ol{u}) \cdot \abs{\pi(w\ol{z}\mid u)-\pi(w\ol{z}\mid\ol{u})}\\
\le\, 
&\abs{\pi(w\ol{z}\mid uv)-\pi(w\ol{z}\mid u\ol{v})} + \abs{\pi(w\ol{z}\mid u)-\pi(w\ol{z}\mid\ol{u})}.
\end{align*}
Summing over all $wz \in \vec{E}$ gives
\begin{align}
\label{eq:triangle-rowsum}
\sum_{wz}\abs{\pcor{\pi}(uv,wz)}
\le& \sum_{wz}\abs{\pi(w\ol{z}\mid uv)-\pi(w\ol{z}\mid u\ol{v})}+\sum_{wz}\abs{\pi(w\ol{z}\mid u)-\pi(w\ol{z}\mid \ol{u})}.
\end{align}
We bound the second term via the coupling lemma:
\begin{align}
\label{eq:rowsum-pair-hamming-distance}
\sum_{wz}\abs{\pi(w\ol{z}\mid u)-\pi(w\ol{z}\mid \ol{u})}
=& \sum_{wz}\inf_{\xi}\E[(X,Y)\sim\xi]{\abs{\*{1}[X \in w\ol{z}] - \*1[Y \in w\ol{z}]}}\notag\\
\le& \inf_{\xi}\E[(X,Y)\sim\xi]{\sum_{wz}\abs{\*{1}[X \in w\ol{z}] -  \*1[Y \in w\ol{z}] }},
\end{align}
where the infimum is over all couplings $\xi$ between $\pi^{u\gets 1}$ and $\pi^{u\gets 0}$.

For any $x,y\in\set{0,1}^V$, we have
\begin{align}
\label{eq:depair-hamming-distance}
&\sum_{wz\in\vec{E}}\abs{\*{1}[x \in w\ol{z}] - \*1[y \in w\ol{z}]}\notag\\
=\, 
&\sum_{\set{w,z}\in E} \tp{\abs{\*{1}[x \in w\ol{z}] -  \*1[y \in w\ol{z}]} + \abs{\*{1}[x \in z\ol{w}] - \*1[y \in z\ol{w}] } } \notag\\
\le\, 
&\sum_{\set{w,z}\in E} \tp{ \abs{\*{1}[x \in w] - \*1[ y \in w]} + \abs{\*{1}[x \in z] - \*1[y \in z] } }\notag\\
\le\, 
&\Delta\cdot \dist(x,y),
\end{align}
where $\dist(\cdot,\cdot)$ denotes the Hamming distance. The first inequality above is applied edge by edge: for every $\{w,z\}\in E$,
\[
\begin{aligned}
&\abs{\*{1}[x \in w\ol{z}] -  \*1[y \in w\ol{z}]}+\abs{\*{1}[x \in z\ol{w}] - \*1[y \in z\ol{w}] }\\
\le& \abs{\*{1}[x \in w] - \*1[ y \in w]}+\abs{\*{1}[x \in z] - \*1[y \in z] } .
\end{aligned}
\]
This is because, although both orientations $w\ol z$ and $z\ol w$ are counted,
a discrepancy at only one endpoint can cause at most one discrepancy among the two oriented indicators.

Combining \eqref{eq:rowsum-pair-hamming-distance}, \eqref{eq:depair-hamming-distance}, and the condition that $\pi$ is $C$-coupling independent,
\begin{align*}
\sum_{wz}\abs{\pi(w\ol{z}\mid u)-\pi(w\ol{z}\mid \ol{u})}
\le \inf_{\xi}\E[(x,y)\sim\xi]{\Delta\cdot \dist(x,y)}
\le \Delta C.
\end{align*}
The same argument yields
\begin{align*}
\sum_{wz}\abs{\pi(w\ol{z}\mid uv)-\pi(w\ol{z}\mid u\ol{v})} \le \Delta C.
\end{align*}
Substituting into \eqref{eq:pSI-rowsum} and \eqref{eq:triangle-rowsum}, we conclude
\begin{align*}
\pSI{\pi} &\le 2\Delta C. \qedhere
\end{align*}
\end{proof}

\ifthenelse{\boolean{Ent}}{%
\subsection{Invariant of \texorpdfstring{$\phi$}{phi}-entropic stability} \label{sec:proof-spectral-stable-compare}
}{%
\subsection{Invariant of spectral stability} \label{sec:proof-spectral-stable-compare}
}
In this section, we prove \Cref{prop:spectral-stable-compare}. 
In fact, the proposition holds in a more general setting. In the following definition, we generalize \Cref{def:FD-event-denoising} to a denoising process for general localization schemes.

\begin{definition}[Projected denoising process]\label{def:denoising-subset-of-vars}
Let $N$ be a finite ground set. 
Let $\pi$ be a distribution over $\set{0,1}^{N}$, and let $\Lambda \subseteq N$. 
Let $(Y_t^\Lambda)_{t \in [0,1]}$ be a denoising process with target distribution $\pi_{\Lambda}$.

For any $V \subseteq N$, the \emph{projected denoising process induced by $(Y_t^\Lambda)_{t \in [0,1]}$ onto $V$}, denoted by $(Y_t^V)_{t \in [0,1]}$, is a denoising process with target distribution $\pi_V$, defined as follows:
\begin{itemize}
    \item sample $Y_1^V \sim \pi_V(\cdot \mid Y_1^\Lambda)$;
    \item for all $t \in [0,1)$, set $Y_t^V = Y_t^\Lambda$.
\end{itemize}
For convenience, when $V = N$, we omit the superscript and write $Y_t := Y_t^{N}$.
\end{definition}

We note that \Cref{prop:spectral-stable-compare} is a specialization of the following more general result which ensures the invariant of $\phi$-entropic stability under projection of denoising processes. 

\ifthenelse{ \boolean{Ent} }{%

\begin{proposition}[Invariant of $\phi$-entropic stability under projection] \label{prop:spectral-stable-denoising-subset}
Let $(Y_t^\Lambda)_{t \in [0,1]}$ and $(Y_t^V)_{t \in [0,1]}$ respectively denote the original and projected denoising processes  in \Cref{def:denoising-subset-of-vars}.

For any $\theta \in [0,1)$, if $(Y_t^\Lambda)_{t \in [0,1]}$ is $\phi$-entropically stable with rate $C$ at time $\theta$, then $(Y_t^V)_{t \in [0,1]}$ is also $\phi$-entropically stable with rate $C$ at time $\theta$.
\end{proposition}
\begin{proof}
Without loss of generality, we assume $V = N$ and write $Y_t := Y_t^{N}$. 
This is without loss because any function $f_V : \Omega(\pi_V) \to \mathbb{R}$ can be naturally extended to a function $f : \Omega(\pi) \to \mathbb{R}$.

    Fix a function $f:\Omega(\pi) \to \mathbb{R}$ and $S$ with $\Pr{Y_\theta = S} = \Pr{Y_\theta^\Lambda = S} > 0$.
    Now, define a function $g:\Omega(\pi_\Lambda) \to \mathbb{R}$ by  
    \[
    g(y) := \E{f(Y_1) \mid Y_1^\Lambda = y}.
    \]

    Since $(Y_t^\Lambda)_{t\in [0,1]}$ is $\phi$-entropically stable with rate $C$ at time $\theta$, we have
    \begin{align*}
        \lim_{h\to 0^+} \frac{1}{h} \Ent[][\phi]{\E{g(Y_1^\Lambda) \mid Y_{\theta + h}^\Lambda} \mid Y_\theta^\Lambda = S} \leq \frac{C}{1-\theta} \Ent[][\phi]{g(Y_1^\Lambda) \mid Y_\theta^\Lambda = S}.
    \end{align*}
    Note that  
    \[
    \E{g(Y_1^\Lambda) \mid Y_{\theta + h}^\Lambda} = \E{\E{f(Y_1) \mid Y_1^\Lambda} \mid Y_{\theta + h}^\Lambda} = \E{f(Y_1) \mid Y_{\theta + h}^\Lambda}.
    \]
    Moreover, $Y_t = Y_t^\Lambda$ for $t \in [0,1)$.
Therefore, the inequality above reduces to
\begin{align*}
        \lim_{h\to 0^+} \frac{1}{h} \Ent[][\phi]{\E{f(Y_1) \mid Y_{\theta + h}} \mid Y_\theta = S} 
        &\leq \frac{C}{1-\theta} \Ent[][\phi]{\E{f(Y_1) \mid Y_1^\Lambda } \mid Y_\theta = S} \\
        &\leq \frac{C}{1-\theta} \Ent[][\phi]{f(Y_1) \mid Y_\theta = S},
    \end{align*}
where the last inequality follows from the law of total $\phi$-entropy.
\end{proof}

}{%

\begin{proposition}[Invariant of spectral stability under projection] \label{prop:spectral-stable-denoising-subset}
Let $(Y_t^\Lambda)_{t \in [0,1]}$ and $(Y_t^V)_{t \in [0,1]}$ respectively denote the original and projected denoising processes  in \Cref{def:denoising-subset-of-vars}.

For any $\theta \in [0,1)$, if $(Y_t^\Lambda)_{t \in [0,1]}$ is spectrally stable with rate $C$ at time $\theta$, then $(Y_t^V)_{t \in [0,1]}$ is also spectrally stable with rate $C$ at time $\theta$.
\end{proposition}
\begin{proof}
Without loss of generality, we assume $V = N$ and write $Y_t := Y_t^{N}$. 
This is without loss because any function $f_V : \Omega(\pi_V) \to \mathbb{R}$ can be naturally extended to a function $f : \Omega(\pi) \to \mathbb{R}$.

    Fix a function $f:\Omega(\pi) \to \mathbb{R}$ and $S$ with $\Pr{Y_\theta = S} = \Pr{Y_\theta^\Lambda = S} > 0$.
    Now, define a function $g:\Omega(\pi_\Lambda) \to \mathbb{R}$ by  
    \[
    g(y) := \E{f(Y_1) \mid Y_1^\Lambda = y}.
    \]

    Since $(Y_t^\Lambda)_{t\in [0,1]}$ is $\phi$-entropically stable with rate $C$ at time $\theta$, we have
    \begin{align*}
        \lim_{h\to 0^+} \frac{1}{h} \Var{\E{g(Y_1^\Lambda) \mid Y_{\theta + h}^\Lambda} \mid Y_\theta^\Lambda = S} \leq \frac{C}{1-\theta} \Var{g(Y_1^\Lambda) \mid Y_\theta^\Lambda = S}.
    \end{align*}
    Note that  
    \[
    \E{g(Y_1^\Lambda) \mid Y_{\theta + h}^\Lambda} = \E{\E{f(Y_1) \mid Y_1^\Lambda} \mid Y_{\theta + h}^\Lambda} = \E{f(Y_1) \mid Y_{\theta + h}^\Lambda}.
    \]
    Moreover, $Y_t = Y_t^\Lambda$ for $t \in [0,1)$.
Therefore, the inequality above reduces to
\begin{align*}
        \lim_{h\to 0^+} \frac{1}{h} \Var{\E{f(Y_1) \mid Y_{\theta + h}} \mid Y_\theta = S} 
        &\leq \frac{C}{1-\theta} \Var{\E{f(Y_1) \mid Y_1^\Lambda } \mid Y_\theta = S} \\
        &\leq \frac{C}{1-\theta} \Var{f(Y_1) \mid Y_\theta = S},
    \end{align*}
where the last inequality follows from the law of total variance.
\end{proof}

}

\section{Mixing Times of Glauber Dynamics at Criticality}\label{sec:rapid-mixing-at-criticality}\label{sec:5}
Let $\mu$ be the Gibbs distribution of an antiferromagnetic two-spin system with parameters $(\beta,\gamma,\lambda)$ on an $n$-vertex graph of maximum degree $\Delta \ge 3$ satisfying the critical uniqueness condition (\Cref{cond:critical}).

We prove rapid mixing of the Glauber dynamics for $\mu$, as stated in \Cref{thm:critical-mixing}, by distinguishing the following two regimes for a suitably chosen threshold $\beta^* := 1 - \frac{1+\sqrt{2}}{\Delta}$:
\begin{enumerate}
    \item \emph{Edge-tilting regime (large $\sqrt{\beta\gamma}$):} $\sqrt{\beta\gamma} \in \left[\beta^*, \frac{\Delta-2}{\Delta}\right]$;
    \item \emph{Vertex-tilting regime (small $\sqrt{\beta\gamma}$):} $\sqrt{\beta\gamma} < \beta^*$.
\end{enumerate}
Note that these cases are exhaustive: If $\sqrt{\beta\gamma} > \frac{\Delta-2}{\Delta}$, then $(\beta,\gamma,\lambda)$ is $(\Delta-1)$-unique with a positive slack, and hence the criticality asserted by \Cref{cond:critical} cannot hold.

The case of large $\sqrt{\beta\gamma}$ can be resolved via the edge-field localization. 
Specifically, we establish the following mixing time bound for the Glauber dynamics at criticality.

\begin{lemma}[Mixing at criticality via edge-tilting]
\label{lem:mixing-2spin-2}
Let $\mu$ be the Gibbs distribution of an antiferromagnetic two-spin system on an $n$-vertex graph of maximum degree $\Delta \ge 3$ satisfying \Cref{cond:critical}.
Suppose that 
$$\sqrt{\beta\gamma} \in \left[\bar{\beta}, \frac{\Delta-2}{\Delta}\right]$$ 
for some $0<\bar{\beta}<\frac{\Delta-2}{\Delta}$.
Then the Glauber dynamics for $\mu$ has mixing time
\begin{align*}
\Tmix = O\tp{\tp{\log\tp{\lambda+\lambda^{-1}}+\Delta\log\tp{2+\beta^{-1}}}K_{\beta,\gamma}(\beta\gamma)^{-\Delta}\-{e}^{2k+\eps} n^{2k+2+\eps}},
\end{align*}
where $k=\Delta(1-\bar{\beta})$, $\eps=2k\frac{(k+1)\Delta-k}{(\Delta-1)(\Delta-k)}$, and %
\begin{align}
\label{eq:def-K-beta-gamma}
K_{\beta,\gamma}:=
\begin{cases}
1, & \gamma\ge1,\\[1ex]
\displaystyle
\left(\gamma/\beta\right)^{\frac{\Delta^2(1-\sqrt{\beta\gamma})\gamma}{(\Delta-1)\sqrt{\beta\gamma}(1-\gamma)}}, & \gamma<1.
\end{cases}
\end{align}

\end{lemma} 
\begin{remark}
In particular, when $\sqrt{\beta\gamma} \ge \bar{\beta}= 1 - O\!\left(\frac{1}{\Delta}\right)$, we have $(\beta\gamma)^{-\Delta} = O(1)$, and the exponent parameters $k$ and $\eps$ in \Cref{lem:mixing-2spin-2} satisfy $k = O(1)$ and $\eps = O\!\left(\frac{1}{\Delta}\right)$.
Moreover, in this regime, $K_{\beta,\gamma}=1$ if $\gamma\ge1$, while for $\gamma<1$,
\[
K_{\beta,\gamma}
=\left(\frac{\gamma}{\beta}\right)^{O(\gamma/(1-\gamma))}
=O_{\beta,\gamma}(1),
\]
since $\Delta(1-\sqrt{\beta\gamma})=O(1)$.
\Cref{lem:mixing-2spin-2} therefore proves an $O_{\beta,\gamma}(n^{O(1)})$ mixing time at criticality in the regime $\sqrt{\beta\gamma} \in \left[1 - O\left(\frac{1}{\Delta}\right), \frac{\Delta-2}{\Delta}\right]$.
This includes the barrier case $[\frac{\Delta-2}{\Delta} - \epsilon, \frac{\Delta-2}{\Delta}]$ for small $\epsilon > 0$, 
which, as discussed in \Cref{sec:proof-main-overview}, posed a challenge for previous approaches.
We overcome this difficulty by introducing the edge-tilting field dynamics, a new localization approach that circumvents these  barriers. 
\end{remark}

On the other hand,
the case of small $\sqrt{\beta\gamma}$ can be resolved via the vertex-field localization, generalizing the analysis of the critical hardcore model~\cite{chen2025uniqueness}, and lifting the rapid mixing results for Glauber dynamics in the subcritical regime~\cite{chen2021rapid} to the critical threshold.
Its proof is based on the analysis of vertex-field dynamics and is deferred to \Cref{sec:mixing-2spin-1}.

\begin{lemma}[Mixing at criticality via vertex-tilting]
\label{lem:mixing-2spin-1}
Let $\mu$ be the Gibbs distribution of an antiferromagnetic two-spin system on an $n$-vertex graph of maximum degree $\Delta \ge 3$ satisfying \Cref{cond:critical}.
Suppose that 
$$\sqrt{\beta\gamma}\le \bar{\beta},$$
for some $\bar{\beta}\le \frac{\Delta-2.1}{\Delta}$.
Then the Glauber dynamics for $\mu$ has mixing time
\begin{align*}
\Tmix = O\tp{\tp{\log\tp{\lambda+\lambda^{-1}}+\Delta\log\alpha}n^{2\kappa+2}},
\end{align*}
where $\kappa=\sqrt{\frac{1-\bar{\beta}^2}{\beta_c^2-\bar{\beta}^2}}$, $\beta_c=\frac{\Delta-2}{\Delta}$, 
and %
$$\alpha=\begin{cases}\tp{2+\beta^{-1}}& \text{if }\beta>0,\\\tp{2+\gamma+\gamma^{-1}} & \text{if }\beta=0. \end{cases}$$ 
\end{lemma} 
 In particular, when $\sqrt{\beta\gamma}\le 1-\frac{c}{\Delta}$, the exponent parameter $\kappa$ in \Cref{lem:mixing-2spin-1}
satisfies  $\kappa=\sqrt{\frac{c}{c-2}}+O(\frac{1}{\Delta})$  and the mixing time is $T_{\-{mix}}=\Tilde{O}\tp{n^{2\sqrt{\frac{c}{c-2}}+2+O(\frac{1}{\Delta})}}$.

\begin{proof}[Proof of \Cref{thm:critical-mixing}]
Combining \Cref{lem:mixing-2spin-1,lem:mixing-2spin-2}, and choosing  $\bar{\beta}$ optimally as
    $$\bar{\beta}=\beta^*=1-\frac{1+\sqrt{2}}{\Delta},$$
the polynomial exponents in \Cref{lem:mixing-2spin-1} and \Cref{lem:mixing-2spin-2} match, giving
\[
2k + 2 +\eps = 4+2\sqrt{2} + O\!\left(\frac{1}{\Delta}\right)\quad\text{ and }\quad 2\kappa + 2 =4+ 2\sqrt{2} + O\!\left(\frac{1}{\Delta}\right).
\]

In the edge-tilting regime, $\sqrt{\beta\gamma}\ge\bar{\beta}$ implies
\[
(\beta\gamma)^{-\Delta}\le \bar{\beta}^{-2\Delta}=O(1),
\]
Moreover, $K_{\beta,\gamma}=1$ when $\gamma\ge1$; when $\gamma<1$, the definition of $K_{\beta,\gamma}$ in \eqref{eq:def-K-beta-gamma} and $\Delta(1-\sqrt{\beta\gamma})=O(1)$ imply
\[
K_{\beta,\gamma}
=\left(\frac{\gamma}{\beta}\right)^{O(\gamma/(1-\gamma))}
=O_{\beta,\gamma}(1).
\]
Finally,
\[
\log(\lambda+\lambda^{-1})+\Delta\log(2+\beta^{-1})=O_{\beta,\lambda}(\Delta).
\]
Thus \Cref{lem:mixing-2spin-2} gives the claimed $O_{\beta,\gamma,\lambda}\!\left(\Delta n^{4+2\sqrt2+O(1/\Delta)}\right)$ bound in the edge-tilting regime.

In the vertex-tilting regime, the factor
$\log(\lambda+\lambda^{-1})+\Delta\log\alpha$ in \Cref{lem:mixing-2spin-1} is $O_{\beta,\gamma,\lambda}(\Delta)$.
So this gives the same claimed bound in the vertex-tilting regime.
The two regimes are exhaustive, completing the proof.
\end{proof}

In the remainder of this subsection, we prove \Cref{lem:mixing-2spin-2}.
The key step is to establish approximate  conservation of variance  for the edge-field dynamics, stated as follows.

\begin{lemma}[Approximate conservation of variance for edge-field dynamics]
\label{lem:AC-var-edge-FD}
Let $\mu$ be the Gibbs distribution of an antiferromagnetic two-spin system on a graph $G$ satisfying \Cref{cond:critical}. 
Let $(Y_t)_{t\in[0,1]}$ be the edge-field denoising process for $\mu$ on $G$, as in \Cref{def:field-dynamics-on-edge}.
Then $(Y_t)_{t\in[0,1]}$ satisfies $R$-approximate conservation of variance up to time $1-\sqrt{\beta\gamma}$ with
\[
R=K_{\beta,\gamma}(\beta\gamma)^{-\Delta}(\-{e}n)^{\frac{2\Delta^2(1-\sqrt{\beta\gamma})}{(\Delta-1)\sqrt{\beta\gamma}}},
\]
where
$K_{\beta,\gamma}$ is defined in \eqref{eq:def-K-beta-gamma}.
\end{lemma}

Assuming \Cref{lem:AC-var-edge-FD}, we now prove \Cref{lem:mixing-2spin-2} and \Cref{thm:mixing-edge-field-dynamics}.

\begin{proof}[Proof of \Cref{lem:mixing-2spin-2}]
By \Cref{lem:AC-var-edge-FD}, the edge-field denoising process $(Y_t)_{t\in[0,1]}$ satisfies $R$-approximate conservation of variance up to time $1-\sqrt{\beta\gamma}$, where
\begin{align*}
R = K_{\beta,\gamma}(\beta\gamma)^{-\Delta}(\-{e}n)^{\frac{2\Delta^2(1-\sqrt{\beta\gamma})}{(\Delta-1)\sqrt{\beta\gamma}}}
\le K_{\beta,\gamma}(\beta\gamma)^{-\Delta}(\-{e}n)^{2k+\eps},
\end{align*}
with $k=\Delta(1-\bar{\beta})$ and $\eps=2k\frac{(k+1)\Delta-k}{(\Delta-1)(\Delta-k)}$. 
The inequality follows by assuming $\sqrt{\beta\gamma}\ge \bar{\beta}$.

By \Cref{prop:posterior-edge-field-dynamics}, when $\theta = 1 - \sqrt{\beta\gamma}$, for any fixed $Y_\theta$ we have
\[
\mathrm{Law}(Y_1 \mid Y_\theta) = \frac{1}{1-\theta} \otimes \mu^{Y_\theta},
\]
where $\mu^{Y_\theta} = \mu^\tau$ for a feasible pinning $\tau$ determined by $Y_\theta$.

With $\theta = 1 - \sqrt{\beta\gamma}$, the distribution $\frac{1}{1-\theta} \otimes \mu$ corresponds to a two-spin system with parameters $(\beta',\gamma',\lambda)$, where $\beta'=\frac{\beta}{\sqrt{\beta\gamma}}$ and $\gamma'=\frac{\gamma}{\sqrt{\beta\gamma}}$. 
Since $\beta'\gamma'=1$, the distribution $\frac{1}{1-\theta} \otimes \mu$ is a product distribution, and the same holds for $\frac{1}{1-\theta} \otimes \mu^\tau$. 
Hence, by \Cref{lem:AT-prod-dist}, $\mathrm{Law}(Y_1 \mid Y_\theta)$ satisfies $1$-approximate tensorization of variance.

Finally, applying \Cref{lem:boosting} together with \eqref{eq:AT-implies-mixing}, we obtain that the Glauber dynamics for $\mu$ has mixing time
\begin{align*}
\Tmix = nR\log\frac{1}{\mu_{\min}} = O\tp{\tp{\log\tp{\lambda+\lambda^{-1}}+\Delta\log\tp{2+\beta^{-1}}}K_{\beta,\gamma}(\beta\gamma)^{-\Delta}\-{e}^{2k+\eps} n^{2k+2+\eps}}.
\end{align*}
\end{proof}

\begin{proof}[Proof of \Cref{thm:mixing-edge-field-dynamics}]
Recall that the edge-field dynamics with tilt parameter $\theta$ corresponds to the  down-up chain $P_{1 \leftrightarrow 1-\theta}$ (see \Cref{def:LS-induce-down-up-walk})  induced by the edge-field denoising process defined in \Cref{def:field-dynamics-on-edge}.
The claimed mixing time bound then follows by combining \Cref{lem:AC-var-to-var-decay,lem:AC-var-edge-FD} with \eqref{eq:var-decay-implies-mixing}.
\end{proof}

\subsection{Approximate conservation of variance for edge-field dynamics}
\label{sec:conservation-var-edge-FD}

In this section, we prove \Cref{lem:AC-var-edge-FD}, establishing approximate conservation of variance for edge-field dynamics on critical antiferromagnetic two-spin systems.

By \Cref{lem:SS-to-AC-variance}, approximate conservation of variance follows from spectral stability, which is closely related to the notion of coupling independence defined in \Cref{def:CI}.

The following lemma establishes coupling independence under a tilted interaction.
In the irregular case, once the tilted parameter crosses $\gamma'=1$, we use the coupling-independence bound at the crossing point $\gamma'=1$.

\begin{lemma}\label{lem:CI-under-otimes}
Let $\mu$ be the Gibbs distribution of an antiferromagnetic two-spin system with parameters $(\beta,\gamma,\lambda)$ on a graph of maximum degree $\Delta \ge 3$ satisfying \Cref{cond:critical}.
Then the following bounds on coupling independence hold for the distribution $\frac{1}{1-\theta}\otimes\mu$.
\begin{itemize}
\item If $\gamma\ge1$, and $G$ is $\Delta$-regular, then for every $0<\theta\le1-\sqrt{\beta\gamma}$, the distribution $\frac{1}{1-\theta}\otimes\mu$ is $C(\theta)$-coupling independent with
\[
C(\theta)= 1+\frac{\Delta(1-\sqrt{\beta\gamma})}{(\Delta-1)\sqrt{\beta\gamma}\theta}.
\]
\item If $\gamma<1$, then for every $0<\theta\le1-\sqrt{\beta\gamma}$, the distribution $\frac{1}{1-\theta}\otimes\mu$ is $C(\theta)$-coupling independent with
\[
C(\theta)= 1+\frac{\Delta(1-\sqrt{\beta\gamma})}{(\Delta-1)\sqrt{\beta\gamma}\min\set{\theta,1-\gamma}}.
\]
\end{itemize}
\end{lemma}

We first prove \Cref{lem:AC-var-edge-FD} using \Cref{lem:CI-under-otimes}.
\begin{proof}[Proof of \Cref{lem:AC-var-edge-FD}]
Let $(Y_t)_{t\in[0,1]}$ be the edge-field denoising process for $\mu$ on $G$. Let
\[
H:=\frac{\Delta(1-\sqrt{\beta\gamma})}{(\Delta-1)\sqrt{\beta\gamma}}.
\]
By \Cref{lem:cond-stable-pFD,lem:CI-under-otimes}, for every $0<\theta\le 1-\sqrt{\beta\gamma}$, the process $(Y_t)_{t\in[0,1]}$ is spectrally stable at time $\theta$ with rate
\[
C(\theta)=2\Delta\cdot \min\set{n,1+\frac{H}{\theta}}
\]
when $\gamma\ge1$, and with rate
\[
C(\theta)=2\Delta\cdot \min\set{n,1+\frac{H}{\min\set{\theta,1-\gamma}}}
\]
when $\gamma<1$.
Here, we use $n$ as a trivial bound for coupling independence.
Combining this with \Cref{lem:SS-to-AC-variance}, it follows that $(Y_t)_{t\in[0,1]}$ satisfies $R$-approximate conservation of variance with
\[
\log R \le \int_{0}^{1-\sqrt{\beta\gamma}}\frac{C(t)}{1-t}\-{d}t.
\]

We estimate the integral according to the two cases in \Cref{lem:CI-under-otimes}.

If $\gamma\ge1$, let $t_0=H/(n-1)$. Without loss of generality, suppose $n$ is sufficiently large so that $t_0=\frac{H}{n-1}\le 1-\sqrt{\beta\gamma}$. Then
\begin{align}\label{eq:base-integral}
\frac{\log R}{2\Delta}
&\le \int_0^{1-\sqrt{\beta\gamma}}\min\set{n,1+\frac{H}{t}}\frac{\-{d}t}{1-t}\notag\\
&= \int_{0}^{t_0}\frac{n}{1-t}\-{d}t
+ \int_{t_0}^{1-\sqrt{\beta\gamma}}\tp{1+\frac{H}{t}}\frac{\-{d}t}{1-t}\notag\\
&= \log\frac{1}{\sqrt{\beta\gamma}}+(n-H-1)\log\frac{1}{1-t_0}
+H\log\frac{1-\sqrt{\beta\gamma}}{\sqrt{\beta\gamma}t_0}\notag\\
&= \log\frac{1}{\sqrt{\beta\gamma}}+(n-1)(1-t_0)\log\frac{1}{1-t_0}
+H\log\frac{(\Delta-1)(n-1)}{\Delta}\notag\\
&\le \log\frac{1}{\sqrt{\beta\gamma}}+H(\log n+1),
\end{align}
where the last inequality uses $(1-x)\log\frac{1}{1-x}\le x$ for $0<x<1$.
This yields the claimed bound in this case, since $K_{\beta,\gamma}=1$.

If $\gamma<1$, the only difference is that for $t\ge 1-\gamma$, the coupling-independence bound in \Cref{lem:CI-under-otimes} is frozen at its value at $t=1-\gamma$. Since $\beta\le\gamma$, we have $0\le 1-\gamma\le1-\sqrt{\beta\gamma}$. Set
\[
I_0:=\int_0^{1-\sqrt{\beta\gamma}}\min\set{n,1+\frac{H}{t}}\frac{\-{d}t}{1-t}.
\]
The calculation in \eqref{eq:base-integral} already gives
\[
I_0\le \log\frac{1}{\sqrt{\beta\gamma}}+H(\log n+1).
\]
Thus we have
\begin{align*}
\frac{\log R}{2\Delta}
&\le \int_0^{1-\sqrt{\beta\gamma}}
\min\set{n,1+\frac{H}{\min\set{t,1-\gamma}}}\frac{\-{d}t}{1-t}\\
&= \int_0^{1-\gamma}\min\set{n,1+\frac{H}{t}}\frac{\-{d}t}{1-t}
+\int_{1-\gamma}^{1-\sqrt{\beta\gamma}}\min\set{n,1+\frac{H}{1-\gamma}}\frac{\-{d}t}{1-t}\\
&\le I_0
+H\int_{1-\gamma}^{1-\sqrt{\beta\gamma}}\left(\frac{1}{1-\gamma}-\frac1t\right)\frac{\-{d}t}{1-t}\\
&\le \log\frac{1}{\sqrt{\beta\gamma}}+H(\log n+1)
+\frac{H\gamma}{2(1-\gamma)}\log\frac{\gamma}{\beta}.
\end{align*}
This yields
\[
R\le
K_{\beta,\gamma}(\beta\gamma)^{-\Delta}(\-{e}n)^{2\Delta H},
\]
as claimed.\qedhere
\end{proof}

To prove \Cref{lem:CI-under-otimes}, we introduce the following coupling independence result, whose proof is deferred to \Cref{subsec:ctrl.2}.

\begin{lemma}[Coupling independence within the uniqueness regime]\label{lem:optimal-CI}
Let $\delta \in (0,1)$. 
Fix any antiferromagnetic two-spin system with parameters $(\beta,\gamma,\lambda)$ on a graph $G$ with maximum degree $\Delta\ge 3$ satisfying:
\begin{itemize}
\item if $\gamma\le 1$, then $(\beta,\gamma,\lambda)$ is $(\Delta-1)$-unique with slack $\delta$;
\item if $\gamma>1$, then $(\beta,\gamma,\lambda)$ is $(\Delta-1)$-unique with slack $\delta$, and $G$ is $\Delta$-regular.
\end{itemize}
Then the Gibbs distribution $\mu$ is $\displaystyle \tp{1 + \frac{\Delta(1-\delta)}{(\Delta-1)\delta}}$-coupling independent.
Moreover, for every $1\le \theta\le(\beta\gamma)^{-1/2}$, the tilted distribution $\theta\otimes\mu$ is $\left(1+\frac{\Delta(1-\delta)}{(\Delta-1)\delta}\right)$-coupling independent.
\end{lemma}\label{sec:stable-2spin-2-SI} 

We also need the following estimate on the uniqueness slack created by edge tilting.

\begin{lemma}
\label{lem:interaction-to-unique-slack}
If $(\beta,\gamma,\lambda)$ is critically $d$-unique, then $(\theta\beta,\theta\gamma,\lambda)$ is $d$-unique with slack 
\[
\frac{(\theta-1)\sqrt{\beta\gamma}}{1-\sqrt{\beta\gamma}}
\]
for $1 \le \theta \le 1/\sqrt{\beta\gamma}$.
\end{lemma}

\begin{proof}%
For $(\beta,\gamma,\lambda)$ that is critically $d$-unique  and  $1\le\theta\le 1/\sqrt{\beta\gamma}$, let $\hat{\beta}=\hat{\beta}(\theta)=\theta\sqrt{\beta\gamma}, \hat{\lambda}=\lambda(\beta/\gamma)^{(d+1)/2}$. Define
\begin{align*}
\delta=\delta(\theta)=1-\frac{d(1-\hat{\beta}^2)\hat{x}^*_d}{(\hat{\beta}\hat{x}^*_d+1)(\hat{x}^*_d+\hat{\beta})},
\end{align*}
where $\hat{x}^*_d$ is the unique solution to the equation
\begin{align*}
\hat{x}^*_d=\hat{\lambda}\left(\frac{\hat{\beta}\hat{x}^*_d+1}{\hat{x}^*_d+\hat{\beta}}\right)^d.
\end{align*}
It is straightforward to verify that $(\theta\beta,\theta\gamma,\lambda)$ is $d$-unique with slack $\delta(\theta)$. Hence, it suffices to show that  $\delta(\theta)\ge\frac{(\theta-1)\hat{\beta}(1)}{1-\hat{\beta}(1)}\,(1\le \theta\le 1/\sqrt{\beta\gamma})$. 

If $\hat{\lambda}>1$, we apply the spin-flip transformation to the symmetrized system, replacing $(\hat{\beta},\hat{\beta},\hat{\lambda})$ by $(\hat{\beta},\hat{\beta},\hat{\lambda}^{-1})$ and $\hat{x}^*_d$ by $1/\hat{x}^*_d$.
This transformation preserves the uniqueness slack, since
\[
\frac{d(1-\hat{\beta}^2)x}{(\hat{\beta}x+1)(x+\hat{\beta})}
=
\frac{d(1-\hat{\beta}^2)x^{-1}}{(\hat{\beta}x^{-1}+1)(x^{-1}+\hat{\beta})}.
\]
Thus, without loss of generality, we may assume $\hat{\lambda}\le 1$.
Then the fixed point of the symmetric recursion satisfies $\hat x_d^*\le 1$.
Let $\hat{x}^*_d=\tanh\eta,\hat{\beta}=\tanh\zeta$. Then we have
\begin{align}
\label{eq:tanh-eq1}
1-\delta = \frac{d(1-\tanh^2\zeta)\tanh\eta}{(\tanh\zeta\tanh\eta+1)(\tanh\zeta+\tanh\eta)}
= \frac{d\sinh(2\eta)}{\sinh(2\eta+2\zeta)},
\end{align}
\begin{align}
\label{eq:tanh-eq2}
\hat{\lambda} = \tanh\eta\left(\frac{\tanh\eta+\tanh\zeta}{\tanh\eta\tanh\zeta+1}\right)^d
= \tanh\eta\tanh^d(\eta+\zeta).
\end{align}
Taking log-derivatives on both sides of \eqref{eq:tanh-eq2}, we have
\begin{align*}
0 = \frac{2\-{d}\eta}{\sinh(2\eta)}+2d\frac{\-{d}\eta+\-{d}\zeta}{\sinh(2\eta+2\zeta)}
= \frac{2}{\sinh(2\eta)}(\-{d}\eta+(1-\delta)(\-{d}\eta+\-{d}\zeta)),
\end{align*}
which implies that
\begin{align}
\label{eq:tanh-eta-zeta}
\-{d}\eta = -\frac{1-\delta}{2-\delta}\-{d}\zeta.
\end{align}
Similarly, taking log-derivatives on both sides of~\eqref{eq:tanh-eq1}, we have
\begin{align*}
-\frac{\-{d}\delta}{1-\delta} = 2\frac{\-{d}\eta}{\tanh(2\eta)}-2\frac{\-{d}\eta+\-{d}\zeta}{\tanh(2\eta+2\zeta)}.
\end{align*}
Thus we have
\begin{align*}
\-{d}\delta &= -2(1-\delta)\left(\frac{\-{d}\eta}{\tanh(2\eta)}-\frac{\-{d}\eta+\-{d}\zeta}{\tanh(2\eta+2\zeta)}\right)\\
&= \frac{2(1-\delta)}{2-\delta}\left(\frac{1-\delta}{\tanh(2\eta)}+\frac{1}{\tanh(2\eta+2\zeta)}\right)\-{d}\zeta\tag{by \eqref{eq:tanh-eta-zeta}}\\
&= \frac{2(1-\delta)}{(1-\hat{\beta}^2)(2-\delta)}\left(\frac{1-\delta}{\tanh(2\eta)}+\frac{1}{\tanh(2\eta+2\zeta)}\right)\-{d}\hat{\beta}\tag{by relation $\hat{\beta}=\tanh\zeta$}.
\end{align*}
Moreover, $0 \le \tanh x\le 1$ for any $x\in\={R_{+}}$, and $\eta,\zeta \ge 0$. Therefore, we have
\begin{align*}
\frac{\delta'(\theta)}{\hat{\beta}'(\theta)}=\frac{\-{d}\delta}{\-{d}\hat{\beta}} \ge \frac{2(1-\delta)}{1-\hat{\beta}^2}.
\end{align*}
Integrating this expression on $[1,\theta]$, we obtain
\begin{align*}
\int^{\theta}_{1}\frac{\delta'(\theta)\-{d}\theta}{1-\delta} \ge \int^{\theta}_{1}\frac{2\hat{\beta}'(\theta) \-{d}\theta}{1-\hat{\beta}^2},
\end{align*}
which implies that
\begin{align*}
\delta(\theta) \ge \frac{2(\theta-1)\hat{\beta}(1)}{(1+\hat{\beta}(\theta))(1-\hat{\beta}(1))} \ge \frac{(\theta-1)\hat{\beta}(1)}{1-\hat{\beta}(1)}.
\end{align*}
The last inequality holds because $\hat{\beta}(\theta)=\theta\sqrt{\beta\gamma}\le 1$ for $\theta\le 1/\sqrt{\beta\gamma}$.
\end{proof}

\begin{proof}[Proof of \Cref{lem:CI-under-otimes}]
Let $r=\sqrt{\beta\gamma}$ and $s=(1-\theta)^{-1}$. First consider the case $\gamma\ge1$.
By \Cref{lem:interaction-to-unique-slack}, the tilted system $s\otimes\mu$ is $(\Delta-1)$-unique with slack at least
\[
\delta_s=\frac{(s-1)r}{1-r}
=\frac{\theta r}{(1-\theta)(1-r)}.
\]
Moreover, its tilted parameter $s\gamma$ is at least $1$, and by \Cref{cond:critical} the graph is $\Delta$-regular.
Applying \Cref{lem:optimal-CI} to $s\otimes\mu$ gives coupling independence at most
\[
1+\frac{\Delta(1-\delta_s)}{(\Delta-1)\delta_s}
\le
1+\frac{\Delta(1-r)}{(\Delta-1)r\theta},
\]
which is the desired bound in the regular case.

Now suppose $\gamma<1$.
If $s\gamma \le 1$, the same argument applies without any regularity assumption, because the tilted parameter remains below $1$.
Thus $s\otimes\mu$ is coupling independent with constant at most
\[
1+\frac{\Delta(1-\delta_s)}{(\Delta-1)\delta_s}
\le
1+\frac{\Delta(1-r)}{(\Delta-1)r\theta},
\]
which is the desired bound since $\min\{\theta,1-\gamma\}=\theta$ in this case.

It remains to handle $s\gamma > 1$, where the tilt has crossed the boundary $\gamma'=1$.
Set
\[
s_0=\frac1\gamma
\quad\text{and}\quad
\mu_0=s_0\otimes\mu .
\]
Then the tilted parameter of $\mu_0$ satisfies $s_0\gamma=1$, and \Cref{lem:optimal-CI} does not require regularity in the case $\gamma\le1$.
By \Cref{lem:interaction-to-unique-slack}, $\mu_0$ is $(\Delta-1)$-unique with slack at least
\[
\delta_0
=\frac{(s_0-1)r}{1-r}
=\frac{(1-\gamma)r}{\gamma(1-r)}.
\]
Since $s\le1/r$, the further tilt factor $q_0=s/s_0=s\gamma$ satisfies
\[
1\le q_0\le \frac{\gamma}{r}
=\bigl((s_0\beta)(s_0\gamma)\bigr)^{-1/2}.
\]
Therefore, the ``moreover'' part of \Cref{lem:optimal-CI}, applied to the base distribution $\mu_0$, implies that
\[
s\otimes\mu=q_0\otimes\mu_0
\]
is coupling independent with constant at most
\[
1+\frac{\Delta\gamma(1-r)}{(\Delta-1)r(1-\gamma)}
\le
1+\frac{\Delta(1-r)}{(\Delta-1)r(1-\gamma)}.
\]
Since $s\gamma\ge1$ is equivalent to $\theta\ge1-\gamma$, this is precisely the claimed bound with $\min\{\theta,1-\gamma\}=1-\gamma$.
\end{proof}

\section{Tight Analysis of Spectral Independence} 
\label{sec:new-CI}

Following~\cite{chen2020contraction}, spectral independence on general graphs can be reduced to bounding total influence on trees.
Specifically, for a two-spin system with parameters $(\beta, \gamma, \lambda)$ on a graph $G$ with maximum degree $\Delta$, the total influence of any vertex $v$ on all vertices is bounded by the total influence of the self-avoiding walk (SAW) tree rooted at $v$. 
Exploiting the recursive structure of trees, this bound can be analyzed via contraction of the tree recursion.

Let $T = (V,E)$ be a tree with maximum degree $\Delta$ and root $r$, and let $\mu$ be the Gibbs distribution of a two-spin system with parameters $(\beta, \gamma, \lambda)$ on $T$. 
For a vertex $v$, let $T_v$ denote the subtree of $T$ rooted at $v$, and let $\mu_{T_v}$ be the Gibbs distribution of the two-spin system with the same parameters on $T_v$.
The \emph{marginal ratio} at vertex $v$ is defined as
\begin{align*}
    R_v = R_{v,T_v} := \frac{\mu_{T_v}(\sigma_v = 1)}{\mu_{T_v}(\sigma_v = 0)}.
\end{align*}
The tree recursion for marginal ratios is given by
\begin{align} \label{eq:tree-recursion}
    R_r = \lambda \prod_{i=1}^d \frac{\beta R_i + 1}{R_i + \gamma},
\end{align}
where $R_i = R_{v_i,T_{v_i}}$ denotes the marginal ratio at the root $v_i$ of the subtree $T_{v_i}$ rooted at the $i$-th child of $r$.

For vertex $v$, let $\-{TI}_v:=\-{TI}^{\mu_{T_v}}_{v}$ denote the total influence (see \Cref{def:total-influence}) from $v$ to all vertices in the subtree $T_v$.
As shown in~\cite{chen2020contraction}, the total influence $\mathrm{TI}_r$ from the root to all vertices in $T$ satisfies the following tree recursion for total influences
\begin{align} \label{eq:TI-recursion}
    \mathrm{TI}_r = 1 + \sum_{i=1}^d  \frac{(1-\beta\gamma)R_i}{(\beta R_i + 1)(R_i + \gamma)}  \cdot \mathrm{TI}_i,
\end{align}
where $\mathrm{TI}_i=\-{TI}_{v_i}$ denotes the total influence from the $i$-th child $v_i$ of $r$ to all vertices in the subtree $T_{v_i}$ rooted at $v_i$. 

To bound $\mathrm{TI}_r$, we introduce a \emph{control function}, namely a function $\Xi$ on $[0,+\infty]$ satisfying $\Xi(R_v) \ge \mathrm{TI}_v$ for all vertices $v$. 
Substituting this bound into~\eqref{eq:TI-recursion} and taking the supremum over all tree structures and external fields yields the sufficient condition
\begin{align} \label{eq:CI-control}
    \Xi\left(\lambda \cdot \prod_{i=1}^d \frac{\beta R_i+1}{R_i+\gamma}\right) \ge 1 + \sum_{i=1}^d \frac{(1-\beta\gamma)R_i}{(\beta R_i + 1)(R_i + \gamma)} \cdot \Xi(R_i).
\end{align}
Conversely, the existence of such a function $\Xi$ implies an upper bound on spectral independence (see \Cref{sec:control-SI}).
Combined with a carefully constructed choice of $\Xi$ (see \Cref{subsec:ctrl.2}), this yields the improved bounds on spectral independence stated in \Cref{thm:optimal-SI}.
Finally, a matching lower bound is proved in \Cref{subsec:lower-bound}, showing the tightness of these bounds.

\subsection{Control function for spectral independence}\label{sec:control-SI}
We now formally define the control function for spectral independence.
\begin{definition}[Control function]\label{def:control-function}
Let $\mu$ be the Gibbs distribution of a two-spin system  with parameters $(\beta,\gamma,{\lambda})$ on a graph $G = (V,E)$ with maximum degree $\Delta$.
For each $v \in V$, let $\Delta_v$ denote the degree of $v$ in $G$.
A function $\Xi:[0,+\infty] \to [0,+\infty)$
is called a \emph{control function} for $\mu$ if, for every  $v \in V$, letting $d_v := \Delta_v - 1$, for all $x_1,\dots,x_{d_v} \in [0,+\infty]$,
    \begin{align}\label{eq:ctrl_functional}
        \Xi \tp{\lambda \cdot \prod_{i=1}^{d_v} \frac{\beta x_i+1}{x_i+\gamma}} \ge 1+\sum_{i=1}^{d_v} \frac{(1-\beta\gamma)x_i}{(\beta x_i+1)(x_i+\gamma)}\cdot \Xi(x_i). %
    \end{align}
\end{definition}
The following result shows that the existence of such a control function implies an upper bound on spectral independence.

\begin{theorem}\label{thm:CI-ctrl-main}
Suppose that $G = \mathbb{T}$ is a tree rooted at $r$, and $\Xi$ is a control function for $\mu$.
Then the total influence $\mathrm{TI}_r$ from the root $r$ to all vertices satisfies
\begin{align}\label{eq:764}
    \mathrm{TI}_r \le 1 + \Delta_r \cdot \max_{x\ge 0} \left\{\frac{(1-\beta\gamma)x}{(\beta x+1)(x+\gamma)  }  \cdot \Xi(x)\right\}.
\end{align}
Moreover, this  holds under any feasible pinning
with the root $r$ unpinned, since $0$ and $+\infty$ lie in the domain of $\Xi$.
\end{theorem}
\begin{proof}
We first claim that for any non-root and unpinned vertex $v \in V \setminus \{r\}$, the total influence $\mathrm{TI}_v$ of $v$ on its respective subtree is upper bounded by $\Xi(R_v)$, where $R_v=R_{v,T_v}$ is the marginal ratio of vertex $v$ in $T_v$. 

We prove this claim by structural induction. In the base case where $v$ is 
an unpinned leaf,
the claim follows immediately since $\mathrm{TI}_v = 1$ 
and $\Xi(R_v)=\Xi(\lambda) \ge 1$ given that $d_v=0$ in \eqref{eq:ctrl_functional}.
For the inductive step, consider a non-leaf, non-root 
and unpinned vertex $v$ with $d_v$ children. Let $R_i=R_{v_i,T_{v_i}}$ denote the marginal ratio of $v$'s $i$-th child $v_i$ in the subtree $T_{v_i}$ rooted at $v_i$. We follow the convention that $R_j=+\infty$ if some child $j$ is pinned $1$, and $R_j=0$ if $j$ is pinned $0$. By the tree recursion for total influence in \eqref{eq:TI-recursion} and the inductive hypothesis, we have
\begin{align} \label{eq:TI-recursion-step}
    \mathrm{TI}_v = 1 + \sum_{i=1}^{d_v} \frac{(1-\beta\gamma)R_i}{(\beta R_i + 1)(R_i + \gamma)} \cdot \mathrm{TI}_i \le 1 + \sum_{i=1}^{d_v} \frac{(1-\beta\gamma)R_i}{(\beta R_i + 1)(R_i + \gamma)} \cdot \Xi(R_i).
\end{align}

According to the properties of the control function $\Xi$,  the right-hand side of \eqref{eq:TI-recursion-step} is at most $\Xi(R_v)$, which completes the proof of the claim.

To conclude, we bound the total influence at the root $r$. By~\eqref{eq:TI-recursion} and our inductive claim,%
\begin{align*}%
    \mathrm{TI}_r &\le 1 + \sum_{i=1}^{\Delta_r} \frac{(1-\beta\gamma)R_i}{(\beta R_i + 1)(R_i + \gamma)} \cdot \Xi(R_i) \le  1 + \Delta_r \cdot \max_{x\ge 0} \left\{\frac{(1-\beta\gamma)x}{(\beta x+1)(x+\gamma)  } \cdot \Xi(x)\right\}. \qedhere
\end{align*}
\end{proof}

Together with~\cite{chen2020contraction}, the following corollary extends the total-influence bound from trees to general graphs.

\begin{corollary}\label{cor:SI-general-graphs}
Suppose $\Xi$ is a control function for the Gibbs distribution $\mu$ with parameters $(\beta,\gamma,\lambda)$ on a graph $G$ with maximum degree $\Delta$.
Then $\mu$ is $\rho$-spectrally independent, where
    \begin{align}\label{eq:SI-general-bound}
        \rho= \Delta \cdot \max_{x \ge 0} \left\{\frac{(1-\beta\gamma)x}{(\beta x+1)(x+\gamma)}  \cdot  \Xi(x)  \right\}.
    \end{align}
\end{corollary}
\begin{remark}
The discrepancy between~\eqref{eq:764} and~\eqref{eq:SI-general-bound} arises because the influence from a vertex $u$ to itself is $1$, whereas in the influence matrix $\Psi(u,u)$ is defined to be $0$.
\end{remark}

Combining the recursive coupling developed in~\cite[Lemma~39]{chen2024rapid} with the above framework, we obtain a corresponding upper bound for coupling independence on general graphs.

\begin{corollary}\label{cor:CI-general-graphs}
Suppose $\Xi$ is a control function for the Gibbs distribution $\mu$ with parameters $(\beta,\gamma,\lambda)$ on $G$ with maximum degree $\Delta$.
Then $\mu$ satisfies $C$-coupling independence, where
    \begin{align*}
        C = 1+\Delta \cdot \max_{x \ge 0} \left\{\frac{(1-\beta\gamma)x}{(\beta x+1)(x+\gamma)}  \cdot  \Xi(x)  \right\}.
    \end{align*}
\end{corollary}
The proofs of \Cref{cor:SI-general-graphs,cor:CI-general-graphs} are deferred to~\Cref{sec:missing-5}.   

\subsection{Construction of the control function}\label{subsec:ctrl.2}
We construct the control function $\Xi$ and use it to prove \Cref{thm:optimal-SI} and \Cref{lem:optimal-CI}.

Let $\mu$ be the Gibbs distribution of an antiferromagnetic two-spin system with parameters $(\beta,\gamma,\lambda)$ on a graph $G$ with maximum degree $\Delta \ge 3$, satisfying:
\begin{itemize}
    \item If $\gamma \le 1$, then $(\beta,\gamma,\lambda)$ is $(\Delta-1)$-unique with slack $\delta$.
    \item If $\gamma > 1$, then $(\beta,\gamma,\lambda)$ is $(\Delta-1)$-unique with slack $\delta$, and $G$ is $\Delta$-regular.
\end{itemize}

If $\lambda > (\gamma/\beta)^{\frac{\Delta}{2}}$, we swap the roles of $0$-spin and $1$-spin in the two-spin system. Specifically, this defines a new Gibbs distribution $\ol{\mu}$ on $\set{0,1}^V$ satisfying 
\[
\ol{\mu}(\sigma) = \mu((1-\sigma_v)_{v\in V}), \quad \forall \sigma \in \set{0,1}^V.
\]
This corresponds to the two-spin system with parameters
 $(\beta',\gamma',\lambda')=(\gamma,\beta,1/\lambda)$.
Note that this transformation preserves spectral independence.

Hence, without loss of generality, we may assume
\begin{align} 
    \label{eq:cond-anti-ferro-flip}
     &\beta\ge 0, \quad  \gamma>0, \quad \lambda >0, \quad \beta\gamma < 1, \quad \lambda \le (\gamma/\beta)^{\frac{\Delta}{2}};\\
    \label{eq:cond-regular-flip}
    &\text{If } \max\{\beta,\gamma\} > 1, \text{ then $G$ is $\Delta$-regular.}
\end{align}
We drop the original requirement $\beta \le \gamma$  in \eqref{eq:cond-anti-ferro} because we may have swapped $\beta$ and $\gamma$. We still have $\gamma > 0$, since $\lambda > (\gamma/\beta)^{\frac{\Delta}{2}}$ can only occur when $\beta > 0$.

Define $D := \Delta - 1$ and introduce the auxiliary functions
\begin{equation} \label{eq:def-aux-func}
    \psi(x) := \frac{(1-\beta\gamma)x}{(\beta x+1)(x+\gamma)}, \qquad 
    f(x) := \lambda \left( \frac{\beta x+1}{x+\gamma} \right)^D.
\end{equation}
Let $\hat{x}$ be the unique solution to the fixed-point equation $x = f(x)$. We assume
\begin{equation} \label{eq:766}
    \delta = 1 - D \cdot\psi(\hat{x}).
\end{equation}
\begin{remark}[Uniqueness condition with exact slackness $\delta$]
    More precisely, the $D$-uniqueness condition with slack $\delta$ only guarantees that  $\delta \le  1 - D \cdot\psi(\hat{x})$.
Here, the $\delta$ in \eqref{eq:766} represents the \emph{exact} slackness with which the $D$-uniqueness is satisfied.  
Since the spectral and coupling independence bounds claimed in \Cref{thm:optimal-SI} and \Cref{lem:optimal-CI} are monotonically decreasing in the slackness $\delta$, it is safe to assume \eqref{eq:766} without loss of generality. 
\end{remark}

\begin{fact} \label{fact:range-hat-x}
Under the assumption \eqref{eq:cond-anti-ferro-flip}, the fixed point satisfies
 $\hat{x} \le \sqrt{\gamma/\beta}$.
\end{fact}

\begin{construction}[A good control function for two-spin systems] \label{cons:STD}
Define
    \begin{align*}
        \Xi(x) := \begin{cases}
            \frac{1}{\delta}, & x \in [0,\hat{x}], \\
            1 + \frac{D}{\delta}\cdot\psi(f^{-1}(x)), & x \in (\hat{x}, \lambda \gamma^{-D}], \\
            0,& x \in  (\lambda \gamma^{-D},+\infty],
        \end{cases}
    \end{align*}
where $\psi$ and $f$ are as in \eqref{eq:def-aux-func}, and  by convention set $\gamma/\beta = +\infty$ when $\beta = 0$.
\end{construction}

\Cref{cons:STD} defines a valid control function as required by \Cref{def:control-function}. 
Moreover, the right-hand side of \eqref{eq:SI-general-bound} can be conveniently bounded:
\begin{lemma}\label{lem:STD-max}
    Assume \eqref{eq:cond-anti-ferro-flip}.
Then, for all $x \ge 0$,
\begin{align*}
        \Xi(x) \cdot \psi(x) \le \frac{1-\delta}{\delta (\Delta-1)}.
    \end{align*}
\end{lemma}
\begin{lemma}\label{lem:STD-sol}
    Assume \eqref{eq:cond-anti-ferro-flip} and \eqref{eq:cond-regular-flip}.
    Then $\Xi$ is a valid control function for all edge-tilted measure $\theta\otimes \mu$; that is, for all $x_1,\dots,x_d \in [0,+\infty]$, with $\theta\in [1, (\beta\gamma)^{-\frac{1}{2}}],\ 0 \le d \le D$ (and $d = D$ if $\max\{\beta,\gamma\} > 1$),
    \begin{align*}%
        \Xi \tp{\lambda \cdot \prod_{i=1}^d \frac{\theta\beta x_i+1}{x_i+\theta\gamma}} \ge 1+\sum_{i=1}^d  \frac{(1-\beta\gamma\theta^2)x_i}{(\beta\theta x_i+1)(x_i + \gamma\theta)} \cdot \Xi(x_i).
    \end{align*}
\end{lemma}

\noindent
The upper bounds on spectral and coupling independence then follow easily:
\begin{proof}[Proof of \Cref{thm:optimal-SI} and \Cref{lem:optimal-CI}]
Assuming \Cref{lem:STD-max} and \Cref{lem:STD-sol}, \Cref{cor:SI-general-graphs} implies \Cref{thm:optimal-SI}, and \Cref{cor:CI-general-graphs} implies \Cref{lem:optimal-CI}.
\end{proof}

It remains to prove~\Cref{lem:STD-max,lem:STD-sol}.  
We will use the following property of  $\psi$:
\begin{fact} \label{fact:psi-monotone}
The function $\psi(x) = \frac{(1-\beta\gamma)x}{(\beta x + 1)(x+\gamma)}$ satisfies
\begin{itemize}[nosep]
    \item $\psi$ is monotonically increasing on $[0, \sqrt{\gamma/\beta}]$;
    \item $\psi$ is monotonically decreasing on $[\sqrt{\gamma/\beta}, +\infty)$.
\end{itemize}
\end{fact}

\begin{proof}[Proof of \Cref{lem:STD-max}]
For $x > \lambda \gamma^{-D}$, %
$\Xi(x) = 0$, and the claimed bound holds trivially.

It remains to consider the cases $x \in [0, \hat{x}]$ and $x \in (\hat{x}, \lambda \gamma^{-D}]$.  

    \paragraph{Case 1: $x \in [0, \hat{x}]$.} 
    By \Cref{fact:range-hat-x}, we have $\hat{x} \leq \sqrt{\gamma/\beta}$.
    In this range, the following bound holds:
    \begin{align*}
        \Xi(x) \cdot \psi(x) = \frac{\psi(x)}{\delta} \overset{(\star)}{\le} \frac{\psi(\hat{x})}{\delta} = \frac{1-\delta}{\delta (\Delta-1)},
    \end{align*}
    where $(\star)$ follows from \Cref{fact:psi-monotone}, and the last equality is due to \eqref{eq:766}.

    \paragraph{Case 2: $x \in (\hat{x}, \lambda \gamma^{-D}]$.} 
    In this case, we have
    \begin{align*}
        \psi(x) \cdot \Xi(x) = \left(1+ \frac{D}{\delta} \psi(f^{-1}(x))\right) \psi(x) \overset{(\ast)}{=} \left(1 + \frac{D}{\delta} \psi(t)\right) \psi(f(t)) =: H(t),
    \end{align*}
    where we change variables by setting $t = f^{-1}(x) \le \hat{x} \leq \sqrt{\gamma/\beta}$ in $(\ast)$. 
    Define function 
    \[G(x) := \frac{\gamma-\beta x^2}{(\beta x+1)(x+\gamma)}.\]
    Differentiating $H(t)$ with respect to $t$, we obtain
    \begin{align*}
        \frac{H'(t)}{H(t)} &= \frac{D(1-\beta\gamma)}{(\beta t+1)(t+\gamma)}\left(\frac{\gamma-\beta t^2}{\delta(\beta t+1)(t+\gamma)+D(1-\beta\gamma)t}-\frac{\gamma -\beta f(t)^2}{(\beta f(t)+1)(f(t)+\gamma)}\right)\\
        &= \frac{D(1-\beta\gamma)}{(\beta t+1)(t+\gamma)}\left(\frac{G(t)}{\delta+D\psi(t)}-G(f(t))\right)\\
        \tag{by $t\leq \hat{x} \leq \sqrt{\gamma/\beta}$ and \Cref{fact:psi-monotone}}
        &\ge \frac{D(1-\beta\gamma)}{(\beta t+1)(t+\gamma)}\left(\frac{G(t)}{\delta+D\psi(\hat{x})}-G(f(t))\right)\\
        \tag{by \eqref{eq:766}}
        &= \frac{D(1-\beta\gamma)}{(\beta t+1)(t+\gamma)} \left(G(t) - G(f(t))\right).
    \end{align*}
    Since $G(x)$ is monotonically decreasing for $x \geq 0$ and $t \le \hat{x} \le f(t)$, we have $G(t) \ge G(f(t))$, which implies $H'(t) \ge 0$. It follows that $H(t) \le H(\hat{x}) = \frac{1-\delta}{\delta(\Delta-1)}$.
\end{proof}

\begin{proof}[Proof of \Cref{lem:STD-sol}]
    The case $\beta\gamma\theta^2=1$ is trivial. Fix $1\le \theta< (\beta \gamma)^{-\frac{1}{2}}$, $d\le D$ and $x_1,x_2,\ldots,x_d\ge 0$ in the statement of \Cref{lem:STD-sol}.
    Write 
    \[h_\theta(x):=\frac{\theta \beta x+1}{x+\theta \gamma},\quad h:=h_1,\quad \psi_\theta(x)=\frac{(1-\beta\gamma\theta^2)x}{(\beta\theta x+1)(x+\gamma\theta)},\quad  \phi_\theta(s):=\psi_\theta(h_\theta^{-1}(e^s))
    .\]
    We define $y:=\lambda\prod_{i=1}^d h_\theta(x_i)\in [0,\lambda \gamma^{-D}]$. Thus
\(\Xi(y)\ne 0\) by the construction of \(\Xi\).

Note that $\psi_\theta\le \psi, \phi_\theta\le \phi_1$ by a direct calculation below,
\begin{gather*}
    \psi(x)-\psi_\theta(x)
    =
    \frac{x(\theta-1)\bigl(\gamma+\beta x^2+2\beta\gamma x
    +\beta\gamma\theta(\gamma+2x+\beta x^2)\bigr)}
    {(\beta x+1)(x+\gamma)(\beta\theta x+1)(x+\gamma\theta)}
    \ge0,\\
    \phi_1(s)-\phi_\theta(s)
    =
    \frac{(\theta-1)\bigl((1+\beta\gamma\theta)(\beta e^{-s}+\gamma e^s)
    -2\beta\gamma(1+\theta)\bigr)}
    {(1-\beta\gamma)(1-\beta\gamma\theta^2)} 
    \ge0.
\end{gather*}
$\beta e^{-s}+\gamma e^s\ge2\sqrt{\beta\gamma}$ and
$\theta\sqrt{\beta\gamma}\le1$ are used in the second inequality to show the factor
\begin{align*}
    &(1+\beta\gamma\theta)(\beta e^{-s}+\gamma e^s)
    -2\beta\gamma(1+\theta) \\
    \ge& 2\sqrt{\beta\gamma}(1+\beta\gamma \theta)-2\beta\gamma(1+\theta)\\
    =&2\sqrt{\beta\gamma}(1-\sqrt{\beta\gamma})(1-\sqrt{\beta\gamma}\theta)\ge 0 .
\end{align*}

There are two separate cases:

    \paragraph{Case 1. $y \le \hat{x}$:}
    In this case, by \Cref{lem:STD-max}, we have
    \begin{align*}
        1 + \sum_{i=1}^d \psi_\theta(x_i) \cdot \Xi(x_i)\le 1 + \sum_{i=1}^d \psi(x_i) \cdot \Xi(x_i) \le 1 + D \cdot \frac{1-\delta}{D \delta} = \frac{1}{\delta} = \Xi(y).
    \end{align*}
    \paragraph{Case 2. $y\in (\hat{x},\lambda \gamma^{-D}]$:}
    In this case, let $z = f^{-1}(y) \le \hat{x} \le \sqrt{\gamma/\beta}$. 

    After plugging in \eqref{eq:ctrl_functional},
    it remains to show that
    \begin{align*}
        \sum_{i=1}^d \psi_\theta(x_i) \cdot \Xi(x_i) \le \frac{D}{\delta}\psi(z).
    \end{align*}
     Let $s_i = \log h_\theta(x_i) $ and $s_z = \log h(z)\ge \log h(\sqrt{\gamma/\beta})=\log \sqrt{\beta/\gamma}$. Note that we  always have $\Xi(x_i) \le \Xi(\hat{x}) = \frac{1}{\delta},\ \phi_\theta\le \phi_1$.
    It suffices to show $(\star)$ holds in the following,
    \begin{align*}
        \delta \cdot \sum_{i=1}^d \psi_\theta(x_i) \cdot \Xi(x_i)\le \sum_{i=1}^d \psi_\theta(x_i)=\sum_{i=1}^d \phi_\theta(s_i)\le\sum_{i=1}^d \phi_1(s_i)\overset{(\star)}{\le} D\psi(z).
    \end{align*}    
      One can compute that $\phi_1(x) = \frac{1+\beta\gamma -\beta e^{-x}-\gamma e^x}{1-\beta\gamma}$, so $\phi_1$ is concave. By the definition of $y$, $z$ and $f^{-1}$, we have
    \begin{align*}
        &\sum_{i=1}^d \phi_1(s_i)\le \sum_{i=1}^d \phi_1(s_i)+\sum_{i=d+1}^D \phi_1(0) \le D \phi_1(\frac{1}{D}\sum_{i=1}^d s_i) = D\phi_1(s_z)=D\psi(z). 
    \end{align*}
    By the assumption on $d$, when $\max\set{\beta,\gamma} > 1$, we have $d = D$. When $\max\set{\beta,\gamma}\le1$ we have $\phi_1(0)=\frac{(1-\beta)(1-\gamma)}{1-\beta\gamma}\ge 0$. Therefore in either case $\sum_{i=d+1}^D \phi_1(0)\ge 0$. 

    Additionally, $\frac{1}{D}\sum_{i=1}^d s_i=s_z$ is due to the definition of $s_i,s_z$ and $y=\lambda h(z)^D=\lambda \prod_{i=1}^d h_\theta(x_i)$. \qedhere
\end{proof}

\subsection{A matching lower bound for spectral independence}\label{subsec:lower-bound}

We show that the spectral independence upper bound established in the preceding sections is tight. 
Specifically, the following lower bound holds for any antiferromagnetic parameter $(\beta, \gamma, \lambda)$ (i.e., satisfying \eqref{eq:cond-anti-ferro}) that is $(\Delta-1)$-unique with \emph{exact} slack $\delta$, i.e., satisfying \eqref{eq:766}.

\begin{theorem}\label{thm:SI-lowerbound}
Fix $\Delta \ge 3$ and $\delta \in (0,1)$. 
Assume that $(\beta, \gamma, \lambda)$ satisfies \eqref{eq:cond-anti-ferro} and is $(\Delta-1)$-unique with exact slack $\delta$. 
Then, there exists a sequence of graphs $\{G_n\}_{n \in \mathbb{N}}$ with maximum degree $\Delta$ such that the corresponding Gibbs distributions $\mu_n$ on $G_n$ satisfy
$$\lim_{n \to \infty} \lambda_{\max} \left( \Psi_{\mu_n} \right) = \frac{\Delta(1-\delta)}{(\Delta-1)\delta}.$$
\end{theorem}
The intuition behind~\Cref{thm:SI-lowerbound} is rooted in the behavior of the Gibbs measure on the infinite $\Delta$-regular tree. 
For a sequence of graphs that converges locally to the infinite $\Delta$-regular tree $\mathbb{T}_\Delta$, the local weak limit of the Gibbs measures $\{\mu_n\}$ corresponds to the unique Gibbs measure on $\mathbb{T}_\Delta$. 
A direct analysis of this measure shows that the total influence on $\mathbb{T}_\Delta$ is exactly 
 $\frac{\Delta(1-\delta)}{(\Delta-1)\delta}$.

We first introduce a construction of large-girth bipartite regular graphs.
\begin{fact}[\cite{MR3192385}]\label{fact:local-tree}
    There is a family of $\Delta$-regular bipartite graphs $G_n$ with girth $g_n=\omega_n(1)$.
\end{fact}

We next relate the influence between vertices on $G_n$ to that on the infinite tree $\mathbb{T}_\Delta$.

\begin{lemma}\label{lem:one-pair}
Fix $r \ge 1$. 
Let $\{G_n\}$ be the sequence of graphs constructed in~\Cref{fact:local-tree}, and let $\mu_n$ denote the Gibbs distribution on $G_n$ with antiferromagnetic parameters $(\beta, \gamma, \lambda)$ that are $(\Delta-1)$-unique with exact slack $\delta$.  
Then, for any vertices $u,v \in G_n$ with $\dist(u,v) \le r$ and $u \ne v$, 
    \begin{align*}
        \lim_{n \to \infty} \abs{\Psi_{\mu_n}(u,v)} = \tp{\frac{1-\delta}{\Delta-1}}^{\-{dist}(u,v)}.
    \end{align*}
\end{lemma}

\begin{proof}
    Recall from the definition, the influence $\Psi_{\mu}(u,v)$ between $u$ and $v$ satisfies
    \begin{align*}
        \Psi_{\mu}(u,v) = \frac{\mu(uv) - \mu(u) \mu(v)}{\mu(u)(1-\mu(u))},
    \end{align*} 
    where $\mu(uv)$ is the probability of both $u$ and $v$ being occupied. %
    Furthermore, for the Gibbs measure $\nu$ on infinite regular tree $\mathbb{T}_\Delta$, the absolute value of the influence from $\hat{u}$ to $\hat{v}$ with distance $\ell$ is exactly $\tp{\frac{1-\delta}{\Delta-1}}^{\ell}$. 
    Therefore, it suffices to show that
    \begin{align*}
        \lim_{n \to \infty} \mu_n(uv) = \nu(\hat{u} \hat{v}) \quad \text{and} \quad \lim_{n \to \infty} \mu_n(u) = \nu(\hat{u}).
    \end{align*}

    Let $S_n$ be the sphere of radius $h_n$ around $u$, where $h_n = \lfloor g_n/2 \rfloor - 2$. By the strong spatial mixing property~\cite{LLY13} and since the radius-$h_n$ ball  around $u$ in $G_n$ is isomorphic to the radius-$h_n$ ball in $\mathbb{T}_\Delta$, the probability of vertex $u$ being occupied satisfies
    \begin{align*}
        \lim_{n \to \infty} \mu_n(u) = \lim_{n \to \infty} \E[\tau \sim \mu_{n,S_n}]{\mu^{\tau}_n(u)} = \lim_{n \to \infty} \E[\tau \sim \mu_{n,S_n}]{\nu(\hat{u})} = \nu(\hat{u}),
    \end{align*}
    where $\nu$ is the Gibbs measure on $\mathbb{T}_\Delta$ and $\hat{u}$ is an arbitrary vertex on $\mathbb{T}_\Delta$.
    The same argument applies to $\mu_n(uv)$, and thus we complete the proof.
\end{proof}
We are now ready to prove~\Cref{thm:SI-lowerbound}.

\begin{proof}[Proof of \Cref{thm:SI-lowerbound}]

Fix $r \ge 1$. We claim that
\begin{align}\label{eq:lowerbound-r}
    \liminf_{n \to \infty} \lambda_{\max}\tp{\Psi_{\mu_n}} \ge \sum_{i = 1}^r \Delta \cdot (\Delta-1)^{i-1} \cdot \tp{\frac{1-\delta}{\Delta-1}}^i.
\end{align}
By taking $r \to \infty$ on both sides, we complete the proof of \Cref{thm:SI-lowerbound}.

It remains to verify~\eqref{eq:lowerbound-r}. 
Note that 
\[
\Psi_{\mu_n}=D_n^{-1}(\-{Cov}_{\mu_n}-D_n), \text{ where }D_n=\-{diag}\tp{\mu(u)(1-\mu(u))}_{u\in V}.
\]
Let $(L,R)$ be the bipartition of  $G_n$.
Define the test vector $\boldsymbol{\alpha} \in \{-1,+1\}^{L \cup R}$ indicating the partition with $\boldsymbol{\alpha}_u=\*1[u\in L]-\*1[u\in R]$.
By the monotonicity of the antiferromagnetic two-spin system on bipartite graphs, and by the Rayleigh quotient
\begin{align*}
    \lambda_{\max}\tp{\Psi_{\mu_n}} \ge \frac{\boldsymbol{\alpha} ^\intercal (\-{Cov}_{\mu_n}-D_n) \boldsymbol{\alpha} }{\boldsymbol{\alpha}^\intercal D_n\boldsymbol{\alpha} } =\frac{\sum_{u\ne v} \abs{\mu_n(uv)-\mu_n(u)\mu_n(v)}}{\sum_u \mu_n(u)(1-\mu_n(u))}.
\end{align*}
Hence, 
\begin{align*}
    \lambda_{\max}\tp{\Psi_{\mu_n}} &\ge \min_{u \in L \cup R}\frac{\sum_{v\ne u} \abs{\mu_n(uv)-\mu_n(u)\mu_n(v)}}{ \mu_n(u)(1-\mu_n(u))} \\
    &\ge \min_{u \in L \cup R} \sum_{i=1}^r \sum_{v \in S_u(i)} \abs{\Psi_{\mu_n}(u,v)},
\end{align*}
where $S_u(i)$ is the sphere of radius-$i$ around $u$.
By the local convergence of $\{G_n\}$ and~\Cref{lem:one-pair}, 
\[
   \lim_{n \to \infty} \sum_{v \in S_u(i)} \abs{\Psi_{\mu_n}(u,v)} = \Delta \cdot (\Delta-1)^{i-1} \cdot \tp{\frac{1-\delta}{\Delta-1}}^i.
\]
This completes the proof.
\end{proof}

\section{Rapid Mixing of the Swendsen--Wang Dynamics} \label{sec:rapid-mixing-SW-dynamics}
\newcommand{\muIsing}[1]{\mu_{\mathrm{Ising}#1}}
\newcommand{\ZIsing}[1]{Z^{\mathrm{Ising}}_{#1}}
\newcommand{\murc}[1]{\mu_{\mathrm{RC}#1}}
\newcommand{\Zrc}[1]{Z^{\mathrm{RC}}_{#1}}
\newcommand{\SStab}[1]{\mathrm{ES}(#1)}

In this section, we prove \Cref{thm:SW-main}, establishing optimal bounds on the spectral gap and mixing time of the Swendsen--Wang dynamics for the ferromagnetic Ising model with an external field on graphs of constant maximum degree.

Let $\muIsing{}$ be the Gibbs distribution of the ferromagnetic Ising model on a graph $G=(V,E)$ with edge activity $\beta \in [1,+\infty)$ and external field $\lambda \in [0,1]$. 
The Swendsen--Wang dynamics, introduced in~\cite{Swendsen1987}, is a global Markov chain for sampling from  $\muIsing{}$.

\begin{definition}[Swendsen--Wang dynamics]\label{def:SW-chain}
The \emph{Swendsen--Wang dynamics} is a Markov chain with stationary distribution $\muIsing{}$.
Define $p := 1 - \beta^{-1}$.
Given a current configuration $\sigma \in \set{0,1}^V$, the chain updates $\sigma$ to $\tau$ by applying  the following two operators:
\begin{itemize}
    \item $\mathbf{D}_{1 \to p}$: Let $M = M(\sigma) \subseteq E$ be the set of monochromatic edges under $\sigma$. Generate a random subset $T \subseteq M$ by independently removing each $e \in M$ with probability $1-p$.
    \item $\mathbf{U}_{p \to 1}$: Initialize $\tau = \emptyset$, and let $\kappa(V,T)$ denote the set of connected components of the graph $(V,T)$. For each $C \in \kappa(V,T)$, add $C$ to $\tau$ independently with probability $\frac{\lambda^{|C|}}{1 + \lambda^{|C|}}$.
\end{itemize}
\end{definition}
The Swendsen--Wang dynamics exploits the connection between the Ising model and its random-cluster representation. 
Let $p := 1 - \beta^{-1}$. 
The random-cluster model on graph $G$ with edge parameter $p \in [0,1)$ and vertex activity $\lambda \in [0,1]$ defines a distribution
\begin{align*}
    \forall S \subseteq E, \quad \murc{}(S) \propto p^{\abs{S}}(1-p)^{\abs{E\setminus S}} \prod_{C\in\kappa(V,S)}\tp{1 + \lambda^{\abs{C}}}.
\end{align*}
The Edwards--Sokal coupling formalizes this connection and implies the correctness of the Swendsen--Wang dynamics, as stated in the following proposition (see, e.g., \cite{feng2023swendsen-wang}).
\begin{proposition}[Edwards--Sokal coupling] \label{prop:ES-coupling}
Let $\mu_{\mathrm{RC}}$ be the distribution of the random-cluster model on a graph $G = (V,E)$ with edge parameter $p \in [0,1)$ and vertex activity $\lambda \in [0,1]$. 
Let $\mu_{\mathrm{Ising}}$ be the Gibbs distribution of the Ising model on the same graph $G$ with edge activity $\beta = (1-p)^{-1}$ and external field $\lambda \in [0,1]$.

For any partial pinning $T \subseteq E$, 
the following procedure constructs a sample $Y \sim \mu_{\mathrm{Ising}}^T$, 
i.e., $Y \sim \mu_{\mathrm{Ising}}$ conditioned on the event that $Y_u = Y_v$ for all $\{u,v\} \in T$:
    \begin{itemize}
    \item Initialize $Y = \emptyset$ and sample $S \sim \mu_{\mathrm{RC}}^T$, i.e., $S \sim \mu_{\mathrm{RC}}$ conditioned on $T \subseteq S$;
    \item For each $C \in \kappa(V,S)$, add $C$ to $Y$ independently with probability $\frac{\lambda^{|C|}}{1 + \lambda^{|C|}}$.
    \end{itemize}

Conversely, the following procedure produces a sample $X \sim \murc{}^T$:
    \begin{itemize}
    \item Initialize $X = T$. Sample $Y \sim \muIsing{}^T$, and for each monochromatic edge $e = \{u,v\} \in E$, independently add $e$ to $X$ with probability $p$.
    \end{itemize}

    In particular, 
\[
\muIsing{} \, \mathbf{D}_{1 \to p} = \murc{} 
\quad \text{and} \quad 
\murc{} \, \mathbf{U}_{p \to 1} = \muIsing{}.
\]
\end{proposition}

We reformulate the Swendsen--Wang dynamics  within the general framework of event-field dynamics, as in \Cref{exm:event-FD}.

\begin{definition}[Swendsen--Wang denoising process]\label{def:SW-denoising}
Let $\mu$ be the Gibbs distribution of the Ising model on a graph $G=(V,E)$. Define the family of events $\mathcal{A}$ by
\[
\mathcal{A} := \left\{ A_{\{u,v\}} \mid \{u,v\} \in E \right\}, \quad \text{where} \quad A_{\{u,v\}} := \{ \sigma \mid \sigma_u = \sigma_v \}.
\]
Let $(Y_t)_{t \in [0,1]}$ be the event-field denoising process for $\mu$ associated with this event family $\mathcal{A}$. 
We refer to the process $(Y_t)_{t \in [0,1]}$ as the \emph{Swendsen--Wang denoising process} for $\mu$.
\end{definition}
As with the edge-field denoising process, the posterior distribution $\mathrm{Law}(Y_1 \mid Y_t)$ corresponds to the conditional distribution with tilted interactions.
\begin{proposition} \label{lem:SW-posterior}
Let $\mu$ be the Gibbs distribution of the Ising model on a graph $G = (V,E)$, and let $(Y_t)_{t \in [0,1]}$ be the Swendsen--Wang denoising process for $\mu$. 
Then
    \begin{align*}
        \mathrm{Law}(Y_1 \mid Y_t) = (1-t) \otimes \mu^{Y_t},
    \end{align*}
where $\mu^{Y_t}$ denotes the conditional distribution of $\mu$ given that for every $\{u,v\} \in E$ with $Y_t(\{u,v\}) = 1$, the event $A_{\{u,v\}}$ occurs; that is, $\sigma_u = \sigma_v$.
\end{proposition}
The proof is identical to that of~\Cref{prop:posterior-edge-field-dynamics} and is omitted.

As observed in~\Cref{exm:event-FD}, the event-field dynamics for $\mu$ with respect to the event family $\mathcal{A}$ constructed in~\Cref{def:SW-denoising}, and with tilt parameter $\theta = \beta^{-1}$, 
coincides with the Swendsen--Wang dynamics. 
Therefore, by~\Cref{lem:SS-to-AC-variance}, the Swendsen--Wang denoising process $(Y_t)_{t \in [0,1]}$ satisfies $R$-approximate conservation of variance up to time $1-\theta$, where
\begin{align*}
    R = \exp\tp{\int_0^{1-\theta} \frac{C(t)}{1-t} \-d t},
\end{align*}
provided the Swendsen--Wang denoising process $(Y_t)_{t \in [0,1]}$ is spectrally stable with rate $C(t)$ for $t\in[0,1-\theta]$.
It therefore remains to establish the following spectral stability result.

\newcommand{\SSnoising}[1]{C_{\ref{lem:SI-SW}.\mathrm{a}}(#1)}
\newcommand{\SScrude}[1]{C_{\ref{lem:SI-SW}.\mathrm{b}}(#1)}
\begin{lemma}[Spectral stability of Swendsen--Wang denoising process]\label{lem:SI-SW}
Fix $\delta \in (0,1)$. 
Let $G = (V,E)$ be a graph with maximum degree $\Delta$, and let $\mu$ be the Gibbs distribution of the ferromagnetic Ising model on $G$ with edge activity $\beta \in [1,\infty)$ and external field $\lambda \in [0,1-\delta]$. 
The Swendsen--Wang denoising process $(Y_t)_{t \in [0,1]}$ for $\mu$ is spectrally stable with rate
    \begin{align*}
        C(t) = \min\tp{\SSnoising{t}, \SScrude{t}}, \qquad\forall t \in [0,1-\beta^{-1}];
    \end{align*}
    where
    \begin{align*}
        \SSnoising{t} := \frac{(1-t) \beta}{(1-t) \beta - 1} \cdot \frac{2}{\delta^2} 
        \quad \text{ and } \quad 
        \SScrude{t} := \frac{6\Delta}{\e \delta^3}.
    \end{align*}
\end{lemma}
We are now ready to prove the spectral gap part of \Cref{thm:SW-main}.
\begin{proof}[Proof of \Cref{thm:SW-main} (spectral gap)]
    By \Cref{lem:SI-SW} and \Cref{lem:SS-to-AC-variance}, the Swendsen--Wang denoising $(Y_t)_{t \in [0,1]}$ satisfies $\exp\tp{I}$-approximate conservation of variance up to time $\theta = 1-\beta^{-1}$, where
    \begin{align*}
        I = \int_0^{1-\beta^{-1}} \frac{C(t)}{1-t} \-d t &= \int_1^\beta \min\set{\frac{2}{\delta^2(s-1)} , \frac{6\Delta}{\e \delta^3 s}} \-d s\\
        &=
        \begin{cases}
            \frac{6\Delta}{\e \delta^3} \log \beta, & \beta \le \frac{3\Delta}{3\Delta - \e \delta},\\
            \frac{6\Delta}{\e \delta^3} \log \frac{3\Delta}{3\Delta - \e \delta} + \frac{2}{\delta^2} \log \frac{(\beta-1) (3\Delta - \e \delta)}{\e \delta} & \beta > \frac{3\Delta}{3\Delta - \e \delta}.
        \end{cases}
    \end{align*}
    In the first case, where $\beta \le \frac{3\Delta}{3\Delta-\e \delta}$, we have
    \begin{align*}
        I \le \frac{6 \Delta}{\e \delta^3} \log \frac{3\Delta}{3\Delta - \e \delta} \le \frac{6\Delta}{\delta^2 (3\Delta - \e \delta)} = O_\delta(1).
    \end{align*}
    In the second case, where $\beta > \frac{3\Delta}{3\Delta-\e \delta}$, similarly
    \[
        I \le O_\delta(1) + \frac{2}{\delta^2} \log \beta \Delta.
    \]
    The $\Omega((\beta \Delta)^{-2/\delta^2})$ bound on the spectral gap follows from \Cref{lem:AC-var-to-var-decay} and \eqref{eq:var-decay-implies-spectral-gap}.
\end{proof}

\newcommand{\ESnoising}[1]{C_{\ref{lem:EI-SW}.\mathrm{a}}(#1)}
\newcommand{\EScrude}[1]{C_{\ref{lem:EI-SW}.\mathrm{b}}(#1)}

Similarly, we establish the following entropic stability result to prove the optimal mixing-time bound for the Swendsen--Wang dynamics.
\begin{lemma}[Entropic stability for Swendsen--Wang denoising process]\label{lem:EI-SW}
    Fix $\delta \in (0, 0.01)$. 
    Let $G = (V,E)$ be a graph with maximum degree $\Delta$, and let $\mu$ be the Gibbs distribution of the ferromagnetic Ising model on $G$ with edge activity $\beta \in [1,\infty)$ and external field $\lambda \in [0,1-\delta]$. 
    The Swendsen--Wang denoising process $(Y_t)_{t \in [0,1]}$ for $\mu$ is entropically stable with rate
    \begin{align*}
        C(t) = \min\tp{\ESnoising{t}, \EScrude{t}}, \qquad \forall t \in [0,1-\beta^{-1}];
    \end{align*}
    where %
    \begin{align*}
        \ESnoising{t} := \frac{(1-t) \beta}{(1-t) \beta - 1} \cdot \frac{2}{\delta^2} 
        \quad \text{ and } \quad 
        \EScrude{t} := \begin{cases}
            \infty, & (1-t)\beta \ge 1+\frac{\delta^2}{\Delta^2}, \\
            321\frac{\Delta^2}{\delta^4}, & (1-t)\beta \le 1+\frac{\delta^2}{\Delta^2}.
        \end{cases}
    \end{align*}
\end{lemma}
\begin{proof}[Proof of \Cref{thm:SW-main} (mixing time)]
    By \Cref{lem:EI-SW} and \Cref{lem:SS-to-AC-variance}, the Swendsen--Wang denoising $(Y_t)_{t \in [0,1]}$ satisfies $\exp\tp{I}$-approximate conservation of entropy up to time $\theta = 1-\beta^{-1}$, where
    \begin{align*}
        I = \int_0^{1-\beta^{-1}} \frac{C(t)}{1-t} \-d t &\le \int_1^{1+\frac{\delta^2}{\Delta^2}} \frac{321\Delta^2}{\delta^4 s} \-d s +\int_{1+\frac{\delta^2}{\Delta^2}}^{\max\set{\beta,1+\frac{\delta^2}{\Delta^2}}}\frac{2\-ds}{\delta^2(s-1)} \\
        &\le
        \begin{cases}
            \frac{321}{\delta^2}, & \beta \le 1+\frac{\delta^2}{\Delta^2},\\
            \frac{321}{\delta^2} + \frac{2}{\delta^2} \log \frac{(\beta-1) \Delta^2}{\delta^2}, & \beta \ge 1+\frac{\delta^2}{\Delta^2}.
        \end{cases} \\
        &\le O_\delta(1)+\frac{4}{\delta^2}\log\beta\Delta
    \end{align*}
    By \Cref{lem:AC-var-to-var-decay}, the Swendsen--Wang transition $P$ satisfies entropy decay with rate $\kappa=\exp(-I)$.
    Directly applying \eqref{eq:ent-decay-implies-mixing} from the worst start would bring in $\log\log \frac{1}{(\mu_{\-{Ising}})_{\min}}$ dependence on $\lambda$.
    Instead, we first run one Swendsen--Wang step to get a warmer start, and then apply the entropy decay.
    Fix a configuration $\sigma$, let $D_\sigma:=\delta_\sigma \mathbf{D}_{1\to p}$ be the law of the open edges after the down-step, and write $Z_{\mathrm{RC}}$ for the partition function.
    For every $T\subseteq M(\sigma)$,
    \[
    \frac{D_\sigma(T)}{\murc{}(T)}
    =
    \beta^{|E|-|M(\sigma)|}
    \frac{Z_{\mathrm{RC}}}{\prod_{C\in\kappa(V,T)}(1+\lambda^{|C|})}
    \le \beta^{|E|}2^n.
    \]
    By data processing under the up-step and \Cref{prop:ES-coupling},
    \[
    \DKL{P(\sigma,\cdot)}{\muIsing{}}
    =
    \DKL{D_\sigma \mathbf{U}_{p\to 1}}{\murc{} \mathbf{U}_{p\to 1}}
    \le |E|\log\beta+n\log 2 .
    \]
    Applying the entropy decay for the remaining $t$ steps gives
    \[
    \DKL{P^{t+1}(\sigma,\cdot)}{\muIsing{}}
    \le (1-\kappa)^t\bigl(|E|\log\beta+n\log 2\bigr).
    \]
    Since $|E|\le \Delta n/2$ and $\exp(I)\le O((\beta\Delta)^{4/\delta^2})$, Pinsker's inequality gives the claimed mixing bound, with no dependence on $\lambda$.

\end{proof}

The remainder of this section is devoted to the proofs of \Cref{lem:SI-SW} and~\Cref{lem:EI-SW}.

\subsection{Spectral stability of the Swendsen--Wang dynamics} \label{sec:SI-SW}
\paragraph{Edge-to-edge correlation for Swendsen--Wang dynamics.}
In this subsection, we prove \Cref{lem:SI-SW}.
We define the correlation matrix for the Swendsen--Wang dynamics.
\begin{definition}[Swendsen--Wang correlation matrix]\label{def-SW-correlation-matrix}
Let $\mu$ be a distribution on $\set{0,1}^V$, and let $G = (V, E)$ be a graph.
For $\{u,v\} \in E$, we denote by $(u=v)$ the event
\[
(u=v) := \set{\sigma \in \set{0,1}^V \mid \sigma_u = \sigma_v}.
\]
The \emph{Swendsen--Wang correlation matrix} $\swcor{\mu} \in \mathbb{R}^{E \times E}$ is defined, for $\{u,v\}, \{w,z\} \in E$, by
\begin{align*}
 \quad
\swcor{\mu}(\set{u,v},\set{w,z}) :=
\begin{cases}
\mu(w=z\mid u=v)-\mu(w=z) \quad &\text{if }\mu\tp{u=v} > 0,\\
0 \quad &\text{otherwise}.
\end{cases}
\end{align*}
\end{definition}
The following corollary follows from \Cref{def:SW-denoising}, \Cref{prop:spectral-stable-compare}, and an application of \Cref{lem:spectral-stable-matrix-form} to the auxiliary vertex-field denoising process $(Y'_t)_{t \in [0,1]}$ on the event variables used for constructing the Swendsen--Wang denoising process $(Y_t)_{t \in [0,1]}$ in \Cref{def:FD-event-denoising}.
\begin{corollary}%
\label{lem:SW-matrix-to-stable}
Let $(Y_t)_{t \in [0,1]}$ be the Swendsen--Wang denoising process for $\mu$ on a graph $G=(V,E)$.
Fix $\theta \in [0,1)$. 
If for every $T \subseteq \mathcal{A}$,
\begin{align*}
\lambda_{\max}\tp{ \swcor{((1-\theta)\otimes\mu)^{T}} } \le C,
\end{align*}
then $(Y_t)_{t \in [0,1]}$ is spectrally stable with rate $C$ at time $\theta$.
\end{corollary}

\paragraph{Coupling independence for random clusters.}
We first introduce a known result on coupling independence for the random-cluster model.
Let $\mu$ be a distribution over $\set{0,1}^E$.
The {correlation matrix} $\mathrm{Cor}_\mu \in \mathbb{R}^{E \times E}$ of $\mu$ is defined as follows:
\begin{align*}
    \forall e,h \in E, \quad \mathrm{Cor}_\mu(e,h) := \begin{cases}
        \mu(h \mid e) - \mu(h),&\text{if }  \mu(e) \in (0,1), \\
        0, &\text{otherwise.}
    \end{cases}
\end{align*}

\begin{lemma}[\text{\cite[Lemma 4.1, 4.4, and 4.5]{chen2023near}}] \label{lem:SI-CZ23}
    Fix $\delta \in (0,1)$.
    The random-cluster model $\murc{}$ with parameters $p \in [0,1)$ and $\lambda \in [0,1-\delta]$ is $\frac{2}{\delta^2}$-coupling independent. 
        Consequently,
    $\lambda_{\max}(\cor{\murc{}}^T) \leq \frac{2}{\delta^2}$ for any feasible pinning $T \subseteq E$.
\end{lemma}

\paragraph{Spectral stability via down operator.}
We then establish a spectral stability bound by analyzing the following generalized down operator for field dynamics.
\begin{definition}\label{def:down-operator}
    Let $\*\theta \in [0,1]^m$.
We define ${D}_{1 \to \*\theta}$ to be the \emph{down operator} specified by $\*\theta$.
Specifically, ${D}_{1 \to \*\theta}$ is a Markov operator such that, given $S \subseteq [m]$, it generates a random subset $T \subseteq S$ by independently removing each element $i \in S$ with probability $1-\theta_i$.
\end{definition}
\begin{lemma}[Spectral stability under down operator]\label{lem:SI-reduction}
Let $\mu$ be a distribution over $\set{0,1}^m$, and let $\pi := \mu {D}_{1 \to \*\theta}$, where $\*\theta \in (0,1]^m$.
    Then,%
    \begin{align*}
        \lambda_{\max}(\cor{\mu}) &\leq \lambda_{\max}(\cor{\pi}) \cdot \max_{i \in [m]} \theta^{-1}_i.
    \end{align*}

\end{lemma}

\begin{proof}  
    A crucial fact from the definition of ${D}_{1\to \*\theta}$ is that, for every distribution $\nu$ on $\set{0,1}^m$,
    \begin{align} \label{eq:aux-pi-mu-marginal}
        \*m(\nu {D}_{1\to \*\theta}) = \*\theta \odot \*m(\nu), 
    \end{align}
    where we use $\odot$ to denote the entry-wise product and $\*m(\nu)$ denotes the barycenter of $\nu$.
    We define the following matrices
    \begin{align*}
        D_\mu := \mathrm{diag}(\*m(\mu)), \quad D_\pi := \mathrm{diag}(\*m(\pi)), \quad \text{and } D_{\*\theta} := \mathrm{diag}\set{\theta_i}_{i\in [m]}.
    \end{align*}
    According to \eqref{eq:aux-pi-mu-marginal}, it can be verified that %
    \begin{align*}
        D_\pi = D_\mu D_{\*\theta},
        \quad \text{and} \quad 
     \mathrm{Cov}_\pi=D_{\*\theta}\mathrm{Cov}_\mu D_{\*\theta}+D_\mu D_{\*\theta}(I-D_{\*\theta}).
    \end{align*}
    The term $D_\mu D_{\*\theta}(I-D_{\*\theta})$ is positive semidefinite; therefore,
        $\mathrm{Cov}_\pi \succeq D_{\*\theta} \mathrm{Cov}_\mu D_{\*\theta}$.

    We complete the proof using the Rayleigh quotient
    \begin{align*}
        \lambda_{\max}(\mathrm{Cor}_{\mu}) 
        &= \sup_{f\neq 0} \frac{\inner{f}{ \mathrm{Cov}_{\mu} f}}{\inner{f}{D_\mu f}} \\
        &\le \sup_{f\neq 0} \frac{\inner{D_{\*\theta}^{-1}f}{ \mathrm{Cov}_{\pi} D_{\*\theta}^{-1}f}}{\inner{D_{\*\theta}^{-1}f}{D_\pi D_{\*\theta}^{-1} f}} \cdot \frac{\inner{D_{\*\theta}^{-1} f}{D_{\mu} f}}{\inner{f}{D_\mu f}} \\
        &\leq \lambda_{\max}(\mathrm{Cor}_\pi) \sup_{f\neq 0} \frac{\inner{f}{D_{\*\theta}^{-1}D_\mu f}}{\inner{f}{D_\mu f}} \\
        &\leq \lambda_{\max}(\mathrm{Cor}_\pi) \cdot \max_i \theta_i^{-1}. \qedhere
    \end{align*}
\end{proof}
\begin{proof}[Proof of the bound $\SSnoising{t}$ in \Cref{lem:SI-SW}]
    Fix $t \in [0,1-\beta^{-1})$, and let $ p = p(t)$, where
    \[
     p(t) := 1 - ((1-t)\beta)^{-1}.
    \]
    By \Cref{lem:SW-posterior}, 
    \begin{align*}
        \mathrm{Law}(Y_1 \mid Y_t) = (1-t) \otimes \muIsing{}^{Y_t}.
    \end{align*}
Denote this distribution by $\nu$; it is the Gibbs distribution of an Ising model with edge activity $(1-t)\beta$, conditioned on the event $Y_t$.

    Let $\widetilde{\nu}$ be the distribution over $\{0,1\}^E$ obtained by first sampling $X \sim \nu$ and then taking the set of monochromatic edges.
    Let ${D}_{1 \to \*p}$ denote the down operator, where $p_e = p$ for edges $e \in E \setminus Y_t$, and $p_e = 1$ for edges $e \in Y_t$.
    Then, by \Cref{prop:ES-coupling}, we have
    \begin{align*}
        \widetilde{\nu} {D}_{1\to \*p} = \murc{}^{Y_t}.
    \end{align*}
Denote this distribution by $\pi$.
    Applying \Cref{lem:SI-reduction} to $\widetilde{\nu}$ and $\pi$, we obtain
    \begin{align*}
        \lambda_{\max}(\swcor{\nu}) = \lambda_{\max}(\cor{\widetilde{\nu}})
        &\leq \lambda_{\max}(\cor{\pi}) p^{-1} 
        \leq p^{-1} \cdot \frac{2}{\delta^2}
        = \frac{(1-t)\beta}{(1-t)\beta - 1} \cdot \frac{2}{\delta^2},
    \end{align*}
    where the last inequality follows from \Cref{lem:SI-CZ23}. 

    Combining this with \Cref{lem:SW-matrix-to-stable}, we conclude that the Swendsen--Wang denoising process exhibits spectral stability with rate 
    $\frac{(1-t)\beta}{(1-t)\beta -1} \cdot \frac{2}{\delta^2}$ at time $t$.
\end{proof}

\paragraph{Spectral stability via Edwards--Sokal coupling.}
Note that the bound $\SSnoising{t}$ becomes unbounded as $t$ approaches $1 - \beta^{-1}$.
Intuitively, this is because as $t \to 1 - \beta^{-1}$, we have $p = p(t) = 1 - ((1-t)\beta)^{-1} \to 0$, and the down operator ${D}_{1 \to p}$ erases too much information.
Therefore, an alternative approach is needed when $t$ is close to $1 - \beta^{-1}$.

Next, we prove an alternative bound $\SScrude{t}$ for this regime by exploiting the Edwards--Sokal coupling (see \Cref{prop:ES-coupling}).

\begin{proof}[Proof of the bound $\SScrude{t}$ in \Cref{lem:SI-SW}]
Let $(Y_t)_{t \in [0,1]}$ be the Swendsen--Wang denoising process for $\muIsing{}$, and let $S \subseteq E$ be an edge subset.
Fix $t \in [0,1-\beta^{-1}]$, and define
\[
\nu:=\mathrm{Law}(Y_1 \mid Y_t = S) = (1-t) \otimes \muIsing{}^{S}.
\]
Note that
\begin{align*}
    \lambda_{\max}(\swcor{\nu}) \leq \norm{\swcor{\nu}}_\infty \leq \max_{e \not\in Y_t} \inf_{\xi = \xi(e)} \E[(X,Y)\sim \xi]{\dist(Z(X),Z(Y))},
\end{align*}
where $\xi = \xi(e)$ ranges over all couplings between $(1-t) \otimes \muIsing{}^{S}$ and $(1-t) \otimes \muIsing{}^{S \cup \{e\}}$, $\mathrm{dist}(\cdot,\cdot)$ denotes the Hamming distance, and 
$Z(X) := \set{e=\{u,v\} \mid X_u = X_v} \in \set{0,1}^{E}$ 
denotes the set of monochromatic edges in configuration $X$.

Therefore, it suffices to construct a coupling $\xi$ with small expected distance.
We construct such a coupling using the Edwards--Sokal coupling (\Cref{prop:ES-coupling}) as follows:
    \begin{enumerate}
    \item Draw $(P,Q) \sim \mathcal{C}$, where $\mathcal{C}$ is a coupling between $\murc{}^{S}$ and $\murc{}^{S \cup \{e\}}$ as in \Cref{lem:SI-CZ23}, satisfying 
    \[
    \E[\mathcal{C}]{\mathrm{dist}(P,Q)} \le \frac{2}{\delta^2}.
    \]

    \item Suppose $P$ and $Q$ differ by exactly one edge. 
    Without loss of generality, assume $P = Q \setminus \{f\}$. 
    We construct a coupling of $X$ and $Y$ as follows:
    \begin{itemize}
        \item For each connected component $C \in \kappa(V,P) \cap \kappa(V,Q)$, include $C$ in both $X$ and $Y$ independently with probability $\frac{\lambda^{|C|}}{1+\lambda^{|C|}}$;
        \item For each $C \in \kappa(V,P) \setminus \kappa(V,Q)$, include $C$ in $X$ with probability $\frac{\lambda^{|C|}}{1+\lambda^{|C|}}$;
        \item For each $C \in \kappa(V,Q) \setminus \kappa(V,P)$, include $C$ in $Y$ with probability $\frac{\lambda^{|C|}}{1+\lambda^{|C|}}$.
    \end{itemize}

    \item In the general case where $P$ and $Q$ differ by multiple edges, we apply the path coupling method to extend the construction.
    \end{enumerate}

The correctness of the coupling $\xi$ follows from \Cref{prop:ES-coupling}. 
Note that $P$ and $Q$ differ by a single edge $f$, so the component families $\kappa(V,P)$ and $\kappa(V,Q)$ are either identical, or satisfy 
$\kappa(V,P) \setminus \kappa(V,Q) =  \{K_a,K_b\}$ and $\kappa(V,Q) \setminus \kappa(V,P) = \{K\}$.
    In either case, we have
    \begin{align*}
    \E{\dist(Z(X), Z(Y))} 
    &\leq \E{\Delta \cdot \tp{\lambda^{\abs{K_a}}\abs{K_a} + \lambda^{\abs{K_b}}\abs{K_b} + \lambda^{\abs{K}} \abs{K}}} \leq \frac{3\Delta}{\e\delta},%
    \end{align*}
where the last inequality uses $(1-\delta)^x x \le (\e\delta)^{-1}$ for $x \ge 1$.

In the general case with $\mathrm{dist}(P,Q) \ge 1$, the path coupling step introduces an extra factor  $\E{\mathrm{dist}(P,Q)} \le \frac{2}{\delta^2}$. Therefore,
    \begin{align*}
    \lambda_{\max}(\swcor{\nu}) \le  \E[(X,Y)\sim \xi]{\dist(Z(X),Z(Y))} \le \frac{6\Delta}{\e \delta^3}.
    \end{align*}
Combining this with \Cref{lem:SW-matrix-to-stable}, we conclude that the Swendsen--Wang  denoising process is spectrally stable with rate $\frac{6\Delta}{\e \delta^3}$.
\end{proof}

\subsection{Entropic stability of the Swendsen--Wang dynamics}
In this subsection, we prove \Cref{lem:EI-SW}, which bounds the entropic stability of the Swendsen--Wang dynamics. 
Following the approach for spectral stability, we bound entropic stability via two routes: first via the down operator, and second via spectral stability combined with marginal bounds.

To simplify the exposition and avoid explicit reference to the denoising process, we introduce the notion of entropic stability for a distribution.
\begin{definition}
    Let $\mu$ be a distribution over $\{0,1\}^m$.
    We say that $\mu$ is \emph{$C$-entropically stable} if for all $\nu \ll \mu$, 
  \begin{align*}
    \sum_{i=1}^m \nu(i) \log \frac{\nu(i)}{\mu(i)} - (\nu(i) - \mu(i)) \le C \cdot D_{\mathrm{KL}}(\nu \parallel \mu),
  \end{align*}
where $\mu(i)$ and $\nu(i)$ denote the marginal probabilities of $i$ being occupied under $\mu$ and $\nu$, respectively.
Let $\SStab{\mu}$ denote the minimum $C \ge 0$ such that $\mu$ is $C$-entropically stable.
\end{definition}

\begin{remark}[A sufficient condition for entropic stability]
Consider the Gibbs distribution $\mu_{\mathrm{Ising}}$ of the Ising model on a graph $G = (V, E)$.
Let $\mu_E$ denote the distribution over $\set{0,1}^E$ obtained by sampling $X \sim \mu_{\mathrm{Ising}}$ and returning the set of monochromatic edges in $X$.

By \Cref{prop:spectral-stable-compare,def:SW-denoising,def:FD-event-denoising}, the Swendsen--Wang denoising process $(Y_t)_{t\in [0,1]}$ for $\mu_{\mathrm{Ising}}$ is entropically stable with rate $C$ at time $t$ if the same holds for $(Y_t')_{t\in [0,1]}$, where $(Y_t')_{t\in [0,1]}$ is the vertex-field denoising process (as in \Cref{def:field-dynamics-denoising}) for $\mu_E$.

Therefore, to prove the entropic stability bounds in \Cref{lem:EI-SW} for the Swendsen--Wang process, it suffices to bound $\SStab{(1-t)\ast \mu_E^T}$ for all $T \subseteq E$, where $\mu_E^T$ denotes the distribution of $S \sim \mu_E$ conditioned on $T \subseteq S$. Indeed, by \Cref{lem:entropic-stable-max-ent-form},
    \begin{align} \label{eq:ES-target}
      C(t) \leq \max_{T\subseteq E} \SStab{(1-t) * \mu^T_E}.
    \end{align}
Note that the vertex-field scaling $(1-t)\ast \mu_E$ now corresponds to modifying the edge activity from $\beta$ to $(1-t)\beta$ in the underlying Ising model, since $\mu_E$ is a distribution over $\set{0,1}^E$.
\end{remark}

\paragraph{Entropic stability via down operator.}
We first prove an entropic stability bound by analyzing how much the down operator preserves entropic stability.
\begin{lemma}[Entropic stability under down operator] \label{lem:ES-reduction-general}
Let $\mu$ be a distribution over $\set{0,1}^m$.
Let $\*\theta \in (0,1]^m$ and $D_{1\to \*\theta}$ be the down operator defined in \Cref{def:down-operator}.
Then,%
\begin{align*}
    \SStab{\mu} \leq \SStab{\mu D_{1\to \*\theta}} \cdot \max_{i\in [m]} \theta_i^{-1}.
\end{align*}
\end{lemma}
\begin{proof}
    It suffices to show that for any distribution $\nu \ll \mu$, 
    \begin{align}\label{eq:ES-down-target}
        \sum_{i=1}^m \nu(i) \log \frac{\nu(i)}{\mu(i)} - (\nu(i) - \mu(i)) \le \SStab{\mu D_{1\to \*\theta}} \cdot \max_{i\in [m]} \theta_i^{-1} \cdot D_{\mathrm{KL}}(\nu \parallel \mu).
    \end{align}

By the definition of $\SStab{\mu D_{1 \to \theta}}$, we have
    \begin{align}\label{eq:ES-tilt}
        \nonumber &\sum_{i=1}^m (\nu D_{1 \to \*\theta})(i) \log \frac{(\nu  D_{1 \to \*\theta})(i)}{(\mu  D_{1 \to \*\theta})(i)} - ((\nu  D_{1 \to \*\theta})(i) - (\mu  D_{1 \to \*\theta})(i)) \\
        \le& \,\SStab{\mu D_{1\to \*\theta}} \cdot D_{\mathrm{KL}}(\nu  D_{1 \to \*\theta} \parallel \mu  D_{1 \to \*\theta}).
    \end{align}
Then \eqref{eq:ES-down-target} follows by combining \eqref{eq:ES-tilt} with the fact that 
\[
(\mu D_{1 \to \theta})(i) = \theta_i \mu(i) \quad\text{ and }\quad (\nu D_{1 \to \theta})(i) = \theta_i \nu(i),
\]
together with the inequality $a \log \frac{a}{b} \ge a - b$ and the following data processing inequality 
\[
D_{\mathrm{KL}}(\nu  D_{1 \to \*\theta} \parallel \mu  D_{1 \to \*\theta})\le D_{\mathrm{KL}}(\nu   \parallel \mu ). \qedhere
\]
\end{proof}

\begin{proof}[Proof of the bound $\ESnoising{t}$ in \Cref{lem:EI-SW}]
    Fix $t \in [0,1-\beta^{-1})$.
    Let 
    \[p = p(t) := 1 - ((1-t)\beta)^{-1}.\]
By~\eqref{eq:ES-target}, it suffices to bound $\SStab{(1-t)\ast \mu_E^T}$ for an arbitrary $T\subseteq E$. 

Let $D_{1\to \*p}$ be the down operator where $p_e = p$ for edges $e \in E \setminus T$, and $p_e = 1$ for edges $e \in T$.
Let $\mu_{\mathrm{RC}}^{T}$ be the random-cluster distribution with parameter $\*p$ conditioning on the edges in $T$ being occupied.
    By \Cref{prop:ES-coupling}, we have
    \begin{align*}
    ((1-t)*\mu_E^T) D_{1\to \*p} = \mu_{\mathrm{RC}}^{T}.
\end{align*}
Applying \Cref{lem:ES-reduction-general}, we obtain
\begin{align*}
    \SStab{(1-t)\ast \mu_E^T}
    \le \SStab{\mu_{\mathrm{RC}}^{T}} \cdot p^{-1}
    \le p^{-1} \cdot \frac{2}{\delta^2}
    = \frac{(1-t)\beta}{(1-t)\beta - 1} \cdot \frac{2}{\delta^2},
\end{align*}
where the second inequality follows from \Cref{lem:SI-CZ23} together with the fact that $\frac{2}{\delta^2}$-coupling independence (under arbitrary external fields) implies $\frac{2}{\delta^2}$-entropic stability, i.e., $\SStab{\mu_{\mathrm{RC}}^{T}} \le \frac{2}{\delta^2}$ (see~\cite{anari2022entropic}).
\end{proof}

\paragraph{From spectral stability to entropic stability assuming marginal bounds.}
The following lemma lifts spectral stability to entropic stability under suitable marginal bounds.
\begin{lemma}\label{lem:2nd-derivative}
Let $\theta \in (0,1)$, $\eta \ge 1$, and $K \ge 1$.
Let $\mu$ be a distribution over $\{0,1\}^m$. 
Suppose that for every $t \in [0,\theta]$ and every feasible pinning $\tau = \mathbf{1}_T$, where $T\subseteq[m]$, the distribution $\nu := (1-t) * \mu^\tau$ satisfies
 \begin{itemize}
    \item \textbf{marginal bound:} $\nu(i) \ge \frac{1}{K}$ for all $i\in[m]$;
    \item \textbf{spectral stability:} $\lambda_{\max}(\mathrm{Cor}_\nu) \le \eta$.
 \end{itemize} 
 Let
 $\alpha := (K+1)(\eta-1) + 1$.
Then $\mu$ satisfies
     \begin{align*}
    \SStab{\mu}\le\frac{\alpha}{1-(1-\theta)^{\alpha}}.
    \end{align*}
Moreover, the vertex-field denoising process $(Y_t)_{t \in [0,1]}$  for $\mu$ satisfies $R$-approximate conservation of entropy up to time $\theta$, where
    \begin{align*}
      R = 1 + \frac{\max_\tau \SStab{(1-\theta) * \mu^\tau}}{\alpha(1-\theta)^{\alpha}}. 
    \end{align*}
\end{lemma}
The proof of \Cref{lem:2nd-derivative} is deferred to \Cref{sec:appendix-B}.

Let $G = (V,E)$ be a graph with maximum degree $\Delta$, and let $\mu$ be the Gibbs distribution of the Ising model on $G$ with edge activity $\beta$ and vertex activity $\lambda$.

Let $\Lambda \subseteq V$, $S \subseteq E$, and $\tau \in \{0,1\}^\Lambda$. Define the event
\begin{align*}
  \mathcal{E} := \set{\sigma\in\set{0,1}^V \mid \sigma_\Lambda = \tau \land \forall e = \set{u,v} \in S, \sigma_u = \sigma_v}.
\end{align*}
Without loss of generality, assume $\mu(\mathcal{E}) > 0$, and define the conditional distribution
\begin{align*}
  \mu^{S,\tau} := \mu(\cdot \mid \mathcal{E}).
\end{align*}
When $\Lambda = \emptyset$, we simply write $\mu^{S,\tau} = \mu^S$.

Given $\delta \in (0,1)$ and $\Delta \ge 1$, consider the following regime where $\beta$ is close to $1$:
\begin{align} \label{eq:cond-beta-lambda-Dobrushin}
  1 - \frac{\delta^2}{\Delta^2} \le \beta \le 1 + \frac{\delta^2}{\Delta^2}
  \quad \text{and} \quad
  \lambda \le 1-\delta.
\end{align}
Note that this regime includes both the ferromagnetic and antiferromagnetic cases.
Within this regime, the following marginal and spectral stability bounds can be proved.
\begin{proposition}[Marginal bounds]\label{prop:marginal-bound}
Let $\delta \in (0,1)$ and $\Delta \ge 1$, and assume~\eqref{eq:cond-beta-lambda-Dobrushin}.
For any $S \subseteq E$, $\tau \in \{0,1\}^\Lambda$, and any vertex $u \in V$ whose spin is not fixed by $(S,\tau)$,
    \begin{align*}
        \frac{\mu^{S,\tau}(u)}{\mu^{S,\tau}(\overline{u})}
    \le (1+\delta)^{-k_u},
    \end{align*}
where $k_u$ denotes the size of the connected component of $u$ in the graph $(V,S)$.
\end{proposition}

\begin{proposition}[Spectral stability] \label{lem:SI-SW-Dobrushin}
Let $\delta \in (0,0.01)$ and $\Delta \ge 1$, and assume~\eqref{eq:cond-beta-lambda-Dobrushin}.
Then for any $S\subseteq E$, the Swendsen--Wang correlation matrix satisfies
 \begin{align*}
    \lambda_{\max}\left(\swcor{\mu^S}\right) \le \frac{64\Delta}{\delta^3}.
 \end{align*}
\end{proposition}

Assuming \Cref{prop:marginal-bound,lem:SI-SW-Dobrushin}, we now bound $\EScrude{t}$ in \Cref{lem:EI-SW}.
\begin{proof}[Proof of the bound $\EScrude{t}$ in \Cref{lem:EI-SW}]
Fix $t$ such that $1 \le (1-t)\beta \le 1 + \frac{\delta^2}{\Delta^2}$, and fix $T \subseteq E$.
By~\eqref{eq:ES-target}, it suffices to consider the distribution $\pi := (1-t) * \mu_E^T$. 

Set $\theta := \frac{\delta^2}{\Delta^2}$.
For any $s \in [0,\theta]$ and any $\Lambda \subseteq E$, define
  \begin{align*}
    \nu := (1 - s) * \pi^\Lambda = (1-s)(1-t) * \mu_E^{\Lambda \cup T}.
  \end{align*}
Let $S := \Lambda \cup T$. By construction, $\nu$ can be obtained by first sampling $X$ from the Gibbs distribution $\tilde{\mu}$ of an Ising model with parameters $(1-s)(1-t)\beta$ and $\lambda$, conditioned on edges in $S$ being monochromatic, and then returning the set of monochromatic edges; that is, $\nu = (\tilde{\mu}^S)_E$.
Note that the parameters $(1-s)(1-t)\beta$ and $\lambda$ satisfy~\eqref{eq:cond-beta-lambda-Dobrushin}.
We claim:
  \begin{itemize}
  \item $\min_e \nu(e) \geq 1/4$;
  \item $\lambda_{\max}(\cor{\nu}) \leq \frac{64\Delta}{\delta^3}$.
  \end{itemize}
Assuming these claims, \Cref{lem:2nd-derivative} together with~\eqref{eq:ES-target} implies
  \begin{align*}
    \EScrude{t} \leq \frac{320\Delta/\delta^3}{1 - \tp{1 - \frac{\delta^2}{\Delta^2}}^{320\Delta/\delta^3}} \leq 321 \frac{\Delta^2}{\delta^4},
  \end{align*}
where we used the inequality $1 - (1-x)^k \ge 1 - \frac{1}{1+kx}$ for $x \in [0,1]$ and $k \ge 1$.

It remains to prove the two claims. 
The bound $\lambda_{\max}(\cor{\nu}) \le 64\Delta/\delta^3$ follows directly from \Cref{lem:SI-SW-Dobrushin}. 
It remains to show that $\min_{e} \nu(e) \ge 1/4$.

Fix an edge $e = \{u,v\}$. By \Cref{prop:marginal-bound}, we have
  \begin{align*}
    \nu(e) \geq \mu^S(u=0,v=0)
    &= \mu^{S,v=0}(\ol{u})\mu^{S}(\ol{v}) 
    \geq \frac{1}{1+(1+\delta)^{-k_u}} \frac{1}{1+(1+\delta)^{-k_v}} \geq \frac{1}{4},
  \end{align*}
where $k_u$ and $k_v$ denote the sizes of the connected components of $u$ and $v$ in the graph $(V,S)$, respectively. The last inequality uses $\delta \in (0,1)$ and $k_u, k_v \ge 1$.
\end{proof}

In the rest of this subsection, we prove \Cref{prop:marginal-bound,lem:SI-SW-Dobrushin}.
\begin{proof}[Proof of \Cref{prop:marginal-bound}]
We first prove the claimed marginal bound in the antiferromagnetic regime $1 - \frac{\delta^2}{\Delta^2} \le \beta \le 1$.  
The general case can be handled by considering $\ol{\beta}=\min(\beta,\beta^{-1})\le 1$.

Let $C_u$ be the connected component containing $u$ in the graph $(V,S)$, and let $\partial C_u$ be the edge boundary of $C_u$, i.e., the set of edges with exactly one endpoint in $C_u$.
Since all vertices in $C_u$ are conditioned to have the same spin, the ratio $\frac{\mu^{S,\tau}(u)}{\mu^{S,\tau}(\overline{u})}$ is maximized by assigning worst-case spins to vertices outside $C_u$. This gives
    \begin{align*}
         \frac{\mu^{S,\tau}(u)}{\mu^{S,\tau}(\overline{u})}
         \le \lambda^{k_u}\beta^{-|\partial C_u|}
         \le \lambda^{k_u}\beta^{-\Delta k_u},
    \end{align*}
where we used $|\partial C_u| \le \Delta k_u$.

Using $\lambda \le 1-\delta$ and $\beta \ge 1-\frac{\delta^2}{\Delta^2}$, we obtain
    \begin{align*}
        \frac{\mu^{S,\tau}(u)}{\mu^{S,\tau}(\overline{u})}
        \le \frac{(1-\delta)^{k_u}}{\bigl(1-\frac{\delta^2}{\Delta^2}\bigr)^{\Delta k_u}}
        \le \frac{(1-\delta)^{k_u}}{(1-\delta^2)^{k_u}}
        = (1+\delta)^{-k_u},
    \end{align*}
where we used $\bigl(1-\frac{\delta^2}{\Delta^2}\bigr)^\Delta \ge 1-\delta^2$ in the second inequality.
\end{proof}

It remains to prove \Cref{lem:SI-SW-Dobrushin}.
Recall the definition of the influence matrix and total influence in \eqref{eq:influence-matrix} and \Cref{def:total-influence}.
We establish the following bound.
\begin{lemma}[Bounded total influence]\label{lem:SW-TI}
Fix $\delta \in (0,0.01)$, $\Delta \ge 1$, $\beta < 1$, and assume \eqref{eq:cond-beta-lambda-Dobrushin}.
Then there exists a constant $C_\delta = 16/\delta$ such that, for any vertex $u$,
    \begin{align*}
        \-{TI}^{\mu^{S,\tau}}_{u} = 1 + \sum_{v \in V} \abs{\Psi_{\mu^{S,\tau}}(u,v)} \le C_\delta \cdot (1+\delta)^{k_u/4}.
    \end{align*}
    Moreover, there exists a coupling $\+C$ between $\mu^{S, \tau, u \gets 0}$ and $\mu^{S, \tau, u \gets 1}$ such that
    \begin{align*}
        \E[(X,Y) \sim \+C]{\mathrm{dist}(X,Y)} \le C_\delta \cdot (1+\delta)^{k_u/4}.
    \end{align*}
\end{lemma}
\begin{proof}
Let $C_\delta$ be specified later. 
We first interpret $\mu^{S,\tau}$ as an Ising Gibbs distribution on a contracted graph.

Let $C_1,\ldots,C_m$ be the connected components in the graph $(V,S)$, and let $G' = ([m],E')$ be the induced multigraph, where the edge multiplicity between $i\neq j$ equals the number of edges between $C_i$ and $C_j$ in $G$.
Define $\nu$ as the Gibbs distribution on $G'$ with edge activity $\beta$ and vertex activities $\lambda^{|C_i|}$, conditioned on pinned components induced by $\tau$.

A sample from $\mu^{S,\tau}$ can be generated by sampling $Y\sim \nu$ and assigning spin $Y_i$ to all vertices in $C_i$. 
Without loss of generality, assume $u\in C_1$. 
    Then,
    \begin{align*}
        1 + \sum_{v \in V} \abs{\Psi_{\mu^{S,\tau}}(u,v)}
        = k_u + \sum_{i \in [m]} \abs{C_i} \cdot \abs{\Psi_{\nu}(1,i)}.
    \end{align*}

By the self-avoiding walk (SAW) tree reduction stated in \Cref{lem:graph-to-tree-comparison}, we may assume $G'$ is a tree (after merging parallel edges), and prove the claim by induction over its size.

Let vertex $1$ have children $u_1,\ldots,u_d$. For each $i$, let $T_i$ be the subtree rooted at $u_i$, let $r_i := |C_{u_i}|$, and let $p_i$ be the number of edges between $C_1$ and $C_{u_i}$. Then $\sum_i p_i \le \Delta k_u$.

Let $\nu_i$ denote the Gibbs distribution on $T_i$. By tree factorization of influence,
\begin{align*}
\sum_{j \in T_i} |C_j|\,|\Psi_\nu(1,j)|
=
|\Psi_\nu(1,u_i)|
\left(
r_i + \sum_{j \in T_i} |C_j|\,|\Psi_{\nu_i}(u_i,j)|
\right).
\end{align*}

Hence,
\begin{align*}
1 + \sum_{v \in V} |\Psi_{\mu^{S,\tau}}(u,v)|
\le
k_u + \sum_{i=1}^d |\Psi_\nu(1,u_i)|
\left(
r_i + \sum_{j \in T_i} |C_j|\,|\Psi_{\nu_i}(u_i,j)|
\right).
\end{align*}

Applying the induction hypothesis on each subtree gives
\begin{align*}
1 + \sum_{v \in V} |\Psi_{\mu^{S,\tau}}(u,v)|
\le
k_u + C_\delta \sum_{i=1}^d (1+\delta)^{r_i/4} |\Psi_\nu(1,u_i)|.
\end{align*}

It remains to bound $|\Psi_\nu(1,u_i)|$. 
Let $R$ be the marginal ratio at $u_i$ in $T_i$. 
Then
    \begin{align*}
        \abs{\Psi_\nu(1,u_i)}
        &= \frac{(1-\beta^{2p_i})R}{(\beta^{p_i}R+1)(R+\beta^{p_i})}\\
        &\le \min\set{\frac{1-\beta^{p_i}}{1+\beta^{p_i}}, \frac{R}{\beta^{p_i}} }\\
        &\le \min\set{\frac{\delta^2 p_i}{\Delta^2}, \beta^{-p_i}(1+\delta)^{-r_i}}.
    \end{align*}

Using the inequality $\min\{a,b\}\le a^{3/4}b^{1/4}$, we have
    \begin{align*}
        (1+\delta)^{r_i/4}\abs{\Psi_\nu(1,u_i)}
        &\le \tp{\frac{\delta^2 p_i}{\Delta^2}}^{3/4}\beta^{-p_i/4} \\
        &\le \tp{\frac{\delta^2 p_i}{\Delta^2}}^{3/4}\tp{1-\frac{\delta^2}{\Delta^2}}^{-k_u \Delta/4} \\
        &\le \tp{\frac{\delta^2 p_i}{\Delta^2}}^{3/4} \tp{1+\delta/2}^{k_u/4}.
    \end{align*}

    Therefore, it suffices to show that
    \begin{align*}
        k_u + C_\delta \cdot (1+\delta/2)^{k_u/4}\cdot \sum_{i=1}^d \tp{\frac{\delta^2 p_i}{\Delta^2}}^{3/4} \le C_\delta \cdot (1+\delta)^{k_u/4}.
    \end{align*}

Since $\sum_i p_i \le \Delta k_u$, we have
\begin{align*}
\sum_{i=1}^d \left(\frac{\delta^2 p_i}{\Delta^2}\right)^{3/4}
\le \frac{\delta^{3/2}}{\Delta^{3/2}} \sum_i p_i
\le \frac{\delta}{6} k_u,
\end{align*}
and thus it suffices to verify
\begin{align}\label{eq:4572}
    k_u + k_u \cdot \frac{1}{6}C_\delta \delta \cdot(1+\delta/2)^{k_u/4} \le C_\delta \cdot (1+\delta)^{k_u/4}.
\end{align}

To verify \eqref{eq:4572}, it suffices to show there exists a constant $\alpha > 0$ such that:
\begin{align}
\left( \frac{1+\delta}{1+\delta/2} \right)^{k_u/4} &\ge \frac{\delta k_u}{6} + \alpha \label{eq:part1} \\
\frac{k_u}{\alpha (1+\delta/2)^{k_u/4}} &\le C_\delta \label{eq:part2}
\end{align}

For $\delta \in (0, 0.01)$, we have $\frac{1+\delta}{1+\delta/2} \ge 1 + \delta/3$. 
By convexity, $\e^x \ge \e^c(x - c + 1)$ for any $c \in \mathbb{R}$. 
Substituting $x = \frac{k_u}{4}\log(1+\delta/3)$ and $c=\log \left( \frac{2\delta}{3\log(1+\delta/3)} \right)$,
we then have
\begin{align*}
\left( \frac{1+\delta}{1+\delta/2} \right)^{k_u/4}
&\ge (1+\delta/3)^{k_u/4} 
= \e^{\frac{k_u}{4}\log(1+\delta/3)} \\
&\ge \e^c \left( \frac{k_u}{4}\log(1+\delta/3) - c + 1 \right) \\
&= \frac{\delta k_u}{6}+\e^c(1-c).
\end{align*}
Therefore, \eqref{eq:part1} holds.

Since $\frac{\delta}{3} - \frac{\delta^2}{18} \le \log(1+\delta/3) \le \frac{\delta}{3}$, we have $c=\log \left( \frac{2\delta}{3\log(1+\delta/3)} \right)\in[ \log 2, \log \frac{2}{1-1/600} ] $. 
Consequently, by choosing $\alpha = e^c(1-c)\ge 0.5$ and substituting $C_\delta = 16/\delta$, \eqref{eq:part2} holds if
\begin{align*}
\frac{k_u \delta}{8 (1+\delta/2)^{k_u/4}} \le 1
\end{align*}
Let $f(k) = \frac{k \delta}{(1+\delta/2)^{k/4}}$. The maximum of $f(k)$ is attained at $k^* = \frac{4}{\log(1+\delta/2)}$, where for $\delta \le 0.01$,
\begin{align*}
f(k^*) = \frac{4 \delta}{\e \log(1+\delta/2)} \le 8.
\end{align*}
Thus \eqref{eq:part2} holds. Combining \eqref{eq:part1} and \eqref{eq:part2} proves \eqref{eq:4572}.

This completes the bound on total influence. 
The coupling statement follows from \Cref{lem:coupling-to-tree-comparison}.
\end{proof}

We are now ready to prove \Cref{lem:SI-SW-Dobrushin}.

\begin{proof}[Proof of \Cref{lem:SI-SW-Dobrushin}]
When $\beta \ge 1$, the bound $\lambda_{\max}(\swcor{\mu^S}) \le \frac{6\Delta}{\e\delta^3}$ already follows from~\Cref{lem:SI-SW}. 
It therefore remains to consider the case $\beta < 1$.

Recall the event indicators $Z:\{0,1\}^V \to \{0,1\}^E$ defined in \Cref{sec:generalized-field-dyanmics}.
In the case of the Swendsen--Wang process, it is defined by
\[
  Z(X) := \{e = \{u,v\} \mid X_u = X_v\}.
\]
By the definition of $\swcor{\mu^S}$, the maximum eigenvalue satisfies
    \begin{align}\label{eq:bound-1}
      \lambda_{\max}(\swcor{\mu^S})
      &\leq \max_{e \not\in S} \inf_{\xi = \xi(e)} \E[(X,Y)\sim \xi]{\dist(Z(X), Z(Y))},
    \end{align}
where the infimum is taken over all couplings $\xi$ between $\mu^S$ and $\mu^{S \cup \{e\}}$.

Fix an edge $e = \{u,v\}$. We construct a coupling $(X,Y)\sim \xi$ as follows:
    \begin{enumerate}
    \item Sample $X_{u,v} \sim \mu^S_{u,v}$ and $Y_{u,v} \sim \mu^{S,\,u=v}_{u,v}$ independently.
    \item Complete $X$ and $Y$ using the coupling $\mathcal{C}$ between $\mu^{S,X_{u,v}}$ and $\mu^{S,Y_{u,v}}$ from \Cref{lem:SW-TI}. In the case $\dist(X_{u,v}, Y_{u,v}) = 2$, we apply a path coupling construction.
    \end{enumerate}
It is straightforward to verify that $\xi$ is a valid coupling between $\mu^S$ and $\mu^{S \cup \{e\}}$.

We now bound the expected discrepancy. First, note that
\[
  \dist(Z(X), Z(Y)) \le \Delta \cdot \dist(X,Y),
\]
since each vertex disagreement can affect at most $\Delta$ edges.

By \Cref{lem:SW-TI} and the above construction, we obtain
\begin{align*}
  &\E[(X,Y)\sim \xi]{\dist(X, Y)}\\
  \le 
  &C_\delta \tp{
    \sum_{w \in \{u,v\}} \Pr{X_w = 1} (1+\delta)^{k_w/4}
    + \Pr{Y_u = Y_v = 1} \tp{(1+\delta)^{k_u/4} + (1+\delta)^{k_v/4}}
  },
\end{align*}
where $C_\delta = 16/\delta$.

Using \Cref{prop:marginal-bound}, we have for $w \in \{u,v\}$,
\begin{align*}
  \Pr{X_w = 1} &\le (1+\delta)^{-k_w},\\
  \Pr{Y_u = Y_v = 1} = \Pr{Y_v = 1} &\le (1+\delta)^{-k_u - k_v}.
\end{align*}
Substituting these bounds gives
\[
  \E[(X,Y)\sim \xi]{\dist(X, Y)} \le 4 C_\delta.
\]
Therefore,
\[
  \E[(X,Y)\sim \xi]{\dist(Z(X), Z(Y))} \le 4\Delta C_\delta = \frac{64\Delta}{\delta}\le \frac{64\Delta}{\delta^3}.
\]
Combining this with~\eqref{eq:bound-1} completes the proof.
\end{proof}

\subsection{Entropic stability from spectral stability}\label{sec:appendix-B}
In this subsection, we prove \Cref{lem:2nd-derivative}, which establishes entropic  stability from spectral stability under suitable marginal-bound assumptions.

Let $\mu$ be a distribution over $\Omega \subseteq \{0,1\}^n$, and fix a function $f \in \mathbb{R}^\Omega_{\ge 0}$.
Let $(Y_t)_{t \in [0,1]}$ be the vertex-field denoising process for $\mu$.
Define the entropy at time $\theta$ as 
\[E(\theta) = \E[]{\Ent[]{f(Y_1) \mid Y_\theta}}.\]
By \Cref{def:phi-ent-stable} and \Cref{lem:entropic-stable-max-ent-form} (see also \cite{chen2025rapid}), the derivative of $E$ at $\theta=0$ satisfies
\begin{align*}
    E'(0)=-\sum_{i=1}^n \mu(i)\tp{\mu(f\mid i)\log \frac{\mu(f\mid i)}{\mu(f)}-\mu(f\mid i)+\mu(f)},
\end{align*}
where $\mu(f) = \E[\mu]{f}$ and $\mu(f \mid i) = \E[\mu^{i \gets 1}]{f}$.
Thus, entropic stability at $\theta=0$ can be expressed as
\begin{align*}
    -E'(0) \le \SStab{\mu} \cdot E(0). 
\end{align*}

We use the second derivative of $E(\theta)$ to control entropic stability.
\begin{lemma} \label{lem:differential-equation}
Under the assumptions of \Cref{lem:2nd-derivative}, we have
  \begin{align}\label{eq:target}
    \forall t\in [0,\theta], \quad (1-t) E''(t) \le (K+1)(\eta -1) \cdot (-E'(t)).
  \end{align}
\end{lemma}

Using \Cref{lem:differential-equation}, we now prove \Cref{lem:2nd-derivative}.
\begin{proof}[Proof of~\Cref{lem:2nd-derivative}]
  Define
    $\alpha := (K+1)(\eta - 1) + 1 \geq 1$.
Note that \eqref{eq:target} implies that the quantity $-\frac{E'(t)}{(1-t)^{\alpha-1}}$ is nondecreasing on $[0,\theta]$, since
  \begin{align*}
    \frac{\mathrm{d}}{\mathrm{d} t} \tp{\frac{-E'(t)}{(1-t)^{\alpha-1}}} = (1-t)^{-\alpha} \left((-E'(t)) (\alpha-1) - (1-t) E''(t) )\right) \geq 0,
  \end{align*}
Therefore, for all $t\in[0,\theta]$,
  \begin{align*}
    -E'(0) &\leq \frac{-E'(t)}{(1-t)^{\alpha-1}} \leq \frac{-E'(\theta)}{(1-\theta)^{\alpha-1}}.
  \end{align*}
This proves the first part of \Cref{lem:2nd-derivative} by integrating:
\begin{align*}
E(0)\ge E(0)-E(\theta)
= \int_0^{\theta} (-E'(t))\,\-dt 
&\ge (-E'(0)) \int_0^{\theta} (1-t)^{\alpha-1}\,\-dt \\
&= \frac{1-(1-\theta)^\alpha}{\alpha}\cdot (-E'(0)).
\end{align*}

For the second part of the lemma, we use the same monotonicity in the opposite direction:
  \begin{align*}
    E(0) - E(\theta)
    &= \int_0^{\theta} (-E'(t)) \-d t
      \leq \frac{-E'(\theta)}{(1-\theta)^{\alpha-1}} \int_0^{\theta} (1-t)^{\alpha-1} \-d t \\
    &= \frac{-E'(\theta)}{(1-\theta)^{\alpha-1}} \cdot \frac{1 - (1-\theta)^{\alpha}}{\alpha}
      \leq \frac{-E'(\theta)}{(1-\theta)^{\alpha-1}} \cdot \frac{1}{\alpha} \\
    &\leq \frac{\max_\tau \SStab{(1-\theta) * \mu^\tau}}{\alpha(1-\theta)^{\alpha}} \cdot E(\theta). \qedhere
  \end{align*}
\end{proof}

Now it remains to prove \Cref{lem:differential-equation}.
We first state a key estimate for $E''(t)$.
\begin{proposition}\label{prop:brute-force}
The second derivative of $E(\theta)$ at $\theta = 0$ satisfies
\begin{align}\label{eq:brute-force}
    E''(0)\le \frac{1}{\mu(f)}\sum_{i,j=1}^n \Big(\mu(ij)\*1[i\ne j]-\mu(i)\mu(j)\Big)(\mu(f|i)-\mu(f))(\mu(f|j)-\mu(f)).
\end{align}
\end{proposition}

With this proposition in hand, we can prove \Cref{lem:differential-equation}.
\begin{proof}[Proof of \Cref{lem:differential-equation}]    
  We first focus on the case $t = 0$.
Let $\mu(f):=\E[\mu]{f}$ denote the expectation of $f$, and define the vector $\*f\in \mathbb{R}^n$  by $\*f_i=\frac{\mu(f|i)}{\mu(f)}-1$.
Let $\+D_\mu$ denote the diagonal matrix of marginals. Then, by \Cref{prop:brute-force},
    \begin{align*}
    E''(0)\le &\, \mu(f)\cdot \*f^\intercal(\-{Cov}_\mu-\+D_\mu)\*f \tag{\Cref{prop:brute-force}}\\
    \le&\, \mu(f)(\eta-1)\cdot \*f^\intercal\+D_\mu\*f \\
    =&\, \mu(f)(\eta-1)\sum_{i=1}^n \mu(i)\tp{\frac{\mu(f \mid i)}{\mu(f)} - 1}^2 \\
    \overset{(\star)}{\le}&\, (K+1)(\eta-1)\sum_{i=1}^n \mu(i) \tp{\mu(f \mid i) \log \frac{\mu(f \mid i)}{\mu(f)} - \tp{\mu(f \mid i) - \mu(f)}}\\
    = &\, (K+1)(\eta - 1) \cdot (-E'(0)).
    \end{align*}
Here $(\star)$ uses $\mu(f\mid i)\le K\mu(f)$ and the inequality $(x-1)^2 \le (K+1)(x \log x - x + 1)$ for all $0 \le x \le K$.
To see this, note that $\log x \ge \frac{2(x-1)}{x+1}$ for $x > 0$,
which implies
    \begin{align*}
        x \log x - x + 1
        \ge \frac{2x(x-1)}{x+1} - x + 1
        = \frac{2x(x-1) - x(x+1) + (x+1)}{x+1}
        = \frac{(x-1)^2}{x+1}.
    \end{align*}
Hence $(x-1)^2 \le (x+1)(x\log x - x + 1)\le (K+1)(x\log x - x + 1)$ for $0< x \le K$.

We now extend the argument to general $t\in(0,1)$ by conditioning.
Recall that
    \begin{align*}
      E(t) = \E{\Ent{f(Y_1) \mid Y_t}}
    \end{align*}
Fix $Y_t$ and define a shifted denoising process $(Y'_\theta)_{\theta\in[0,1]}$ by
    \begin{align*}
      Y_\theta' := Y_{t + \theta(1-t)} - Y_t.
    \end{align*}
Let $f_{Y_t}$ be defined by $f_{Y_t}(Y_1-Y_t)=f(Y_1)$, and define
    \begin{align*}
      E_{Y_t}(\theta)
      := \E{\Ent{f_{Y_t}(Y'_1) \mid Y_\theta'} \mid Y_t}
      = \E{\Ent{f(Y_1) \mid Y_{t + \theta(1-t)}} \mid Y_t}.
    \end{align*}
Taking expectation over $Y_t$ gives
    \begin{align*}
      \E{E_{Y_t}(\theta)}  &= \E{\Ent{f(Y_1) \mid Y_{t + \theta(1-t)}} } = E\tp{t + \theta(1-t)}.
    \end{align*}
Differentiating in $\theta$, we obtain
    \begin{align*}
      (1-t) E'\tp{t + \theta(1-t)} &= \E{E_{Y_t}'(\theta)}, \\
      (1-t)^2 E''\tp{t + \theta(1-t)} &= \E{E_{Y_t}''(\theta)}.
    \end{align*}
Setting $\theta=0$ reduces the general case to $t=0$, completing the proof.
\end{proof}

It now remains to prove~\Cref{prop:brute-force}.

\begin{proof}[Proof of \Cref{prop:brute-force}]
Define the  distribution $\nu$ by $\nu/\mu=f/\mu(f)$. Let $\nu_t = \nu D_{1\to t}$ and $\mu_t = \mu D_{1\to t}$. 
Then
\begin{align*}
    E(t) 
    &= \Ent{f(Y_1)} - \Ent{\E{f(Y_1) \mid Y_t}} \\
    &= \mu(f) \tp{\DKL{\nu}{\mu} - \DKL{\nu_t}{\mu_t}}.
\end{align*}
This implies that
\begin{align*}
    E''(0) = -\mu(f) \frac{d^2}{d t^2} \left.  \DKL{\nu_t}{\mu_t} \right\vert_{t = 0}.
\end{align*}
By a direct expansion, for any $S \subseteq [n]$ we have
\begin{align}\label{eq:prob-S}
    \mu_t(S) =
    \begin{cases}
        \mu(i) \cdot t - \tp{\sum_{j \neq i}\mu(ij)} \cdot t^2 + o(t^2), & S = \{i\},\\
        \mu(ij) \cdot t^2 + o(t^2), & S = \{i,j\},\\
        o(t^2), & \abs{S} \ge 3,\\
        1 - (\sum_{1 \le i \le n} \mu(i)) t + \tp{\sum_{1 \le i<j \le n} \mu(ij)} t^2 +o(t^2), & S = \emptyset,
    \end{cases}
\end{align}
where $\mu(i)$ is the marginal of $i$ and $\mu(ij)$ is the joint marginal of ${i,j}$. 
The same expansion holds for $\nu_t$.
Hence it suffices to expand $\DKL{\nu_t}{\mu_t}$ over sets of size at most $2$:
\begin{align*}
    \DKL{\nu_t}{\mu_t} &= \nu_t(\emptyset)\log\frac{\nu_t(\emptyset)}{\mu_t(\emptyset)} + \sum_i \nu_t(\set{i})\log \frac{\nu_t(\set{i})}{\mu_t(\set{i})} \\
    &\quad + \sum_{i < j} \nu_t(\set{i,j}) \log \frac{\nu_t(\set{i,j})}{\mu_t(\set{i,j})} + o(t^2). 
\end{align*}
We analyze the three contributions separately.

\paragraph{Empty set.}
By \eqref{eq:prob-S}, write 
$\mu_t(\emptyset) = 1 + g t + h t^2 + o(t^2)$ and $\nu_t(\emptyset) = 1 + b t + c t^2 + o(t^2)$.
Using the Taylor expansion of $\log$, we obtain
\begin{align*}
    &\nu_t(\emptyset)\log\frac{\nu_t(\emptyset)}{\mu_t(\emptyset)} 
    = t^2 \left(\frac{1}{2} (b-g)^2+c-h\right)+t (b-g) + o(t^2) \\
    &= t^2 \tp{\frac{1}{2}\tp{\sum_i (\nu(i) - \mu(i))}^2 + \sum_{i<j} \tp{\nu(ij) - \mu(ij)}} + t \sum_i (\mu(i) - \nu(i)) + o(t^2).
\end{align*}
\paragraph{Singletons.} 
By \eqref{eq:prob-S}, write 
$\nu_t(\set{i}) = a t + b t^2 + o(t^2)$ and $\mu_t(\set{i}) = c t + d t^2 + o(t^2)$.
Then
\begin{align*}
    &\nu_t(\set{i}) \log \frac{\nu_t(\set{i})}{\mu_t(\set{i})} \\
    &= t^2 \left(b \log \left(\frac{a}{c}\right)-\frac{a d}{c}+b\right)+a t \log \left(\frac{a}{c}\right) + o(t^2) \\
    &= t^2 \tp{-\tp{\sum_{j\neq i} \nu(ij)} \tp{1 + \log \frac{\nu(i)}{\mu(i)}} + \frac{\nu(i)}{\mu(i)} \tp{\sum_{j\neq i} \mu(ij)}} + t \nu(i) \log \frac{\nu(i)}{\mu(i)} + o(t^2).
\end{align*}
Summing over $i$ gives
\begin{align*}
    &\sum_i \nu_t(\set{i}) \log \frac{\nu_t(\set{i})}{\mu_t(\set{i})} - t \nu(i) \log \frac{\nu(i)}{\mu(i)} + o(t^2)\\
    &= t^2 \tp{-\sum_{i<j} \nu(ij) \tp{2 + \log \frac{\nu(i)\nu(j)}{\mu(i)\mu(j)}} + \sum_{i<j} \mu(ij) \tp{\frac{\nu(i)}{\mu(i)} + \frac{\nu(j)}{\mu(j)} }}.
\end{align*}
\paragraph{Pairs.}
By \eqref{eq:prob-S}, for $i\neq j$, write $\nu_t(\set{i,j}) = a t^2 + o(t^2)$, and $\mu_t(\set{i,j}) = b t^2 + o(t^2)$.
Then
\begin{align*}
    \nu_t(\set{i,j}) \log \frac{\nu_t(\set{i,j})}{\mu_t(\set{i,j})} 
    &= t^2  a \log \left(\frac{a}{b}\right) + o(t^2) 
    = t^2 \nu(ij)\log \frac{\nu(ij)}{\mu(ij)} + o(t^2).
\end{align*}
\paragraph{The $t^2$-coefficient of $\DKL{\nu_t}{\mu_t}$.}
Summing up all the terms, 
the coefficient of $t^2$ in the Taylor expansion of  $\DKL{\nu_t}{\mu_t}$ equals
\begin{align*}
A_2 &=\frac{1}{2}\tp{\sum_i (\nu(i) - \mu(i))}^2 + \sum_{i<j} \tp{\nu(ij) - \mu(ij)} \\
&\quad + \sum_{i<j} \nu(ij) \tp{\log \frac{\nu(ij)}{\mu(ij)} - \log \frac{\nu(i)\nu(j)}{\mu(i)\mu(j)}} \\
&\quad + \sum_{i<j} \mu(ij) \tp{\frac{\nu(i)}{\mu(i)} + \frac{\nu(j)}{\mu(j)} } - 2\sum_{i<j} \nu(ij)\\
&\geq\frac{1}{2}\tp{\sum_i (\nu(i) - \mu(i))}^2 + \sum_{i<j} \tp{\nu(ij) - \mu(ij)} \\
&\quad+ \sum_{i<j} \mu(ij) \tp{\frac{\nu(ij)}{\mu(ij)} - \frac{\nu(i)\nu(j)}{\mu(i)\mu(j)}} \tag{$a\log\frac{a}{b} \geq a - b$ for $a, b \geq 0$}\\
&\quad + \sum_{i<j} \mu(ij) \tp{\frac{\nu(i)}{\mu(i)} + \frac{\nu(j)}{\mu(j)} } - 2\sum_{i<j} \nu(ij) \\
&=\frac{1}{2}\tp{\sum_i \mu(i) \tp{\frac{\nu(i)}{\mu(i)} - 1}}^2 - \sum_{i < j} \mu(ij)\tp{\frac{\nu(i)}{\mu(i)} - 1}\tp{\frac{\nu(j)}{\mu(j)} - 1}.
\end{align*}
Finally, since $E''(0) = -2\mu(f)A_2$ and $\mu(f|i)=\mu(f)\frac{\nu(i)}{\mu(i)}$, we obtain
\begin{align*}
E''(0) 
&\leq -\mu(f)\tp{\sum_i \mu(i) \tp{\frac{\nu(i)}{\mu(i)} - 1}}^2 + 2\mu(f)\sum_{i < j} \mu(ij)\tp{\frac{\nu(i)}{\mu(i)} - 1}\tp{\frac{\nu(j)}{\mu(j)} - 1} \\
&= \frac{1}{\mu(f)}\sum_{i,j=1}^n \Big(\mu(ij)\*1[i\ne j]-\mu(i)\mu(j)\Big)(\mu(f|i)-\mu(f))(\mu(f|j)-\mu(f)).\qedhere
\end{align*}
\end{proof}

\bibliographystyle{alpha}
\bibliography{refs}

\appendix

\section{SAW-Tree Reduction for Coupling Independence}\label{sec:missing-5}
For completeness, we include the proofs of \Cref{cor:SI-general-graphs,cor:CI-general-graphs} via SAW-tree reduction. See~\cite{weitz2006counting,chen2020contraction} for more details.

For a graph $G=(V,E)$, a feasible pinning $\tau$, and a vertex $u\in V$, let $T_{\mathrm{SAW}}(G,u)$ denote the self-avoiding walk tree rooted at $u$.
This tree enumerates all self-avoiding walks $v_0, v_1, \dots, v_\ell$ starting from $r = v_0$ in $G$ that satisfy:
\begin{itemize}
\item All vertices $v_0, v_1, \dots, v_{\ell-1}$ are distinct;
\item $v_\ell$ has degree $1$ in $G$, or $v_\ell$ is a cycle-closing vertex (i.e., there exists $i < \ell$ such that $v_\ell = v_i$);
\item For each cycle-closing vertex $v_\ell$ (with $v_\ell = v_i$ for some $i < \ell$), its value is fixed as follows: it is pinned to $0$ if $v_{i+1} \succ v_{\ell-1}$, and pinned to $1$ if $v_{i+1} \prec v_{\ell-1}$ given the order $(V,\prec)$. 
\end{itemize}
Together with the original pinning $\tau$ on $G$, this induces a feasible pinning on
$T_{\mathrm{SAW}}(G,u)$, which we denote by $\tau_{\mathrm{SAW}}$. 
In the definition of the SAW tree, each vertex $v \in V$ of the graph $G$ may have multiple copies in $T_{\text{SAW}}(G, r)$. A copy is called \emph{free} if its value is not fixed. We use only the fact that
\begin{fact} \label{fact:SAW-tree-lambda-degree}
Let $v \in V$ be a vertex in $G$. 
Then every free copy $v'$ of $v$ in $T_{\mathrm{SAW}}$ inherits the same external field $\lambda_v$ and the same degree $\Delta_v$.
\end{fact}

The following two lemmas reveal the relationship among total influence on trees, total influence on general graphs, and coupling independence on general graphs.

\begin{lemma}[{\cite[Lemma 8]{chen2020contraction}}]\label{lem:graph-to-tree-comparison}
Let $\mu$ be a distribution of a two-spin system on graph $G=(V,E)$.
Let $u \in V$ be a fixed vertex, $\tau$ be a feasible pinning, and $T = T_{\mathrm{SAW}}(G,u)=(V_{\mathrm{SAW}},E_{\mathrm{SAW}})$ be the self-avoiding walk tree rooted at $u$ equipped with pinning $\tau_{\mathrm{SAW}}$.
Let $\nu$ be the corresponding two-spin system on tree $T$. It holds
\begin{align*}
\sum_{v\in V\setminus\{u\}} \abs{\Psi_{\mu^\tau}(u,v)} \le \sum_{v \in V_{\mathrm{SAW}} \setminus \{u\}} \abs{\Psi_{\nu^{\tau_{\mathrm{SAW}}}}(u,v)}.
\end{align*}
\end{lemma}

\Cref{cor:SI-general-graphs} directly follows from~\Cref{lem:graph-to-tree-comparison}, \Cref{fact:SAW-tree-lambda-degree}, and~\Cref{thm:CI-ctrl-main}.

For establishing coupling independence, we use the following lemma.
\begin{lemma}[{\cite[Lemma 39]{chen2024rapid}}]\label{lem:coupling-to-tree-comparison}
Let $\mu$ be a distribution of a two-spin system on graph $G=(V,E)$.
Let $u \in V$ be a fixed vertex, $\tau$ be a feasible pinning, and $T = T_{\mathrm{SAW}}(G,u)=(V_{\mathrm{SAW}},E_{\mathrm{SAW}})$ be the self-avoiding walk tree rooted at $u$ equipped with pinning $\tau_{\mathrm{SAW}}$.
Let $\nu$ be the corresponding two-spin system on tree $T$. There exists a coupling $\+C$ between $\mu^{\tau \vee (u \gets 1)}$ and $\mu^{\tau \vee (u \gets -1)}$ (if applicable) such that
\begin{align*}
    \E[(X,Y) \sim \+C]{\abs{X \oplus Y}} \le \sum_{v \in V_{\mathrm{SAW}}} \abs{\Psi_{\nu^{\tau_{\mathrm{SAW}}}}(u,v)}.
\end{align*}
\end{lemma}

\Cref{cor:CI-general-graphs} directly follows from~\Cref{lem:coupling-to-tree-comparison}, \Cref{fact:SAW-tree-lambda-degree}, and~\Cref{thm:CI-ctrl-main}.
\section{Rapid Mixing via Vertex-Tilting}\label{sec:mixing-2spin-1}

In this section, we prove \Cref{lem:mixing-2spin-1}, establishing rapid mixing of the Glauber dynamics for critical antiferromagnetic two-spin systems with small $\sqrt{\beta\gamma}$.
In this regime, the spin system behaves similarly to the critical hardcore model, and we employ the traditional vertex-field dynamics to establish rapid mixing at criticality.

Let $G=(V,E)$ be an undirected graph on $n$ vertices with maximum degree $\Delta = D+1 \ge 3$, and let $0 \le \ol{\beta} \le \beta_c := \frac{\Delta-2}{\Delta}$.
Let $\mu$ be a Gibbs distribution on $G$ with parameters $(\beta,\gamma,\lambda_c)$ satisfying $\sqrt{\beta \gamma} \le \ol{\beta}$ and \Cref{cond:critical}.

By the flipping argument in~\Cref{subsec:ctrl.2}, we may assume without loss of generality that $\lambda_c \le (\gamma/\beta)^{\frac{\Delta}{2}}$. 
In particular, the parameters satisfy \eqref{eq:cond-anti-ferro-flip} and \eqref{eq:cond-regular-flip}, that is
\begin{align}\label{eq:3202}
    &\beta,\gamma\ge 0,  \quad \beta\gamma < 1, \quad  0<\lambda_c \leq (\gamma/\beta)^{\frac{\Delta}{2}}, \quad \sqrt{\beta\gamma}\le \ol{\beta}\le \beta_c:=\frac{\Delta-2}{\Delta};
\end{align}
and $G$ is $\Delta$-regular if $\max\{\beta,\gamma\} > 1$.
The critical external field $\lambda_c$ is defined as follows. Let
\begin{align*}
    \lambda(x) := x \left(\frac{x+\gamma}{\beta x+1}\right)^D, 
    \qquad  
    \delta(x) := 1 - \frac{D(1-\beta \gamma) x}{(\beta x + 1)(x+\gamma)}.
\end{align*}
Then $\lambda_c = \lambda(x_c)$, where $x_c$ is the smaller root of the equation $\delta(x)=0$.

By the monotonicity of $\lambda(\cdot)$ and the bound $\lambda_c \le (\gamma/\beta)^{(D+1)/2}$, it follows that $x_c \le \sqrt{\gamma/\beta}$.

Recall that the exponent in \Cref{lem:mixing-2spin-1} is given by
\[
\kappa:= \sqrt{\frac{1-\ol{\beta}^2}{\beta_c^2-\ol{\beta}^2}}.
\]

We begin with the following technical lemma.
\begin{lemma}\label{lem:coefV}
    Suppose $\ol{\beta} \le \frac{\Delta-2.1}{\Delta}$. Then
    \begin{align*}
        \frac{\Delta}{\Delta-2}\le \frac{\Delta(1-\beta\gamma)x_c}{\gamma-\beta x_c^2}=\sqrt{\frac{1-\beta\gamma}{\beta_c^2-\beta\gamma}}\le \kappa\le {10}
    \end{align*}
\end{lemma}
\begin{proof}
    We first prove $\frac{\Delta(1-\beta\gamma)x_c}{\gamma-\beta x_c^2}=\sqrt{\frac{1-\beta\gamma}{\beta_c^2-\beta\gamma}}$ under the condition that $\frac{D(1-\beta\gamma)x_c}{(\beta x_c+1)(x_c+\gamma)}=1$. The condition can also be written as
    \begin{align}\label{eq:3057}
        \beta x_c^2+\gamma=\tp{D-1-\Delta \beta\gamma}x_c
    \end{align}
    With $(\gamma-\beta x_c^2)^2=(\gamma+\beta x_c^2)^2-4\beta\gamma x_c^2=(\gamma+\beta x_c^2-2\sqrt{\beta\gamma}x_c)(\gamma+\beta x_c^2+2\sqrt{\beta\gamma}x_c)$, therefore
    \begin{align*}
        \frac{\Delta^2(1-\beta\gamma)^2x_c^2}{(\gamma-\beta x_c^2)^2}=& \frac{\Delta^2(1-\beta\gamma)^2 }{(D-1-\Delta \beta \gamma-2\sqrt{\beta\gamma})(D-1-\Delta \beta \gamma+2\sqrt{\beta\gamma})} \\
        =& \frac{\Delta^2(1-\beta\gamma)^2 }{(D+1)^2(\sqrt{\beta\gamma}+1)(\beta_c-\sqrt{\beta\gamma})(1-\sqrt{\beta\gamma})(\beta_c+\sqrt{\beta\gamma})} \\
        =& \frac{1-\beta\gamma}{\beta_c^2-\beta\gamma}
    \end{align*}
    The function $\frac{1-\beta\gamma}{\beta_c^2-\beta\gamma}$ is increasing in $\beta\gamma\in (0,\ol{\beta}^2)$. Hence,
    \begin{align*}
        \frac{\Delta}{\Delta-2}=\frac{1}{\beta_c}\le \sqrt{\frac{1-\beta\gamma}{\beta_c^2-\beta\gamma}} \le \sqrt{\frac{1-\ol{\beta}^2}{\beta_c^2-\ol{\beta}^2}} \le 10.
    \end{align*}
    The last inequality follows from $\beta_c=\frac{\Delta-2}{\Delta}$, $\ol{\beta} \le \frac{\Delta -2.1}{\Delta}$ and $\Delta \ge 3$.
\end{proof}
\begin{lemma}
\label{lem:field-to-unique-slack-good}
$(\beta,\gamma,\lambda(\e^{-10} x_c))$ is $D$-unique with slack $1 - \frac{1}{\e}$.
\end{lemma}
\begin{proof}
     Recall the definition of $G$ from the proof of \Cref{lem:STD-max},
    \[G(x) := \frac{\gamma-\beta x^2}{(\beta x+1)(x+\gamma)}.\]
    Recall that $G(\cdot)$ is monotonically decreasing on $[0,+\infty)$. Hence,
    \begin{align*}
        \frac{\-d\log (1-\delta(x))}{\-d\log x}=&\frac{\gamma-\beta x^2}{(\beta x+1)(x+\gamma)}=G(x)\ge G(x_c)=\frac{\gamma-\beta x_c^2}{D(1-\beta\gamma)x_c} .
    \end{align*}
    By \Cref{lem:coefV}, we have $\frac{\gamma-\beta x_c^2}{D(1-\beta\gamma)x_c}\ge \frac{\Delta}{10D}\ge 0.1$. Hence, 
    \begin{align*}
        \log \frac{1-\delta(x_c)}{1-\delta(e^{-10}x_c)}\ge \frac{1}{10}\log \frac{x_c}{e^{-10}x_c}=1,
    \end{align*}
    implying $\delta(\e^{-10} x_c) \ge 1 -\frac{1}{\e}$, i.e. $(\beta,\gamma,\lambda(\e^{-10} x_c))$ is $D$-unique with slack $1-\frac{1}{\e}$.
\end{proof}

The following lemma establishes approximate tensorization of variance in the subcritical regime, which, by \eqref{eq:AT-implies-mixing}, implies an upper bound on the mixing time of the Glauber dynamics.

\begin{lemma}[\cite{chen2021rapid}, Lemma 8.4]\label{lem:easy-regime}
    Consider an antiferromagnetic two-spin system with parameters $(\beta,\gamma,\lambda)$ on an $n$-vertex graph $G=(V,E)$ with maximum degree $\Delta \ge 3$. 
    Suppose $(\beta,\gamma,\lambda)$ is $(\Delta_v-1)$-unique with slack $\delta$ for every $v \in V$, where $\Delta_v$ denotes the degree of $v$. 
    Then $\frac{\delta^2}{64}\ast \mu$ satisfies $\frac{8}{\delta}$-approximate tensorization of variance.
\end{lemma}

With these lemmas in hand, we are ready to prove the mixing time bound in \Cref{lem:mixing-2spin-1}.

\begin{proof}[Proof of \Cref{lem:mixing-2spin-1}]
We prove this through the vertex-field localization framework. 
As in \Cref{def:field-dynamics-denoising}, a vertex-field denoising process $(Y_t)_{t\in [0,1]}$ is defined for the Gibbs distribution $\mu$ of a two-spin system on graph $G$ with parameters $(\beta,\gamma,\lambda(x_c))$ satisfying \eqref{eq:3202}.

At time $t = 1 - \frac{\lambda(x)}{\lambda(x_c)}$, the conditional law $\-{Law}(Y_1 \mid Y_t)$ is the Gibbs distribution of an antiferromagnetic system with parameters $(\beta,\gamma,\lambda(x))$ on graph $G$ under a feasible pinning. By \Cref{lem:optimal-CI}, $\mathrm{Law}(Y_1 \mid Y_t)$ is $C(t)$-coupling independent with the trivial bound $C(t) \le n$, and
    \begin{align*}
        C(t)\le 1+\frac{\Delta(1-\delta(x))}{D\delta(x)}=1+\frac{\Delta(1-\beta\gamma)x}{(x_c-x)(\gamma x_c^{-1}-\beta x)}.
    \end{align*}
    Therefore, by $x\le x_c$ we have
    \begin{align*}
        C(t)\le& \min\set{n, 1+\frac{\Delta(1-\beta\gamma)x}{(x_c-x)(\gamma x_c^{-1}-\beta x)}} \\
        \le& 1+\min\set{n, \frac{\Delta(1-\beta\gamma)x_c}{(x_c-x)(\gamma x_c^{-1}-\beta x_c)}} \\
        \le& 1+ \sqrt{\frac{1-\beta\gamma}{\beta_c^2-\beta\gamma}}\min \set{n, \frac{x_c}{x_c-x}} \tag{\Cref{lem:coefV}}
    \end{align*}
    Fix time $t_0 = 1-\frac{\lambda(x_0)}{\lambda(x_c)}$ where $x_0$ is to be determined.
    The quantity $R = \int_0^{t} \frac{C(t)}{1-t} \-d t$ (see \Cref{lem:SS-to-AC-variance}) satisfies
    \begin{align}\label{eq:R-quantity}
    \nonumber R&=\int_{0}^{t_0}\frac{C(t)}{1-t}\-dt \\
    \nonumber \text{(change of variables)}\quad &\le \int_{x_0}^{x_c} \tp{1+ \sqrt{\frac{1-\beta\gamma}{\beta_c^2-\beta\gamma}}\min \set{n, \frac{x_c}{x_c-x}}}\cdot \frac{\-d\lambda(x)}{\lambda(x)} \\
    \text{($\frac{\lambda'(x)}{\lambda(x)} \le \frac{2-\delta(x)}{x} \le \frac{2}{x}$)}\quad &\le  2\sqrt{\frac{1-\beta\gamma}{\beta_c^2-\beta\gamma}}\int_{x_0}^{x_c} \min \set{n, \frac{x_c}{x_c-x}}\frac{\-dx}{x}+2\log \frac{x_c}{x_0}
    \end{align}
    We now choose $x_0 \ge 0$ such that $\lambda(x_0) = \frac{(1-\e^{-1})^2}{64} \lambda(\e^{-10} x_c)$. At time $t_0 = 1-\frac{\lambda(x_0)}{\lambda(x_c)}$, \Cref{lem:easy-regime} shows that the distribution $\mathrm{Law}(Y_1 \mid Y_t)$ exhibits the approximate tensorization of variance with constant $\frac{8}{1-\e^{-1}} \le 16$.
    Since
    \begin{align*}
        \frac{\-d\log\lambda(x)}{\-d\log x}=1+\frac{D(1-\beta\gamma)x}{(\beta x+1)(x+\gamma)}\ge 1,
    \end{align*}
    we have
    \begin{align}\label{eq:xc-and-x0}
        \log \frac{\e^{-10}x_c}{x_0}\le \log \frac{\lambda(\e^{-10}x_c)}{\frac{(1-\e^{-1})^2}{64}\lambda(\e^{-10}x_c)}\le \log 256.
    \end{align}
    Combining~\eqref{eq:R-quantity}, \eqref{eq:xc-and-x0}, and \Cref{lem:SS-to-AC-variance}, the distribution $\mu$ exhibits the approximate tensorization of variance with constant $C$ with 
    \begin{align*}
        C = 16 \cdot \exp\tp{R} \overset{(\star)}{=} O(1) \cdot \exp\tp{2 \kappa \log n} = O\tp{n^{2\kappa}},
    \end{align*}
    where $(\star)$ follows from
    \begin{align*}
         \int_{x_0}^{x_c} \min\set{n, \frac{x_c}{x_c-x}}\frac{\-dx}{x} 
    =\,& \int_{x_0}^{x_c(1-n^{-1})}\frac{\-dx}{x}+\int_{x_0}^{x_c(1-n^{-1})}\frac{\-dx}{x_c-x} +\int_{x_c(1-n^{-1})}^{x_c}\frac{n\-dx}{x} \\
    \le\,& \log\frac{x_c}{x_0}+\log n(1-\frac{x_0}{x_c})+\frac{n}{n-1} \\
    \le\, & \log n +O(1).
    \end{align*}
    By \eqref{eq:AT-implies-mixing}, the Glauber dynamics has mixing time $O(n^{2\kappa+2})$. This concludes the proof.
\end{proof}

\begin{remark}[Barrier regime for vertex-tilting] \label{rem:bad-dependency-field-dynamics}
    We enforce the condition $\bar{\beta}\le \frac{\Delta-2.1}{\Delta}$, 
    which prevents $\sqrt{\beta\gamma}$ from being too close to $\beta_c=\frac{\Delta-2}{\Delta}$.
    As discussed in \Cref{sec:proof-main-overview}, this limitation arises because the slack $\delta$ depends poorly on the external field $\lambda$ in this regime.

    In particular, when $\sqrt{\beta\gamma} \in \left[\frac{\Delta-2-\epsilon}{\Delta}, \frac{\Delta-2}{\Delta}\right)$, one can show that
    \begin{align*}
        \frac{\-d\delta(x)}{\-d\log \lambda(x)}\Big|_{x=x_c}=\frac{\delta'(x_c)}{(\log \lambda)'(x_c)}=\frac{x_c\delta'(x_c)}{2-\delta(x_c)}=-\frac{\gamma-\beta x_c^2}{D(1-\beta\gamma)x_c}=-\frac{\Delta}{D}\sqrt{\frac{\beta_c^2-\beta\gamma}{1-\beta\gamma}}=O(\sqrt{\eps})
    \end{align*}
        where we used $\delta(x_c) = 0$. 

    Consequently, by a Taylor expansion,  $(\beta,\gamma,(1-t)\lambda(x_c))$ is $D$-unique with slack
\begin{align*} %
\delta= O(\sqrt{\eps}\cdot t).
\end{align*}
    This proves the bound claimed in \eqref{eq:bad-reliance}, showing that even a significant multiplicative change in $\lambda$ leads to only a very small increase in the slack $\delta$, which is unfavorable for the vertex-field dynamics.

\end{remark}

\section{Optimality of the Swendsen--Wang Mixing Bound}\label{sec:C}
Fix $\beta>1$ and $\lambda\in(0,1)$, and let $G$ consist of $m=n/2$ disjoint edges.
Each monochromatic-edge indicator evolves independently as a two-state chain with stationary mean $q\in(0,1)$ and nontrivial eigenvalue $\rho=(1-\beta^{-1})\frac{2\lambda}{(1+\lambda)^2}\in(0,1)$.
Starting with all spins zero, the count $S_t$ of monochromatic edges has mean $mq+m(1-q)\rho^t$, whereas its stationary mean is $mq$; both variances are at most $m/4$.
Chebyshev's inequality at the midpoint of these means gives
$\DTV{\mathcal{L}(X_t)}{\mu}\ge 1-\frac{2}{m(1-q)^2\rho^{2t}}$,
where $X_t$ is the spin configuration.
At $t=\lfloor\log m/(4\log(1/\rho))\rfloor$, this is $1-O(m^{-1/2})$.
Thus $\Tmix=\Omega_{\beta,\lambda}(\log n)$, proving worst-case optimality of the logarithmic dependence in \Cref{thm:SW-main}.
The constant-order spectral-gap bound is also optimal, since the gap is at most $1$.

\end{document}